\documentclass[acmsmall]{acmart}

\AtBeginDocument{%
  }

\usepackage{threeparttable,booktabs}

\usepackage{blindtext}
\usepackage{xcolor}
\usepackage{enumitem}
\usepackage{booktabs} % For formal tables
\usepackage{mdframed}
\usepackage{framed}
\usepackage{multicol}
\usepackage{hyphenat}
\usepackage{fnpct}
\usepackage{algpseudocode}
\usepackage{amsmath,epsfig}
\usepackage{tikz}
\usepackage{subcaption}

\usepackage[boxed,ruled,vlined,linesnumbered,lined, noend]{algorithm2e}
\usepackage{amsthm}
\usepackage{setspace}
\usepackage{enumitem}
\usepackage{multirow}
\usepackage{hhline}
\usepackage{caption}
\usepackage{subcaption}
\usepackage{graphicx}
\usepackage{wrapfig}
\usepackage{amsmath,epsfig}
\usepackage{amsthm}
\usepackage{tabularx}

\newcolumntype{P}[1]{>{\centering\arraybackslash}p{#1}}

\newtheoremstyle{bfnote}%
  {}{}
  {\itshape}{}
  {\bfseries}{.}
  { }{\thmname{#1}\thmnumber{ #2}\thmnote{ (#3)}}
\theoremstyle{bfnote}

\newtheorem{definition}{Definition}[section]

\setcopyright{rightsretained}
\usepackage{tikz}
\usepackage{balance}

\usepackage{csquotes}
\usepackage{enumitem}
\usepackage{fnpct}
\usepackage{bbm}
\usepackage{lipsum}
\usepackage{setspace}

\makeatletter
\def\thm@space@setup{\thm@preskip=0pt
\thm@postskip=0pt}
\makeatother

\usepackage{array}
\usepackage{makecell}

\newcommand{\B}{\vspace*{-\smallskipamount}}

\newcommand{\BBB}{\vspace*{-\bigskipamount}}

\usepackage{pifont}
\newcommand{\cmark}{\ding{51}}

\definecolor{redx}{RGB}{180,0,0}
\definecolor{greenx}{RGB}{0,180,0}

\definecolor{redx}{RGB}{180,0,0}
\definecolor{greenx}{RGB}{0,180,0}
\usepackage{longtable}
\usepackage{mdframed}

\usepackage[hang,flushmargin]{footmisc}

\usepackage{natbib}
\newcommand*\wrapletters[1]{\wr@pletters#1\@nil}
\def\wr@pletters#1#2\@nil{#1\allowbreak\if&#2&\else\wr@pletters#2\@nil\fi}
\usepackage{tikz}
\usetikzlibrary{automata,matrix,shapes,arrows,positioning,chains,calc}
\usetikzlibrary{snakes}
\usetikzlibrary{arrows,scopes}
\usetikzlibrary{positioning,chains,fit,shapes,calc}
\usepackage{caption}
\usepackage{subcaption}

\newcommand{\blue}[1]{{\color{black}{#1}}}

\usepackage{threeparttable,booktabs}
\usepackage[utf8]{inputenc}
\begin{document}

\title{Hiding Access-pattern is Not Enough!
\textsc{Veil}: A Storage and Communication Efficient Volume-Hiding Algorithm}

\author{Shanshan Han}
\email{shanshan.han@uci.edu}
% \orcid{xxx}
\author{Vishal Chakraborty}
\email{vi.c@uci.edu}
\author{Michael Goodrich}
\email{goodrich@uci.edu}
\author{Sharad Mehrotra}
\email{sharad@ics.uci.edu}

\affiliation{%
  \institution{University of California, Irvine}
  % \streetaddress{P.O. Box 1212}
  \city{Irvine}
  \state{California}
  \country{USA}
  \postcode{92697}
}

\author{Shantanu Sharma}
\email{shantanu.sharma@njit.edu}

\affiliation{%
  \institution{New Jersey Institute of Technology}
  % \streetaddress{1 Th{\o}rv{\"a}ld Circle}
  \city{Newark}
  \state{New Jersey}
  \country{USA}
}

% 

%%
%% By default, the full list of authors will be used in the page
%% headers. Often, this list is too long, and will overlap
%% other information printed in the page headers. This command allows
%% the author to define a more concise list
%% of authors' names for this purpose.
\renewcommand{\shortauthors}{Han et al.}

%%
%% The abstract is a short summary of the work to be presented in the
%% article.
\begin{abstract}

This paper addresses  \emph{volume }\emph{leakage} (\textit{i}.\textit{e}., leakage of the number of records in the answer set) when processing keyword queries in encrypted  key-value (KV) datasets. Volume leakage, coupled with prior knowledge about data distribution and/or previously executed queries, can reveal both ciphertexts and current user queries.
We develop a solution to prevent volume leakage, entitled \textsc{Veil}, that partitions the dataset by randomly mapping keys to a set of equi-sized buckets.
\textsc{Veil} provides a  tunable mechanism for data owners to explore a trade-off between storage and communication overheads. 
To make buckets
indistinguishable to the adversary, \textsc{Veil}  uses a novel padding strategy that
allow buckets to overlap, reducing the need to add fake records. 
Both theoretical and experimental results show \textsc{Veil} to significantly outperform
existing state-of-the-art.

\end{abstract}

\maketitle
\section{Introduction}\label{sec:intro}

Over the past two decades, secure data outsourcing to untrusted clouds has emerged as an important research area.  Techniques to support  operations such as  keyword search~\cite{goh2003secure,curtmola2006searchable,kamara2012dynamic,cash2014dynamic,chase2010structured,wang2018multi,guo2020verifiable}, range search~\cite{popa2011cryptdb,guo2018enabling,faber2015rich,yuan2018enabling,pappas2014blind,poddar2016arx,wu2019servedb,yuan2017enckv}, join computation~\cite{Wang2021SecureYJ,Jutla2021EfficientSS,Chang2022TowardsPO,bater2016smcql,poddar2021senate}, aggregations \cite{Wang2021SecureYJ,ren2022heda}, as well as, techniques to support SQL queries~\cite{bater2016smcql,poddar2021senate} have
been developed.  Many such techniques use cryptographic primitives that allow the untrusted cloud to match queries by checking
whether the ciphertexts corresponding to the query keys are stored at the cloud
without having to decrypt the data. 
One of the first methods proposed as searchable encryption~\cite{searchable_encryption,curtmola2006searchable,kamara2012dynamic,cash2014dynamic,guo2020verifiable,wang2018multi} 
embeds a trapdoor/token (which is an encrypted query predicate) for a given query key into a random string, such that the equality/inequality of a query key can be checked against ciphertext.  
Several 
order-preserving encryption techniques~\cite{lewi2016order,bogatov2019comparative} to support range queries have subsequently been proposed. 
To find matching records  over ciphertext in sublinear time, several encrypted indexing techniques have
been proposed \cite{curtmola2006searchable, kamara_dynamic,kamara2018suppression,kamaraVLH,goh2003secure,kamara2012dynamic,cash2014dynamic}.

{\color{black}
These cryptographic techniques suffer from information leakages via access patterns and volumes (or, output-size). Access patterns refer to the identity of the returned records. Volumes refer to the number of records returned to answer queries. The impact of access pattern leakage has been extensively discussed in
literature~\cite{Islam2012AccessPD,Kellaris2016GenericAO,Garg2016TWORAMEO,Chen2018DifferentiallyPA}. Oblivious Random Access Memory (ORAM)~\cite{goldreich1987ORAM,goldreich1996software,ostrovsky1990efficient} and its improved version, known as Path-ORAM~\cite{stefanov2018path} are well-known tools to hide access patterns. However, both ORAM and Path-ORAM do not hide volume. Furthermore, they suffer from two problems: 
(\textit{i}) query inefficiency, % --- when we want to fetch a single data item/row/object using Path-ORAM, 
it requires fetching poly-logarithmic amount of data for a query, and 
(\textit{ii}) limited throughput and limited support for concurrent users due to the underlying tree structure, which can 
support 0.055MB/s throughput (see Table 1 of~\cite{chakraborti2017sqoram}) and 30 concurrent clients~\cite{chakraborti2018concuroram}, while existing DBMS, such as MySQL, offers high throughput and support for concurrent users. 
In contrast, volume-hiding techniques have not been given much attention, until recently. The seminal work by Kellaris et. al.~\cite{kellaris2016generic} and subsequent work~\cite{yao2020sok,lacharite2018improved,gui2019encrypted,Popa2020,durak2016else}
%shown that volume leakage can reveal sensitive information both about the query and the ciphertext datastored at the cloud. 
have shown that systems that hide access patterns can still be vulnerable to attacks that allow an adversary to reconstruct the database counts, %Specifically, attackers can determine the number of documents in the database that contain specific values, even though the access patterns are concealed 
\textit{e}.\textit{g}., if an adversary has prior knowledge about the number of records corresponding to a query, it can potentially deduce/narrow down the query by observing the number of records in the answer set.} %, \textit{i}.\textit{e}., the volume.}

%Many prior works  on encryptedsearch~\cite{yao2020sok,kellaris2016generic,lacharite2018improved,gui2019encrypted,Popa2020,durak2016else} have not considered the problem of  {volume leakage}.
%wherein  an adversary gains information from the number of the returned records (or, the output size), instead,  to consider such leakage as acceptable.  
% More recent works~\cite{yao2020sok,kellaris2016generic,lacharite2018improved,gui2019encrypted,Popa2020,durak2016else} have, however, 
% shown that volume leakage can reveal sensitive information both about the query and the ciphertext data
%  stored at the cloud. 
% For instance,  if an adversary  has prior knowledge about the number of records corresponding to a query key, it can potentially deduce/narrow down  the  query by observing the number of values in the answer set, \textit{i}.\textit{e}., the volume. 
%clarify -VC
% \sharad{why is this true? Why would oram be not effective?}.
%A typical leakage, known as volume leakage~\cite{Popa2020}, reveal information about query keys through the number of returned values of queries, that may reveal information about the query key. 
%Specifically, if an adversarial server has prior knowledge about the data distribution, \textit{i}.\textit{e}., the number of records corresponding to each key, it can deduce the plaintext query key when observing an equal number of values being retrieved to answer a query. Once the query is guessable, the technique used for access pattern hiding (such as ORAM~\cite{goldreich1987ORAM}) is no longer effective, and the plaintext can be easily obtained.
 Volume hiding can be achieved by ensuring that the number of records returned is always equal,  irrespective of the query being asked. Such records must 
include all the matching records, though they may also return additional records that are then filtered out
by the querier. In the context of keyword queries (\textit{e}.\textit{g}., as in key-value stores), this necessitates
that the cloud returns at least the number of results that is greater than or equal to 
the number of key-value pairs corresponding to any key, \textit{i}.\textit{e}., the maximum key size, denoted by $L_{\mathit{max}}$. 
Clearly, a technique returning less than $L_{\mathit{max}}$ records will reveal to the adversary 
that the corresponding query is not for the key corresponding to the $L_{\mathit{max}}$ values.
{\color{black}
A trivial solution to prevent volume leakage is to use Path-ORAM~\cite{stefanov2018path}. 
However, Path-ORAM will incur a huge communication cost of $O(L_{max}\times \log^2 |\mathcal{D}|)$, where $|\mathcal{D}|$ is the total number of key-value pairs (as discussed in \S1.1 of~\cite{kamaraVLH}). }

{\color{black}
To avoid such communication cost, recent volume-hiding techniques~\cite{feifei_hybridx,Moti_dprfMM,Moti21_dynamic,kamara2018suppression,kamara_dynamic,kamaraVLH} have explored alternative approaches that do not use ORAM.} As will become clear, 
such techniques
store spurious (fake) encrypted records to prevent volume leakage. These methods may return more than $L_{\mathit{max}}$ records in response to a query, or may store a typically small number of data records on the local side in plaintext, referred to as a ``\emph{stash}''. 
Note that
an optimal approach would not store any fake records, would not use a stash, and would retrieve $L_{\mathit{max}}$ records for any query.  We can, thus, compare the volume-hiding techniques based on the following metrics:

% \noindent$\bullet$ \emph{\textbf{Query amplification}} ($\boldsymbol{\mathit{QA}}$): the ratio of the number of records returned for a query and $L_{\mathit{max}}$ (optimal is 1).
\begin{itemize}[leftmargin=0.14in,nolistsep,noitemsep]
\item \emph{\textbf{Query amplification}} ($\boldsymbol{\mathit{QA}}$): the ratio of the number of records returned for a query and $L_{\mathit{max}}$ (optimal is 1).

\item \emph{\textbf{Storage amplification}} ($\boldsymbol{\mathit{SA}}$): the ratio of the number of records stored in the encrypted database
and the number of records in the plaintext dataset (optimal is 1).

\item  \emph{\textbf{Stash ratio}} ($\boldsymbol{\mathit{SR}}$): the ratio of the number of records stored locally and the number of records in the plaintext dataset (optimal is 0).
\end{itemize}

These metrics offer trade-offs, and we can easily
design schemes with zero stash size (or, $\mathit{SR}$) that are optimal for one of $\mathit{QA}$ or $\mathit{SA}$ (but not
for both).
%{\color{red} shantanu - ??[[what is the meaning of $\mathit{QA}$ (SA)]]??} 
%if we ignore the other  metric  (viz. $\mathit{SA}$ ($QA)$). 
For instance,  a  strategy that retrieves the entire 
ciphertext database for any keyword query  
%{\color{red} shantanu - ??[[but how it is optimal if u fetch entire db.... .... u wrote ``optimal for a given metric QA'']]??}
prevents volume leakage and has optimal storage overhead 
($\mathit{SA}$ is 1 since it does not store any fake 
records). But it has abysmal $\mathit{QA}$ (which could be $O(|\mathcal{D}|)$, if, for instance, $L_{\mathit{max}}$ is considered a constant and $|\mathcal{D}|$ is the size of the database).
%\footnote{Consider, a key with $L_{\mathit{max}} = {|\mathcal{D}|}/{2}$ and a set of ${|\mathcal{D}|}.{2}$ keys with one record each.}
An alternate baseline strategy is to store key-value pairs as a multimap~\cite{kamara2018suppression,kamara_dynamic,kamaraVLH,Moti_dprfMM,Moti21_dynamic,feifei_hybridx}, wherein volume leakage is prevented  by storing enough fake records in the multimap
associated with each key to  ensure the number of records
for each key is equal to $L_{\mathit{max}}$. Such a strategy will always retrieve $L_{\mathit{max}}$ records
(hence optimal $\mathit{QA}$), but may result in very poor $\mathit{SA}$. In particular, if $L_{\mathit{max}}$ is $O(|\mathcal{D}|)$, then the resulting ciphertext dataset may have $O(|\mathcal{D}|^2)$ key-value pairs!
\footnote{Consider a key with $L_{max} = {|\mathcal{D}|}/{2}$ records and the rest keys with only a single record each. The number of fake records is $({|\mathcal{D}|}/{2} - 1)^2$, which is $O(|\mathcal{D}|^2)$. }

One could possibly design schemes that are more efficient in terms of 
$\mathit{QA}$
compared to 
retrieving the entire database, and more efficient in terms of $SA$ compared
to padding each key to $L_{max}$. For example, using Path-ORAM over multi-map,
one can design an approach to obliviously retrieve $L_{max}$ records 
including all records associated with the query key, but at the cost of $\mathit{QA}$ to be $\log^2(|\mathcal{D}|)$, which is impractical. 
Instead, the existing volume-hiding techniques~\cite{kamara2018suppression,kamaraVLH,kamara_dynamic,XorMM,Moti_dprfMM} have explored significantly better strategies for better $\mathit{QA}$ and/or $\mathit{SA}$ (and some of them~\cite{kamara_dynamic,kamara2018suppression,kamaraVLH} are likely integrated into MongoDB). For instance, one of these
strategies, 
\textsf{dprfMM}~\cite{Moti_dprfMM}, uses cuckoo hash to achieve  $\mathit{QA}$ and $\mathit{SA}$ of 2, with a small stash having tight upper bounds. The most recent approach, \textsf{XorMM}~\cite{XorMM} uses an XOR filter~\cite{graf2020xor}  to store the dataset and achieves $\mathit{QA}$ of 1 and $\mathit{SA}$ of 1.23. %However, \textsf{XorMM} is not always successful in building the XOR filter. In our implementation, we found that if the algorithm fails, more than 40\% of the records may be not placed in the XOR filter, and those records have to be placed in a local stash, introducing high storage overhead. 
We will discuss these and other related strategies in \S\ref{sec: related_work}.

\medskip
\noindent
\textbf{Our contribution:} We develop a novel strategy, entitled \textsc{Veil},  that, given a key $k$, retrieves records associated with the key from the encrypted key-value store while preventing volume leakage\footnote{Implementation of \textsc{Veil}: https://github.com/han-shanshan/VEIL. }.
Our approach is based on bucketing, wherein keys are mapped to a set of buckets
that are then encrypted and stored. To retrieve the records for a given key, \textsc{Veil} retrieves
all the buckets that could contain the records of the given key.  Unlike prior approaches, \textsc{Veil} allows the database owner (or, user) to fine-tune the parameters/metrics -- viz., $\mathit{QA}$ and $\mathit{SA}$, which are input
parameters to the system. A user can set the values of $\mathit{QA}$ and $\mathit{SA}$ to be any value more than or
equal to 1, % and can thus achieve the desired value of these parameters. 
and \textsc{Veil} uses a local stash in order to guarantee the $\mathit{SA}$ and $\mathit{QA}$ that the user had desired. {\color{black}Like existing volume-hiding techniques, \textsc{Veil} also does not focus on hiding access-patterns.}

Given that~\cite{Moti_dprfMM} already achieves a $\mathit{QA}$ and $\mathit{SA}$ values of 2 with a small stash, our primary 
exploration in \textsc{Veil} is the resulting stash size when $\mathit{QA}$ and $\mathit{SA}$ are below 2, \textit{i}.\textit{e}.,  
the number of records retrieved remains below
$2L_{\mathit{max}}$ and the total number of records in the ciphertext database
are below $2|\mathcal{D}|$.
%{\color{red} shantanu - ??[[i know the meaning of 2 here.... but this is no early... can we assume that reviewer will digest and understand the meaning of 2, by simply saying 2 here????]]??}. 
We show both analytically and experimentally that even when we choose relatively small values of these parameters
(\textit{e}.\textit{g}., $\mathit{SA}= 1.2$ and $\mathit{QA}$ close to 1), \textsc{Veil} achieves a very small stash, which
experimentally is significantly smaller than that of the scheme in~\cite{Moti_dprfMM}.  Thus, \textsc{Veil} is
a significantly better strategy that guarantees the prevention of volume leakage and achieves a near-optimal value of $\mathit{QA}$ and acceptably
small values of $\mathit{SA}$ compared to the 
best-known strategy. We further optimize on \textsc{Veil} by developing a modified strategy that 
allows buckets to share common records to further reduce $\mathit{SA}$. The modified 
strategy requires a significantly less number of fake records to be added without increasing $\mathit{QA}$ or
the stash size. %\noindent
%\textbf{Our contributions.}
In summary, our contributions are:
\begin{itemize}[nolistsep,noitemsep,leftmargin=0.15in]
\item %We propose 
A flexible volume-hiding strategy that allows tuning storage overhead and query overhead to achieve a trade-off. 

\item %We propose 
A random bucketing strategy that distributes records of a key to buckets in a greedy way.

\item %We design 
Two strategies to add fake values to the created buckets, including a basic strategy that pads each created bucket to equal size, and an overlapping strategy that further reduces the number of fake records using a $d$-regular graph.

\item Experimental evaluation shows that \textsc{Veil} achieves flexible tuning of $\mathit{SA}$ and $\mathit{QA}$ and uses a small stash.  
\end{itemize}

\B
\section{Settings}
\label{sec:setting}
This section describes the problem more concretely, including the adversarial model and the security model in  \textsc{Veil}.

\subsection{Problem Definition}

We consider a key-value (KV) dataset $\mathcal{D}$ with a set of unique keys $\mathcal{K}$, where each key $k_i\in\mathcal{K}$ is associated with $|k_i|$ records. %, denoted as $v_1^{i}, v_2^{i}, \ldots, v_{|k_i|}^{i}$. 
The maximum number of records associated with any key in $\mathcal{K}$ is denoted by $L_{\mathit{max}}$, \textit{i}.\textit{e}., $L_{\mathit{max}}=\texttt{MAX}\{|k_i|\}_{k_i\in\mathcal{K}}$. 
The dataset $\mathcal{D}$ is encrypted in a ciphertext-secure manner, ensuring that no information is revealed from the ciphertexts, and subsequently outsourced to an untrusted public cloud server. Users, or the database owner, can query the encrypted dataset by sending encrypted queries for a key $k_i$ to the cloud.
In the absence of volume-hiding techniques, an adversary learns the number of records returned in response to an encrypted query, and this leakage enables the adversary to deduce the plaintext query key based on prior knowledge of the data distribution.

Our goal in this paper is to develop solutions to hide real volumes, \textit{i}.\textit{e}., the number of records associated with the query key, during query processing. Before presenting an overview of our approach, we first discuss the adversarial model.

\subsection{Adversarial Model}\label{sec:adv}

We consider a powerful adversary who knows the data distribution. That is, for a KV dataset $\mathcal{D}$ and its key set $\mathcal{K}$, the adversary is aware of each key $k_i\in \mathcal{K}$ and its volume $|k_i|$. The adversary also has full access to the ciphertext database. Let $q$ be a query. Let $\textsf{Cipher}(q)$ be the ciphertexts that $q$ touches. Consequently, on executing the
query, the adversary learns the association between  $q$ and the ciphertexts $\textsf{Cipher}(q)$. The  adversary may also
know the keywords corresponding to 
(a subset of) prior queries. 
Let $Q$ be the set of queries that have been executed in the database so far, and let $Q' \subseteq Q$. Suppose for each $q \in Q'$, 
the adversary knows the plaintext keyword associated with $q$. In the
worst case, $Q' = Q$, \textit{i}.\textit{e}., the adversary knows the query keyword
for all prior queries executed in the database.

%\textit{i}.\textit{e}.for each known query $q_{\lambda}$, the adversary knows the key that $q_{\lambda}$ was querying. In the worst case, this holds true for all prior queries. %{\color{red} shantanu - ??[[what is the meaning of the last sentence and what is lambda....is lambda related to number of keys???]]??}

The adversary's objective is to determine the query keyword. It can achieve this by deducing the real volume of the current query key or other keys by observing queries and then mapping the ciphertexts to the corresponding plaintexts based on prior knowledge and/or query execution. It can also use  known keys from previous queries to infer information about the plaintext query key. \emph{\textit{The goal of hiding volumes is to prevent the adversary from deducing the query keys and obtaining information for other keys based on the knowledge of past queries.}}

\noindent\textbf{Security Requirement (VSR). } Consider a KV dataset $\mathcal{D}$ with a key set $\mathcal{K}$ and a series of past queries to keys $\mathcal{K}_Q = \{k_1, \ldots, k_m\}$.
We assume that there are at least two keys, say $k_1, k_2\in\mathcal{K}$ 
%, denoted by $k^{\prime}$ and $k^{\prime\prime}$in $\mathcal{K}$ 
that have never been queried before, {\color{black}\textit{i}.\textit{e}., $|\mathcal{K}_Q| \leq |\mathcal{K}| - 2$, $k_1, k_2\in\mathcal{K}$, and $k_1, k_2\notin \mathcal{K}_Q$. 
Consider an adversary (based on the adversarial model discussed above) with the  knowledge of:
(\textit{i}) data distribution, %, \textit{i}.\textit{e}., for each keys $k$ in $\mathcal{K}$, the number of records  for the key $k$;
(\textit{ii}) the corresponding ciphertext records and volume for each queried key $k_i\in\mathcal{K}_Q$, %, \textit{i}.\textit{e}., $\textsf{Cipher}(k_i)$ and $|\textsf{Cipher}(k_i)|$; 
%, and which specific encrypted records that were  retrieved (\textit{i}.\textit{e}., access pattern associated with the query) denoted by $\textsf{AP}(\textsf{Cipher}(q_i))$. 
and (\textit{iii}) the plaintext key for some queried keys $k_{i} \in \mathcal{K}_Q$.
Suppose the adversary observes a new query to key $k_\alpha$, where $k_\alpha\in\mathcal{K}-\mathcal{K}_Q$, the goal of the adversary is to deduce whether $k_\alpha=k_1$ or $k_\alpha=k_2$. A technique will be volume-hiding if the following condition holds: 
\begin{equation}
\mathit{Prob}[k_{\alpha} = k_1|Adv] = \mathit{Prob}[k_{\alpha} = k_2|Adv]
\end{equation}

% {\centerline{$\mathit{Prob}[k_{\alpha} = k_1|Adv] = \mathit{Prob}[k_{\alpha} = k_2|Adv]$}
That is the probability of the adversary ($Adv$) for guessing $k_\alpha=k_1$ is identical to the probability of guessing $k_\alpha=k_2$.

%{\color{red}\textbf{??SHANTANU ---- MOTI model is didfferent..... we need to say something here?????}} -- SHARAD : NOT HERE... in the appendix.

\section{Overview of \textsc{Veil}}\label{sec:overview}
This section overviews \textsc{Veil}, a secure volume-hiding strategy for key-value stores. \textsc{Veil} partitions a KV dataset into buckets by associating records for a given key with one or more buckets. Given a query key, the buckets corresponding to the key, which may potentially store the records for the key, are retrieved. These retrieved buckets may include extra records that are not associated with the query key. Such records are filtered out to obtain the query answer.

\subsection{Notation}
To describe \textsc{Veil} formally, we define the following notations. 
Let $\mathcal{D}$ be a key-value dataset, 
$\mathcal{K}$ be the set of keys in $\mathcal{D}$, 
and ${\mathcal{B} = \{B_1, \ldots, B_n\}}$ be the set of $n$ buckets created over $\mathcal{D}$. 
\textsc{Veil} associates each 
 key $k_i\in\mathcal{K}$ with a set of $f$ buckets from $\mathcal{B}$, where $f$ is referred to as the \emph{fanout}.
 We define a function $\textnormal{\textsf{MAP}}$ that associates/maps a given key $k_i$ to a set of buckets.

\begin{definition}[\textnormal{\textsf{MAP}}]
%Let $\mathcal{D}$ be a key-value dataset. Given $\mathcal{K}$, the set of all keys in $\mathcal{D}$, and a set $\mathcal{B}$ of buckets created over $\mathcal{D}$, 
We define $\textnormal{\textsf{MAP}}: \mathcal{K} \rightarrow \mathcal{P}(\mathcal{B})$ where $\mathcal{P}(\mathcal{B})$ is the powerset of $\mathcal{B}$. For a key $k_i \in \mathcal{K},$ the function
{$\textnormal{\textsf{MAP}}(k_i)$} returns a set of $f$ bucket-ids, denoted as $\mathcal{B}[k_i]$, that corresponds  to $f$ buckets in $\mathcal{B}$, \textit{i}.\textit{e}., $|\mathcal{B}[k_i]| = f$, such that each record of $k_i$ may reside in one of the buckets in $\mathcal{B}[k_i]$.
Further, $\forall k_i$ and $\forall B_j$
 such that $B_j \not \in \textnormal{\textsf{MAP}}(k_i)$, the records of $k_i \not \in B_j$. $\Box$

 %Note that for each key $k_i$,
 \end{definition}
 Note that based on the definition above, records corresponding to a key $k_i$ may or may not be in bucket $B_j$.
%%%NOTATION INCONSISTENCY
We illustrate the  notation  above using the example below.
\begin{example}[Example of \textnormal{\textsf{MAP}(*)}]
\label{ex:bucketsMAP}
Consider a key-value dataset $\mathcal{D}$ that contains three records: $\mathcal{D}=\{ \langle k_1, v_{1} \rangle, \langle k_1, v_{2} \rangle, \langle k_2, v_{3}\rangle \}$. Let $\mathcal{B} = \{ B_1, B_2, B_3 \}$ be the set of buckets created over $\mathcal{D}$, where bucket $B_1$ contains $\{\langle k_1, v_{1} \rangle\}$, bucket $B_2 $ contains $ \{\langle k_1, v_{2} \rangle, \langle k_2, v_3 \rangle\}$, and bucket $B_3$ is empty. Here, $\textnormal{\textnormal{\textsf{MAP}}}(k_1)=\{B_1, B_2\}$ and ${\textnormal{\textnormal{\textsf{MAP}}}(k_2) = \{B_2, B_3\}}$ indicates that records of $k_1$ reside in $B_1$ and/or $B_2$, and records of $k_2$ reside in $B_2$ and/or $B_3$.
%Note that all records of $k_2$ may only reside in $B_2$. %According to $\textnormal{\textsf{MAP}}(k_2)$, records with key $k_2$ can be found in $B_2$ and/or $B_3$, but not in any other buckets. Furthermore, records with a specific key $k_i$ may be stored in up to $f$ distinct buckets, according to $\textnormal{\textsf{MAP}}(k_i)$, and a single bucket may contain records from multiple keys, \textit{e}.\textit{g}., $B_2$ holds records from both $k_1$ and $k_2$.
$\Box$
\end{example}

Below we define \emph{well-formed buckets}. Intuitively, we say that a set of buckets formed over a dataset is \emph{well-formed} if each bucket are of the same size and $\textnormal{\textsf{MAP}}$ is defined appropriately.

\begin{definition}[Well-Formed Buckets]\label{def: well_form_bucket}
Let \textnormal{\textsf{MAP}} be
the function as defined above. %that maps keys in $\mathcal{K}$ to $f$ buckets in $\mathcal{B}$.
The buckets %We say that the buckets  $\mathcal{B} = \{B_1, \ldots, B_n\}$ 
are well-formed, if and only if the following hold:

\begin{enumerate}[noitemsep,nolistsep,leftmargin=0.22in]
{\color{black}    
\item \emph{Equal bucket size.} For all buckets $B_p$ and  $B_{q}\in\mathcal{B}$, $|B_p| = |B_{q}|$.

\item \emph{Disjoint Buckets.} For all buckets $B_p$ and  $B_{q}\in\mathcal{B}$, $B_p \cap B_q = \varnothing$.

\item \emph{Consistent mapping.}} For all buckets $B_p\in\mathcal{B}$  and all key $k_i\in\mathcal{K}$, if $B_p$ contains one or multiple records of $k_i$, then $B_p \in \textnormal{\textsf{MAP}}(k_i)$. 

    We denote the set of well-formed buckets after padding by $\mathcal{B}_f$.
    $\Box$
\end{enumerate}
\end{definition}

% Observe that buckets shown in Example \ref{ex:bucketsMAP} are not well-formed since buckets $B_1$, $B_2$, and $B_3$ contain different number of records (1, 2, and 0 respectively). Consider an alternative scenario where the buckets $B_1=\{\langle k_1, v_{1} \rangle\}$, $B_2 = \{\langle k_1, v_{2} \rangle \}$ and $B_3 = \{\langle k_2, v_3 \rangle\}$. Here the set of buckets $\{B_1, B_2, B_3\}$ are well-formed since they are of equal size and consistent with the \textnormal{\textsf{MAP}}() function. 

Buckets can be made equal-sized by adding fake records to them appropriately. Suppose $\theta_1, \theta_2$ and $\theta_3$ refer to ``fake'' records.  Consider, again, Example \ref{ex:bucketsMAP}. We add $\theta_1$ to bucket $B_1$ and $\theta_2$ and $\theta_3$ to bucket $B_3$. Thus, we have $B_1=\{\langle k_1, v_{1} \rangle, \theta_1\}$, bucket $B_2 = \{\langle k_1, v_{2} \rangle \langle k_2, v_3 \rangle\}$, and bucket $B_3 = \{\theta_2, \theta_3\}$. The set of buckets $\{B_1, B_2, B_3\}$ is now well-formed.

% We have introduced all the required notation to state \textsc{Veil} formally. 

\subsection{Components of \textsc{Veil}}\label{sec:components of veil}
\textsc{Veil}  is characterized by the 
 following five operations: Bucket Creation, Padding, Data Outsourcing, Query Evaluation, and Filtering. %We describe each of them in the following.

\paragraph{Bucket Creation $\mathcal{BC}$ ($\mathcal{D}$,$ QA, SA, f$) $\rightarrow
\langle$\textnormal{\textnormal{\textsf{MAP}}}, $\mathcal{B}, {\color{black}\mathit{stash}} \rangle$:}
The function $\mathcal{BC}$
takes the dataset $\mathcal{D}$, the parameters $\mathit{QA}$ (query amplification), the parameter $\mathit{SA}$ (storage amplification), and a fanout $f$ as inputs and returns
a mapping \textnormal{\textsf{MAP}}
from keys to buckets, a set 
of buckets $\mathcal{B} = \{B_1, \ldots, B_n\}$, and {\color{black}a stash (containing a few records that do not fit in any buckets).}  {\color{black}Observe that $\mathcal{B}$ is consistent to $\textnormal{\textsf{MAP}}$ but buckets in $\mathcal{B}$ may be of unequal size.

\paragraph{Padding} To make the buckets well-formed, \textsc{Veil} adds fake records to make the buckets equi-sized. 
Let $\ell_b$ be the bucket size (\S\ref{sec:veil_steps} will explain the method of computing $\ell_b$). For  each bucket $B_j \in \mathcal{B}$, if $|B_j| < \ell_b$, we add fake records to $B_j$ to pad it to size $\ell_b$.
We denote the set of well-formed buckets after padding by $\mathcal{B}_f$.}

{\color{black}
\paragraph{Data Outsourcing} This operation takes the well-formed buckets $\mathcal{B}_f$ and produces the following:

\begin{itemize}[noitemsep,nolistsep,leftmargin=0.15in]
\item \emph{Encrypted Record Set:} that includes the set of all real or fake records in each bucket $B_i \in \mathcal{B}_f$. All such records are appropriately encrypted and outsourced as a record set. Each record is associated with a RID that will be used in a multimap index. For example, a key-value pair or a record $\langle k_i, v\rangle$ in a bucket is represented as: $\langle RID, E(k_i, v)\rangle$, where $E$ is an encryption function, such as AES256~\cite{daemen1999aes}.\footnote{\scriptsize \color{black}AES produces an identical ciphertext for more than one appearance of a cleartext. Since each $\langle k_i, v\rangle$ pair is unique (\textit{i}.\textit{e}., differ in at least one bit), all the ciphertext values will be non-identical (due to avalanche effect property~\cite{feistel1973cryptography}). In case, there are two or more appearances of a $\langle k_i, v\rangle$ pair due to insert operation, we could add a random number to $\langle k_i, v, r\rangle$ before encryption to produce non-identical ciphertext.}

\item \emph{Multimap Index, $Mmap(B_i)$:} contains a map for each bucket $B_i {\in} \mathcal{B}_f$ consisting of a list of RIDs associated with records in $B_i$.
\end{itemize}
 
The encrypted record set and the multimap index for each bucket are outsourced to a server.
Furthermore, \textsc{Veil} also maintains information at the client for converting user queries into appropriate server-side queries. In particular, \textsc{Veil} stores the following information at the client:
\begin{itemize}[noitemsep,nolistsep,leftmargin=0.15in]

\item $f$: fanout that is the number of buckets in which records of a key may get mapped to.
\item $n$: total number of buckets created by bucketing.
\item Stash.
\end{itemize}

\emph{Aside.} Note that the above strategy for outsourcing the $Mmap(B_i)$ corresponds to implementing the multimap index as a secondary index over an encrypted database of records. We could, instead, also implement $Mmap(B_p)$ as a primary index in which case the $Mmap(B_p)$ would store encrypted records instead of RIDs to the encrypted record stored in the encrypted record store. 

}

{\color{black}
\paragraph{Query Evaluation $\mathcal{QE}(k_i) \rightarrow \mathcal{B}[k_i]$:}
$\mathcal{QE}$ takes
a query key $k_i$ as input from the user and fetches encrypted buckets stored at the public cloud. Particularly, the client utilizes \textsf{MAP} function, $\textnormal{\textsf{MAP}}(k_i)$, to determine the $f$ bucket-ids that may store 
the encrypted records of $k_i$ and sends the $f$ bucket-ids to the cloud to fetch the $f$ buckets. Depending on the way the data is stored at the cloud (either in the form of a secondary index, -- multimap index, or a primary index), the cloud returns all the $f$ buckets to the client.

\paragraph{Filtering:} This operation takes the query key $k_i$ and the encrypted records in the buckets retrieved by $\mathcal{QE}$ as inputs and decrypts them. All the records that are not corresponding to $k_i$ are discarded. Further, the client reads the stash to find records having the key $k_i$.} 

%Since buckets generated by $\mathcal{BC}$ are well-formed (hence equisized) and since $\mathcal{QE}$ retrieves the same number of buckets,\textit{i}.\textit{e}., $f$,  in \textsc{Veil} described above, the volume of data retrieved  remains the same irrespective of the query $k_i$.
Observe that in \textsc{Veil}, irrespective of the key,
the volume of data retrieved remains the same
because (1) buckets generated by $\mathcal{BC}$ are well-formed (hence equal-sized); (2) $\mathcal{QE}$ always retrieves the same number of buckets, \textit{i}.\textit{e}., $f$. 

{\color{black}
\noindent
\textbf{Leakage from $\mathcal{BC}$.} A potential leakage arises when the adversary possesses knowledge of the algorithm used by $\mathcal{BC}$ to create buckets. This serves as a motivation to address the design requirement of the $\mathcal{BC}$ algorithm developed in \textsc{Veil}, which will be discussed in \S\ref{sec:random_bucketing}. 
We illustrate a scenario where the adversary is aware that the algorithm $\mathcal{BC}$ uses to create buckets is the first-fit decreasing\footnote{\scriptsize FFD is a well-known bin-packing algorithm that sorts keys in descending order by the size and places each key into the first available bucket with sufficient space. If a key cannot fit in an existing bucket, the algorithm creates a new bucket for the key. FFD creates at most  $\frac{11}{9}\times OPT$ equisized buckets~\cite{ffd}, where $OPT$ is the least number of buckets required to store the database.} (FFD)~\cite{ffd} algorithm.
%to create buckets over the following dataset.
}

\begin{example}
Consider a dataset $\mathcal{D}$ with three keys $k_1$, $k_2$, and $k_3$ containing 3, 2, and 1 records, respectively. The bucket size is 3 (same as the  largest number of records associated with any key in $\mathcal{D}$). Using FFD~\cite{ffd}, we allocate the keys to buckets by first placing the largest key, \textit{i}.\textit{e}., $k_1$, in  bucket $B_1$. Since bucket size is 3 and $k_1$ has three corresponding records, $B_1$ lacks space for $k_2$. We create a new bucket $B_2$ for $k_2$. We also add $k_3$ to $B_2$. Thus, each query retrieves a bucket with 3 records from an adversary's perspective.

Consider an adversary (based on the adversarial model in \S\ref{sec:adv}) that has the following knowledge:
\begin{itemize}[noitemsep,nolistsep,leftmargin=0.15in]
\item Data distribution, \textit{i}.\textit{e}., knowledge that
$k_1$, $k_2$, and $k_3$ associate with 3, 2, and 1 records respectively.
\item Prior queries, \textit{i}.\textit{e}., the adversary knows that a prior query that retrieved $B_2$ corresponds to  $k_2$.
\end{itemize}
Given the data distribution, the adversary can execute the FFD algorithm by itself to determine that of the two buckets, one 
contains records of $k_1$ and the other
has records corresponding to $k_2$ and $k_3$. Now, if a query retrieves $B_1$, the adversary can infer 
that the query is for the key $k_1$. Likewise, if
the query returns $B_2$, 
 the adversary can determine that the query is for either $k_2$ or $k_3$ but NOT for $k_1$.
\end{example}

The example above illustrates leakage in scenarios where key are mapped to a single bucket using FFD, \textit{i}.\textit{e}., the fanout $f$ is 1.
Analogous examples can be created for other bucket creation algorithms, whether they yield a fanout of 1 or more, as these strategies construct buckets based on key sizes and a desired bucket size, which may inadvertently aid adversaries to infer information about keys. {\color{black} A possible solution to avoid such a leakage when using FFD is to use ORAM to fetch a bucket; however, this will come with communication overhead, as has been discussed in~\S\ref{sec:intro}.}

In this paper, we develop \emph{\textit{\textsc{Veil}, a bucketization-based strategy}} that 
prevents the attack illustrated above and \emph{\textit{ensures that the adversary cannot determine the volume of the results and thus cannot gain any information from the queries}}. 

To pad the buckets to well-formed, we also propose two strategies for padding,  including a basic strategy that adds fake records to buckets to make them equal-size (in~\S\ref{sec:random_bucketing}), and a sophisticated overlapping strategy that allows buckets to ``borrow'' records from other buckets to further minimum the number of fake records added while making the buckets well-formed (in~\S\ref{sec:overlapping}).

Throughout the rest of this paper, we mainly concentrate on the bucket creation operation, which allocates records to buckets and pads them to achieve a well-formed structure. We do not extensively discuss the outsourcing and filter operations, as they are relatively straightforward. Since the query operation is intrinsically connected to the way we create buckets, we  discuss it alongside bucket creation. 

\section{\textsc{Veil} with Random Bucket Creation }\label{sec:random_bucketing}

% In this section, our objective is to design bucket creation strategies to prevent volume leakage against adversaries. A straightforward idea is to use bin-packing algorithm to pack the given inputs into a set of buckets of a given size while ensuring that all the buckets (except possibly one) are at least half-full. In the context of key-value datasets, the bin-packing algorithm takes the key sizes (the number of values associated with each key) as inputs, fixes the bucket size to be the maximum key size, and allocates one or more keys to a bucket. The values associated with the keys are also allocated to the bucket.
% However, the bin-packing approach is not secure when the adversary knows the bin-packing algorithm as well as some previous queried keys.

\textsc{Veil} uses a random bucket creation strategy, denoted by 
$\mathcal{BC}$, that allows users
to specify input parameters of query amplification $\mathit{QA}$, storage amplification $\mathit{SA}$, 
as well as the fanout $f$. $\mathcal{BC}$ generates
a randomized  mapping between keys and buckets, denoted by $\textnormal{\textsf{MAP}}$, based
on which it creates set of buckets $\mathcal{B}$. 
In $\mathcal{BC}$, since keys are  assigned randomly to buckets
(\textit{i}.\textit{e}., $\textnormal{\textsf{MAP}}$ is a randomized function), there is always a chance that
not all records in $\mathcal{D}$ fit into the buckets in $\mathcal{B}$.  Records that do not fit into the assigned buckets, \textit{i}.\textit{e}., records in 
$\mathcal{D} - \cup_{B_j \in \mathcal{B}} B_j$ 
are assigned to a local storage that we refer to as a 
\emph{stash}.  Such a stash is stored at the local site (and not the public cloud).
With the possibility of a stash, the retrieval algorithm in \textsc{Veil} is slightly modified. 

Given a query key $k_i$, in addition to retrieving all the corresponding
buckets in $\textnormal{\textsf{MAP}}(k_i)$, the user also checks the stash for presence of records 
for key $k_i$. Note that the  effectiveness of the strategy depends upon the size of
the stash which we would like to be as small as possible.
The stash size in $\mathcal{BC}$ , as will become clear,  depends upon factors including
the $\mathit{QA}$, $\mathit{SA}$, and $f$. We theoretically show that
even when we set these factors close to their optimal values , \textit{i}.\textit{e}., 1, the expected size of the stash remains very small. This is also reflected by our experiments which clearly establish the
superiority of the $\mathcal{BC}$ not only in terms of $\mathit{QA}$ and $\mathit{SA}$ but also in a much
smaller size of stash compared to \textsf{dprfMM}~\cite{Moti_dprfMM} which also 
exploits the stash to store spillover records that do not fit into
the cuckoo hash tables.

The $\textnormal{\textsf{MAP}}$ function in \textsc{Veil} is implemented  using a hash function $\mathcal{H}$ such as SHA-256~\cite{sha-256}.
To generate bucket-ids for a query key $k_i$, \textsc{Veil} appends an integer counter $\gamma$ to the key $k_i$, where $1 \leq \gamma \leq f$. Subsequently, the hash function $\mathcal{H}$ processes the concatenated strings and produces $f$ distinct bucket-ids that are
then associated with the key $k_i$. The algorithm for $\textnormal{\textsf{MAP}}$ is shown in Algorithm~\ref{alg: map}. \\
%Note that to be able to generate the \textnormal{\textsf{MAP}} function for any key $k_i$ during query processing, the user simply computes $\textnormal{\textsf{MAP}}(k_i)$ using the hash function $\mathcal{H}$ over   $k_i$ concatenated with
%counters from 1 to $f$ to generate the buckets
%where records of $k_i$ might be present.
% needed to execute the query can then use the  key to generate the hash for the input qu

\LinesNotNumbered \begin{algorithm}[!t]
% \footnotesize
\small
%=========Inputs=========Inputs=========Inputs=========Inputs=========Inputs=========
\textbf{Inputs:} $k_i$: a query key; $n$: total number of buckets; $\mathcal{H}$: a hash function. 
\textbf{Outputs:} $\mathcal{B}[k_i]$: a list of $f$ bucket-ids corresponding to $k_i$.

{\bf Function} $\textnormal{\textsf{MAP}}(k_i, n, \mathcal{H})$
\Begin{
\nl $\mathcal{B}[k_i]\leftarrow []$

\nl \For{$\gamma\in [1,f]$}{

% \nl$\mathcal{B}[k_i]\leftarrow  

\nl $\mathcal{B}[k_i]$.$\mathit{append}$($\mathcal{H}$($k_i|\gamma$) \% $n$) \Comment{Map to $[0, n-1]$}

}

\nl $\mathbf{return}$ $\mathcal{B}[k_i]$
}

\caption{Map Algorithm}
\label{alg: map}
\end{algorithm}
\setlength{\textfloatsep}{0pt}

%The buckets created by $\mathcal{BC}$, while conforming to the randomized mapping of keys to buckets, might not be equal-sized. To transform
%the generated buckets into well-formed buckets, \textsc{Veil} further pads the buckets.  In the rest of this section, we describe the randomized bucket creation algorithm $\mathcal{BC}$ (\S\ref{sec: bucketing})  and then discuss the padding algorithm that we refer to as \emph{basic padding} (\S\ref{sec: disjoint}).
%\subsection*

\noindent 
\textbf{\textsc{Veil} Steps.}\label{sec:veil_steps}
We next discuss the components of \textsc{Veil} based on the
randomized bucket creation.

\paragraph{Bucket Creation  $\mathcal{BC}$}
\label{sec: bucketing}

%==========ALGORITHM FOR data encryption
\LinesNotNumbered \begin{algorithm}[!t]
% \footnotesize

\small
%=========Inputs=========Inputs=========Inputs=========Inputs=========Inputs=========
\textbf{Inputs:} $\mathit{SA}$: a desired storage amplification. $\mathit{QA}$: a desired query amplification. $f$: fanout. $L_{\mathit{max}}$: maximum key size in the dataset. $\mathcal{D}$: dataset. $\mathcal{K}$: key set. $\textnormal{\textsf{MAP}}$: a map function for each key and their corresponding buckets.

\textbf{Outputs:} $\mathcal{B}$: a list of buckets; $\mathcal{S}$: a stash.

\nl $\ell_b\leftarrow \lceil {QA\cdot L_{\mathit{max}}}/{f} \rceil$ \Comment{bucket size} \nllabel{ln:calculate_b_size}

\nl $n\leftarrow \lceil{SA\cdot|\mathcal{D}|}/{\ell_b}\rceil$ \Comment{the total number of buckets} \nllabel{ln:calculate_b_num}

\nl $\mathcal{B}\leftarrow \mathit{create}\_\mathit{empty}\_\mathit{buckets}(n)$ \Comment{create $n$ empty buckets} \nllabel{ln:create_empty_b}

% \nl $\mathit{map} \leftarrow []$

% \nl \For{$k \in \mathcal{K}$}{\For{$i\in [1,f]$}{$\mathit{map}[k]$.$\mathit{add}$($\mathcal{H}(k|i)$)}}

\nl $\mathit{shuffle}(\mathcal{D})$\Comment{mix the dataset $\mathcal{D}$} \nllabel{ln:mix_d}

\nl $\mathcal{S}\leftarrow []$ \Comment{a local stash}  \nllabel{ln:create_stash}

\nl \For{each tuple $\mathit{t}\in \mathcal{D}$}{ 

\nl $\mathit{key}\leftarrow \mathit{t}.\mathit{extract}\_\mathit{key}()$ \Comment{extract the key from $t$} \nllabel{ln:get_key}

\nl $\mathit{bucket}\_\mathit{id}\leftarrow \mathit{find}\_\mathit{smallest}\_\mathit{bucket}(\textnormal{\textsf{MAP}}(\mathit{key}))$ \nllabel{ln:get_smallest_bucekt}

\Comment{find the smallest bucket from the $f$ buckets of the key}

\nl \eIf{$\mathcal{B}[\mathit{bucket}\_\mathit{id}]$ is not full }{
\nl $\mathcal{B}[\mathit{bucket}\_\mathit{id}].\mathit{add}(\mathit{t})$ \Comment{add $t$ to the smallest bucket} \nllabel{ln:add_to_smallest_b}
}{ 
\nl $\mathcal{S}.\mathit{add}(\mathit{t})
\nllabel{ln:add_to_stash}$ \Comment{all buckets are full: add $t$ to stash}

}

}

\nl $\mathbf{return} $ $\mathcal{B}, \mathcal{S}$

\caption{$\mathcal{BC}$: Random Bucket Creation Algorithm.}
\label{alg: random_bucketing}
\end{algorithm}
\setlength{\textfloatsep}{0pt}

The pseudo-code for $\mathcal{BC}$ is presented in Algorithm~\ref{alg: random_bucketing}. Initially, the algorithm calculates the bucket size $\ell_b$, based on the maximum key size $L_{\mathit{max}}$, a user desired query amplification ($\mathit{QA}$), and the fanout $f$ (Line~\ref{ln:calculate_b_size}). Subsequently, the algorithm determines the number of buckets $n$ according to the desired storage amplification ($\mathit{SA}$) and the data size $|\mathcal{D}|$ (Line~\ref{ln:calculate_b_num}).
\begin{equation}
\ell_b = \frac{QA\times L_{\mathit{max}}}{f}, n = \frac{\mathit{SA}\times |\mathcal{D}|}{\ell_b} \label{eq: bucket_size_num}
\end{equation}
 % to prevent any leakage that may occur from the existing ordering of key-value pairs in the database  (Line~\ref{ln:mix_d}). % -- see Appendix~\ref{} on how this ordering can lead to leakage. 

Next, the algorithm generates $n$ empty buckets (Line~\ref{ln:create_empty_b}) and shuffles the dataset to mix the key-value pairs of different keys (Line~\ref{ln:mix_d}). 
Additionally, a local stash is established for key-value pairs that cannot fit into a bucket (Line~\ref{ln:create_stash}). 
For each key-value pair $\langle k_i, v\rangle$ arranged in a random order, 
the algorithm identifies the corresponding $f$ bucket-ids for $k_i$ based on $\textnormal{\textsf{MAP}}(k_i)$ (\textit{i}.\textit{e}., Algorithm~\ref{alg: map})
and locates the smallest among the $f$ buckets corresponding to $k_i$ (Line~\ref{ln:get_smallest_bucekt}) to place the key-value pair (Line~\ref{ln:add_to_smallest_b}). If the bucket is at capacity, indicating that all $f$ buckets corresponding to $k_i$ are full, the key-value pair will be placed into the local stash (Line~\ref{ln:add_to_stash}). 

{\color{black} 
Note that shuffling (in Line~\ref{ln:mix_d}) prevents the adversary from learning the order in which keys
are inserted into buckets. If we insert keys in the order in which they appear in the database,
the adversary may be able exploit such information to gain information about the ciphertext.
Say $\textsf{MAP}(k_1)=\{B_1, B_2, B_3\}$ and $\mathit{QA}=1$. Thus, all records in the 3 buckets are for $k_1$. Records of any other key mapped to these buckets will go to stash. Now, if an adversary gets to learn that a query $q$ is for $k_1$, it will learn that all records in $B_1$, $B_2$, and $B_3$ are for the same key $k_1$. Shuffling prevents such a situation. Shuffling is implemented by permuting records in the database prior to creating 
buckets.
}

Also, note that while employing $\mathcal{BC}$ utilizes user-defined $\mathit{SA}$ and $\mathit{QA}$ to determine the bucket size and the number of buckets. Consequently, the desired $\mathit{SA}$ sets a limit on the number of fake records needed to pad the buckets later.

\paragraph{Padding} 
As the buckets created by Algorithm~\ref{alg: random_bucketing} may contain different numbers of records. In order to generate equi-sized buckets of size $\ell_b$, we pad the buckets with fake records once $\mathcal{BC}$ has terminated. 
% {\color{purple}More formally, for each bucket $B_j \in \mathcal{B},$ if $|B_j| < \ell_b$, then add fake values.  
% We refer to the set of buckets in $\mathcal{B}$ wherein all the 
%  buckets $B_j$ have been padded with fake records as  
%  $\mathcal{B}_{f}$. Note that in $\mathcal{B}_{f}$ all buckets $B_p$ are equisized while the original set of buckets in $\mathcal{B}$ might not be.

% \paragraph{ Outsourcing. }
% Let $\mathcal{B}_f$ be the set of well-formed buckets created by 
% \textsc{Veil} (Algorithm~\ref{alg: random_bucketing}) after padding.
% Data in \textsc{Veil} is outsourced as consisting of
% two parts: an encrypted set of real and fake records in the 
% database (\textit{i}.\textit{e}. in the union of all the buckets in $\mathcal{B}_f$)
% and a map for each bucket $B_p$.

% \begin{itemize}
%     \item \emph{Encrypted Record Set}: The set of all records real or fake, in each bucket $B_p \in \mathcal{B}_f$,
%  are appropriately encrypted, shuffled, and outsourced as a record set. Each record is associated with a RID which will be used in a multimap index. 

%     \item \emph{Multimap Index:} consists of a map for each bucket $B_p \in \mathcal{B}_f$, denoted by 
%     $Mmap(B_p)$, which consists of a list of RIDs associated with  records in $B_p$.
% \end{itemize}
% The encrypted record set, along with the $Mmap[B_p]$  for each 
%  bucket $B_p$ is outsourced to the server. 
% }

{\color{black}
\paragraph{Data Outsourcing.} Let $\mathcal{B}_f$ be the set of well-formed buckets created by 
\textsc{Veil} (Algorithm~\ref{alg: random_bucketing}) after padding. Finally, the encrypted record set and multimap index are created and outsourced to the cloud, as explained in~\S\ref{sec:components of veil}.

\paragraph{Query Evaluation and Filter} A query for $k_i$ is executed by fetching $f$ buckets from the cloud, and then filter operation is executed at the client, as has been explained in~\S\ref{sec:components of veil}. } \\
%\textsc{Veil} employs $\textnormal{\textsf{MAP}}(k_i)$ to determine the $f$ bucket-ids associated with $k_i$, and subsequently retrieves all records in the corresponding buckets from the public cloud.
% {\color{purple}For each such bucket $B_p\in \textnormal{\textsf{MAP}}(k_i)$, the server retrieves each encrypted key-value pair that is indexed using $B_p$, \textit{i}.\textit{e}., $\langle B_p, E(k, v)\rangle$.

% {\color{black}
% \paragraph{Filtering}
% Finally, \textsc{Veil} executes filtering operation as explained in~\S\ref{3.2}. }

% decrypts the records using metadata stored at the client. On decryption, it filters away all records not corresponding to the 
% query keyword. In addition, it retrieves records in the stash, if any, corresponding to the query keyword to determine the final answer to the query. Since the filtering step remains identical in all strategies  developed in the paper, henceforth, we will not discuss filtering in the remainder of the paper.

\noindent
\textbf{Discussion:}
Analysis of \textsc{Veil} based on both performance \& security is formally presented in Appendix \S\ref{sec:security_analysis_veil}. 
From the performance perspective, \textsc{Veil} provides guaranteed storage and
query amplification as specified by the user while ensuring that the expected size of stash is zero. From the
security perspective, \textsc{Veil} is secure based on the security requirement VSR in \S\ref{sec:adv}.

\section{Reducing Storage Amplification }\label{sec:overlapping}
% \BB
\textsc{Veil} creates equal-sized buckets by padding
each bucket to the same size, as discussed in \S\ref{sec:random_bucketing}. The resulting buckets are
``disjoint'', indicating that no two buckets share common records. In this section, we propose an optimization to \textsc{Veil} that 
reduces storage amplification $\mathit{SA}$ \emph{without} increasing the stash size or query amplification $\mathit{QA}$. We introduce an alternate strategy called
\textsc{Veil-O} that allows buckets to share records, thereby decreasing the need to insert fake records to equalize the bucket sizes. The letter 
O represents ``overlapping'' between buckets due to the common records.
\begin{definition}
\label{def: overlapping}
Given a set of buckets $\mathcal{B}$ where each bucket contains key-value pairs, we say two distinct buckets $B_p$ and $B_q$ in $\mathcal{B}$ are \emph{overlapping} if there exists a key-value pair $\langle k,v \rangle$ in both $B_p$ and $B_q.$ The records shared by the overlapping buckets are referred to as \emph{common records}.  The number of common records in $B_p \cap B_q$, \textit{i}.\textit{e}., $|B_p \cap B_q|$, is the \emph{overlapping size}.  $\Box$
\end{definition}

Intuitively, the optimized strategy allows
buckets to borrow key-value pairs from other buckets instead of inserting fake records to pad themselves
to the desired bucket size $\ell_b$, thereby reducing the storage amplification $\mathit{SA}$.

\subsection{Overlapping Buckets}\label{sec:veil_o_intro}
Creating well-formed buckets by borrowing records from other buckets may seem straightforward. However, borrowing records indiscriminately may lead to information leakage.

\begin{example}\label{example:unsecure_clustering}
Consider a dataset $\mathcal{D}$ with three keys, ${k_1, k_2, k_3}$, containing 3, 1, and 3 values, respectively. Let the desired fanout $f$ be 1, and let $\mathcal{D}$ be partitioned into three buckets $\mathcal{B}=\{B_0, B_1, B_2\}$ of size 3, where $B_0$ contains records of $k_1$, $B_1$ contains records of $k_2$, and $B_2$ contains records of $k_3$. To create equal-sized buckets, we allow $B_1$ to borrow records from $B_0$ and $B_2$ to increase its size to 3, as illustrated in Fig.~\ref{fig: overlapping_unsecure_example}. 
Observe that in this example, the buckets are well-formed: each bucket is equal-sized and consistent with the $\textnormal{\textsf{MAP}}$ function, which maps $k_1$, $k_2$, and $k_3$ to $B_0$, $B_1$, and $B_2$, respectively.

\begin{figure}[h]
% \BBB\B
\centering
\includegraphics[width=6cm]{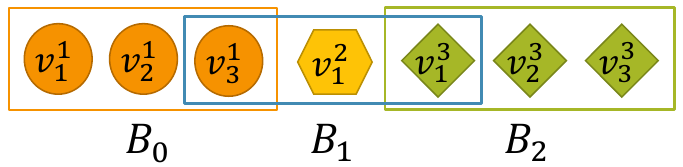}
% \BBB
\caption{An Intuitive Example of Unsecure Buckets} \label{fig: overlapping_unsecure_example}
% \B\BBB\B
\end{figure}

Such overlapping buckets can, however, lead to leakage. An adversary may infer buckets for some keys using the data distribution and access pattern (\textit{i}.\textit{e}., the records retrieved when a query is executed). Let's examine how an adversary could infer that $\textnormal{\textsf{MAP}}(k_2) = B_1$ in the example above. 
Assume the adversary observes three queries that retrieve records in buckets $B_0$, $B_1$, and $B_2$, respectively. $B_0$ and $B_1$ share a common record (in the intersection of $B_0$ and $B_1$), so do $B_1$ and $B_2$. Note that $B_0$ and $B_2$ do not intersect.
Given that $k_1$ and $k_3$ each have 3 records, and $k_2$ has 1 record, the adversary can easily deduce that neither $k_1$ nor $k_3$ could be mapped to $B_1$. If either were, all records in $B_1$ would correspond to $k_1$ (or $k_3$), in which case the other key with 3 records (\textit{i}.\textit{e}., $k_3$ or $k_1$) would not have enough space for its records in either $B_0$ or $B_2$, since at least one of the records in those buckets (\textit{i}.\textit{e}., the intersecting record with $B_1$) belongs to $k_1$ (or $k_3$).
As a result, it must be the case that $\textnormal{\textsf{MAP}}(k_2) = \{B_1\}$, and the first three records and the last three records in Fig.~\ref{fig: overlapping_unsecure_example} are for keys $k_1$ or $k_3$, respectively. Thus, the adversary not only learns which query is for key $k_2$, but also which ciphertext (in this case, the ciphertext in $B_1$ that does not intersect with either $B_0$ or $B_2$) corresponds to the key-value pair for key $k_2$.$\Box$
\end{example}

{\color{black}
 %push to paper
As shown in the example above,
while allowing buckets to overlap
can help ensure that results
returned for queries are equi-sized, 
the resulting  access pattern (which specific
records get retrieved) could lead to 
inferences about the keyword. 
Of course, if the underlying system used
an  access pattern hiding technique such   
as ORAM to prevent access pattern leakage, the leakage above would not occur.
But as  we mentioned in \S\ref{sec:intro}, technques 
such as ORAM are can be computationally prohibitive. Instead, in \textsc{Veil-O}  we devise 
 clever ways to share records amongst buckets such that resulting access patterns  do not leak additional information. This is described in the remainder of this section.

}
To prevent leakage for overlapping buckets, we must ensure that the adversary cannot distinguish between buckets based on the intersections between them. For a bucket, we define the notion of its {\bf neighborhood} as the set of all buckets it overlaps with.

\begin{definition}\label{def:neighborhood}
Given a set of buckets $\mathcal{B}$, for a bucket $B_j \in \mathcal{B}$, we define its \emph{neighborhood} as %$\mathcal{N}(B_j)= \{ B_p\in \mathcal{B} \mid B_j \text{ and } B_p \text{ share common records}\}$
$\mathcal{N}(B_j)= \{ B_p\in \mathcal{B} \mid |B_j \cap B_p|>0\}$. 
We call each bucket in $\mathcal{N}(B_j)$ a neighbor of $B_j$. 
%For a bucket of $B_p$, we call each bucket in $\mathcal{N}(B_p)$, a \emph{neighboring bucket} and denote them with $B^j_p$ where $1 \leq j \leq \mid \mathcal{N}(B_p)\mid.$
\end{definition}
%REPEATING In the example above, the intersection between buckets $B_2$ and each of $B_1$ and $B_3$ contains one record, while the intersection of $B_1$ and $B_3$ is empty, making the buckets distinguishable. This, combined with the data distribution, can lead to leakage.
The well-formed bucket criteria in Definition~\ref{def: well_form_bucket}, while sufficient, is not necessary to prevent leakage, thus we generalize the well-formed criteria for overlapping buckets.
%REPEATING: To prevent the adversary from gaining further knowledge based on intersections between overlapping buckets, we generalize the criteria for well-formed buckets in the context of overlapping buckets.

\begin{definition}[\textbf{Well-Formed Buckets with Overlap.}]\label{def:well_formed_overlapping}
Let $\mathcal{D}$ be a key-value dataset, $\mathcal{K}$ be the set of keys in $\mathcal{D}$, and $\mathcal{B}$ be the set of $n$ buckets created over $\mathcal{D}$.
Let \textnormal{\textsf{MAP}}(*) be the function that maps keys in $\mathcal{K}$ to $f$ buckets in $\mathcal{B}$. We say that the buckets in $\mathcal{B}$ are well-formed if and only if  %we have:
\begin{enumerate}[noitemsep,nolistsep]

    \item  \emph{Equal bucket size}. For all buckets $B_p$ and $B_{q}{\in}\mathcal{B}$, we have $|B_p| {=} |B_q|$.
    \item {\color{black}\emph{Constraints on Overlap.} 
    \begin{itemize}
        \item \emph{Equal sized neighborhood.}  For all buckets $B_p$ and $B_{q}\in\mathcal{B}$,$|\mathcal{N}$($B_p$)$|= |\mathcal{N}$($B_{q}$)$|$.
        \item \emph{Equal overlapping size}.  For all   $\mathcal{B}_1$ and $\mathcal{B}_2$ in $\mathcal{P}(\mathcal{B})$, where 
        $\mathcal{P}(\mathcal{B})$ is the power set of $\mathcal{B}$, with
    \begin{itemize}
        \item $|\mathcal{B}_1| = |\mathcal{B}_2|$,
        \item for all $B_1, B_1'' \in \mathcal{B}_1$, buckets
        $B_1'$ and $B_1''$ are overlapping, and for all $B'_2, B_2'' \in \mathcal{B}_2$,
        %we have $B'_2 \bigcap B_2'' \neq \varnothing$ 
        buckets $B'_2$ and $B_2''$ are overlapping,
    we have  $|\bigcap\limits_{B \in \mathcal{N}(B_p)} B| = |\bigcap\limits_{B' \in \mathcal{N}(B_q)} B'|$.      
    \end{itemize} 
    \end{itemize}
    \item \emph{Consistent mapping.} For all buckets $B_p\in\mathcal{B}$  and all key $k_i\in\mathcal{K}$, if $B_p$ contains one or multiple records of $k_i$, then $B_p \in \textnormal{\textsf{MAP}}(k_i)$.}
$\Box$
    \end{enumerate}
\end{definition}

% \begin{theorem}
% \label{th: veil_o}
% Let $\mathcal{D}$ be a dataset with the key set $\mathcal{K}$, and $\mathcal{B}$ be a set of buckets created over $\mathcal{D}$ using $\mathcal{BC}$. 
% Suppose we use a strategy to create well-formed overlapping buckets, we have 
% $|\bigcap\limits_{B \in \mathcal{B}_p} B| = |\bigcap\limits_{B' \in \mathcal{B}_q} B'|$. 
% Then the strategy
% is secure with respect to VSR.
% \end{theorem}

Intuitively, the well-formed definition above requires each bucket to overlap with the same number of buckets, and the number of common records between a subset of overlapping buckets of any size to be the same. This ensures buckets are indistinguishable from each other. 
%Buckets created using the random bucketing strategy that satisfies the well-formed overlapping criteria above, prevent leakage. 
{\color{black} Note that the well-formed definition in \S\ref{def: well_form_bucket} is subsumed by the definition above - \textit{i}.\textit{e}., if buckets 
are not overlapping, the above definition reduces to that of the well-formed buckets as in Definition~\ref{def: well_form_bucket}.
Henceforth, by well-formed, we refer to the well-formed buckets with overlap. 
%In the following section, we present a strategy based on the above criteria.

To form well-formed overlapping buckets, we superimpose 
a structure of a random $d$-regular graph over $\mathcal{B}$, 
the set of buckets. In the graph, $G = (V,E)$, each vertex $v_i \in V$ corresponds to the bucket $B_j \in \mathcal{B}$, and edges correspond to $d$ neighbors, which we will see are assigned randomly. For illustration, let us consider a  situation where $d$ is 3, thus, each node has  exactly 3 neighbors 
 (\textit{i}.\textit{e}., intersection of any four buckets is empty). Fig.~\ref{fig: d_regular_graph_overlapping_buckets} illustrates such a 3-regular graph with an initial set of buckets $\mathcal{B}$= $\{B_0, B_1, B_2, B_3\}$ 
 and the corresponding neighbors.
% In this section, we consider a special case in Definition~\ref{def:well_formed_overlapping}-Requirement (2), where $|\mathcal{B}_1| = |\mathcal{B}_2| = 2$, \textit{i}.\textit{e}.,
% the subsets of $\mathcal{B}$ only have two buckets in them. This entails that a record can exist in at most two buckets, \textit{i}.\textit{e}., a record may reside only in one bucket, or be shared between two buckets.
%We utilize a regular graph $\mathcal{G}$ with a degree of $d$, in which the nodes (also known as vertices) symbolize buckets. Each edge within $\mathcal{G}$ links a pair of neighboring buckets containing shared records. As $\mathcal{G}$ is a $d$-regular graph, each node has $d$ incident edges, implying that each bucket has $d$ neighboring buckets.
Neighboring buckets borrow/lend records to each other, and hence may overlap. 

With the $d$-regular graph $\mathcal{G} = (V, \mathcal{E})$, we further associate the following:
(\textit{i}) \textbf{Weight} (denoted by  $\delta$) that specifies the number of records shared between one bucket and each of its neighbors. Note that $\delta$ is a constant - \textit{i}.\textit{e}., the number of records shared between neighbors is always equal to $\delta$. 
(\textit{ii}) \textbf{Direction} that is a function $\mathit{dir}(v_p, v_q): E\rightarrow\{0,1\}$
that assigns to each edge in $G$ a direction representing which corresponding bucket borrows/lends to its neighbor. Suppose bucket $B_p$ borrows from bucket $B_q$ then $\mathit{dir}(v_p, v_q) = 0$ and $\mathit{dir}(v_q, v_p) = 1$.
(\textit{iii}) \textbf{Labels} that indicate the specific common records that neighbors borrow/lend from each other. For an edge $(v_p, v_q)$, if $\mathit{dir}(v_p, v_q) = 0$, then the label $\mathit{label}(v_p, v_q)$ is the set of the records that  $B_p$ borrows from $B_q$. Conversely, if $\mathit{dir}(v_p, v_q) = 1$, then $\mathit{label}(v_p, v_q)$ is the set of the records that $B_q$ borrows from $B_p$.
We illustrate an example graph $\mathcal{G}$. {\blue{To simplify notations, we only show values in the buckets.}}
}

% \begin{figure*}[htbp]
%     \centering
%     \begin{minipage}{0.33\textwidth}
%         \centering
%         \includegraphics[scale=0.25]{images/d_regular_graph_overlapping_buckets_modified.pdf}
%     \caption{An Example of a 3-Regular Graph. (To simplify notations, we only show the values in the buckets.)}\label{fig: d_regular_graph_overlapping_buckets}
%     \end{minipage}%
%     % \hfill
%     \begin{minipage}{0.33\textwidth}
%         \centering
%         \includegraphics[width=\linewidth]{exp_images/exp_real_sa_vs_d.png}
%         \BB\BBB\BBB\BBB
%         \caption{Impact of Degree ($\mathit{QA}$=1, $f$ = 6)}\label{fig: exp_real_sa_vs_d}
%     \end{minipage}%
    
% \end{figure*}

%\begin{example}
Suppose the generated buckets $\{B_0, B_1, B_2, B_3\}$ contain 4, 2, 1, and 3 records, respectively, with the bucket size to be 4.
We can  pad each bucket to the size of 4 as follows.
 $B_0$ currently contains 4 records and, thus, cannot accept additional values. It can, however, provides one value to each of its neighbors $B_1$, $B_2$, and $B_3$. In turn, $B_3$ receives one record from $B_0$, achieving the target bucket size of 4, and contributes one record to each of $B_1$ and $B_2$. Now, $B_1$ receives two records from $B_0$ and $B_3$ to reach the desired size of 4, subsequently offering one record to $B_2$. Finally, $B_2$ receives three values from $B_0$, $B_1$, and $B_3$ to fulfill its size requirement of 4. In Fig.~\ref{fig: d_regular_graph_overlapping_buckets}, each bucket overlaps with 3 buckets, the weight over each edge is 1, and each bucket is padded to size 4, thus the buckets are indistinguishable from the adversary's perspective. 
%\end{example}

Observe that overlapping strategies reduce $\mathit{SA}$ by allowing sharing records across buckets. In Fig.~\ref{fig: d_regular_graph_overlapping_buckets}, each bucket overlaps with 3 buckets to achieve the desired bucket size. The number of fake records required is 0. On the other hand, when simply padding each bucket to size 4 using fake records, $B_1$, $B_2$, and $B_3$ require 2, 3, and 1 fake records, respectively, and 
the total number of fake records is 6.

%\begin{corollary}
%Let $\mathcal{D}$ be a dataset with the key set $\mathcal{K}$, and $\mathcal{B}$ be a set of buckets created over $\mathcal{D}$ using $\mathcal{BC}$. 
%Suppose we use a strategy to create well-formed overlapping buckets
%such that for all $\mathcal{B}_{p}$ and $\mathcal{B}_q$  in $\mathcal{P}(\mathcal{B})$ with $|\mathcal{B}_{p}|=|\mathcal{B}_{q}|$ = 2, we have 
%$|\bigcap\limits_{B \in \mathcal{B}_p} B| = |\bigcap\limits_{B' \in \mathcal{B}_q} B'|$. 
%Then the strategy
%is secure with respect to VSR.
%\end{corollary}

%This follows from Theorem~\ref{th: veil_o}, where the size of all subsets of $\mathcal{B}$ is two.

% \begin{theorem}[Overlapping Theorem]
% \label{th: overlapping_buckets}
% Consider \textnormal{\textsf{MAP}} as a mapping of keys to clusters and buckets. Each mapped bucket contains all values of the corresponding key, and a value may appear in one or more buckets. If the clusters satisfy the following conditions: (i) each cluster forms a $d$-regular graph of $\beta$ buckets; (ii) each bucket is well-formed; (iii) the weight on each edge is identical; and (iv) a common value can only be shared by two buckets, then the corresponding \textsf{retrieval} can hide volumes and the adversary cannot tell the differences between clusters.
% \end{theorem}

{\color{black}
Based on the idea to pad buckets to make them
equi-sized by lending/borrowing records from neighboring buckets, we develop a strategy, entitled 
\textsc{Veil-O}, that ensures the security requirement  VSR. 
\textsc{Veil-O} uses the same randomized strategy to create initial buckets $\mathcal{B}$ as that used in \textsc{Veil}. It, however, 
allows buckets to overlap to make the buckets equi-sized as discussed in the following subsection.

\subsection{\textsc{Veil-O}: Padding}
\label{sec:veil_o}

\textsc{Veil-O} uses a padding strategy that allows sharing of records between buckets to reduce the number of fake records added, thereby reducing $SA$.
Given a KV-dataset $\mathcal{D}$, 
\textsc{Veil-O} first creates a set of $n$ buckets $\mathcal{B}=\{B_0, \ldots, B_{n-1}\}$ using the randomized bucket creation algorithm $\mathcal{BC}$ (\S\ref{sec: bucketing}), then
implements a padding strategy that, unlike \textsc{Veil}, enables buckets to borrow records from their neighboring buckets, creating  well-formed buckets with overlap as defined in 
Definition~\ref{def:well_formed_overlapping}.
The padding strategy consists of a sequence of steps (an overview of the strategy is presented in Algorithm \ref{alg:VEIL-OPadding}) that starts by first creating a 
 $d$-regular graph from the buckets in $\mathcal{B}$ 
(function \textsc{Graph Creation} (GC)). Such a graph has
$n$ nodes, each corresponding to a bucket in $\mathcal{B}$, and  undirected edges between overlapping buckets (\textit{i}.\textit{e}.,
buckets that share records). 
The graph $\mathcal{G} = (V, \mathcal{E})$ is represented 
as $\mathcal{B}$, representing
the set of buckets (vertices), and $\mathcal{M}$, an 
adjacency matrix corresponding to $\mathcal{E}$.
The next step  determines the maximum possible  
 overlap size
 between neighbors that will still result in well-formed overlapping buckets, \textit{i}.\textit{e}., function \textsc{Maximum Overlap Determination} (\textsc{MOD}).
 The \textsc{MOD} function returns the number of records
 that a bucket $B_p$ can borrow/lend from/to a specific 
 neighbor $B_q$. Next, \textsc{Veil-O} uses the function 
\textsc{Edge Direction Determination (EDD)}
to determine the directions of edges, \textit{i}.\textit{e}., for each pair of neighbors $B_p$ and $B_q$ whether $B_p$ borrows or lends records from/to $B_q$. This determines $\mathit{dir}(v_p, v_q)$ for each $(v_p, v_q) \in \mathcal{E}.$ We represent 
$\mathit{dir}(v_p, v_q)$ by transforming 
 $\mathcal{M}$ into an adjacency matrix $\overline{\mathcal{M}}$ that represents
 directed edges. Thus, if $\mathit{dir}(v_p, v_q) = 0$ and $\mathit{dir}(v_q, v_p) = 1$ (\textit{i}.\textit{e}., 
 $B_p$ borrows from $B_q$), 
 we remove the edge $(v_p, v_q)$ from 
 $\overline{\mathcal{M}}$ but let the edge
 $(v_q, v_p)$ remains.
 % In the former case, the undirected edge between
% $B_p$ and $B_q$ are converted into a directed edge from
% $B_p$ to $B_q$ while in the latter case  from $B_q$ and $B_p$. 
Next, fake records are added to buckets to ensure
that the size of each bucket consisting of all records
assigned to it originally, in addition to all the 
records it borrows from its neighbors, plus 
fake tuples added up to exactly the bucket
size. 
Addition of fake tuples changes the 
buckets in $\mathcal{B}$  resulting in $\mathcal{B}^f$, where $f$ represents ``full''. A bucket $B_p \in \mathcal{B}^f$ contains fake records in 
addition to the original records associated
with $B_p$.
Next, the algorithm determines
the labels to be associated with each
edge in the $d$-regular graph $\mathcal{G}$. Each label represents the records that are borrowed/lent between two overlapping buckets. 
As a final step, 
using the labels generated, and the
buckets in $\mathcal{B}^f$, 
 \textsc{Veil-O} generates the 
set of well-formed buckets which are then outsourced.
%As we will see, the $d$-regular graph, along with its weight $\delta$, the functions $dir$ and $label$ together represent well-formed buckets.

\begin{figure*}[htbp]
% \BBB\BB
    \centering
    \begin{minipage}{0.33\textwidth}
        \centering
         \includegraphics[scale=0.355]{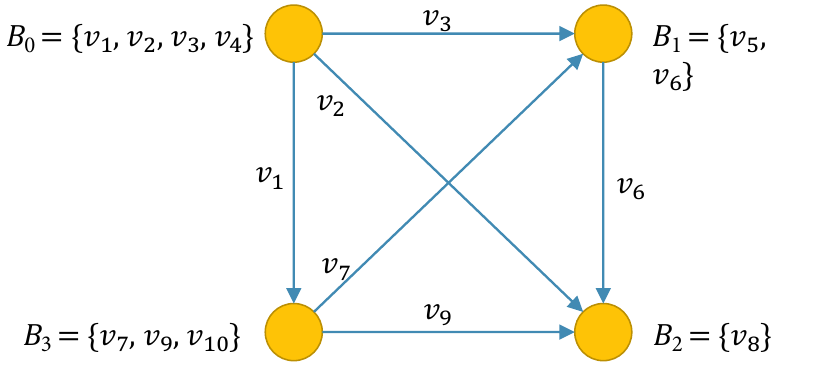}
% \BBB
\vspace{0.05em}
\caption{A 3-Regular Graph. } \label{fig: d_regular_graph_overlapping_buckets}
    \end{minipage}%
    % \hfill
    \begin{minipage}{0.33\textwidth}
        \centering
         \includegraphics[scale=0.25]{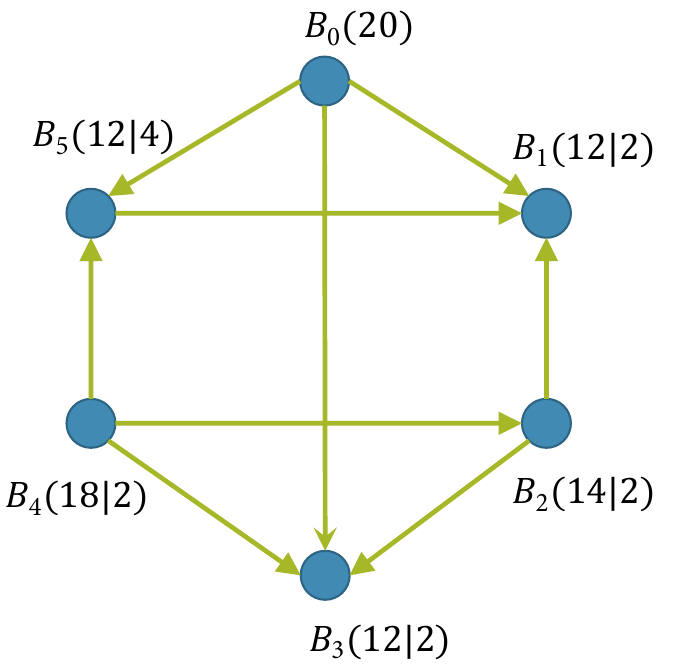}
         \BBB
    \caption{A 3-Regular Graph of \textsc{Veil-O}}\label{fig: single_cluster}
    \end{minipage}%
    \begin{minipage}{0.33\textwidth}
        \centering
        
{\scriptsize 
\begin{tabular}
{|m{0.2cm}|m{0.2cm}|m{0.2cm}|m{0.2cm}|m{0.2cm}|m{0.2cm}|m{0.2cm}|}
\hline
  % \centering
& $B_0$ & $B_1$ & $B_2$ & $B_3$ & $B_4$ & $B_5$ \\\hline
$B_0$ & & \cmark & & \cmark & & \cmark \\\hline
$B_1$ & & & & & &  \\\hline
$B_2$ & & \cmark & & \cmark & &  \\\hline
$B_3$ & & & & & & \\\hline
$B_4$ & &  & \cmark & \cmark & & \cmark \\\hline
$B_5$ & & \cmark & & & &  \\\hline

% \bottomrule
\end{tabular}
}

% \end{table}
% \end{center}
\vspace{1.1em}
\caption{Adjacency Matrix $\overline{\mathcal{M}}$} \label{fig: veil-OGraphMatrixExample}

    \end{minipage}%
    % \BBB\BB
\end{figure*}

We illustrate the above described algorithm in  Algorithm~\ref{alg:VEIL-OPadding} and then describe each step in details.  The algorithm takes the set of buckets $\mathcal{B}$ created using $\mathcal{BC}$ as inputs and
returns $\delta$, the overlapping size, \textit{i}.\textit{e}., the number of records shared by any two overlapping buckets, and a 
matrix $\overline{\mathcal{M}}$, representing the $d$-regular graph containing
information about neighbors and a directed edge representing which one of each two neighboring buckets borrow/lends to the other.  The strategy of which specific records a bucket $B_p$ will borrow from/lend
to $B_q$, its neighbor such that the resulting buckets become well-formed will be discussed
separately.
 \\

%\vishal{Change bucket index in figures and write up. (2) Padding algorithm rewriting with individual sub algorithms (3) Outsourcing - S1 and S1 and corresponding Quty strategy Q1 and Q2}
%==========ALGORITHM FOR data encryption

\LinesNotNumbered \begin{algorithm}[!t]

\small
\textbf{Inputs:} $\mathcal{B}$: buckets created using $\mathcal{BC}$ (Algorithm~\ref{alg: random_bucketing}); 
$d$: the number of neighbors for each bucket; 
$\ell_b$: bucket size.

\textbf{Outputs:} $\mathcal{B}_f$: A set of well-formed overlapping buckets;
%
%$\overline{\mathcal{M}}$: An adjacency matrix containing the direction of each edge in the graph;
%
%$\mathcal{L}$: Labels for each edge in the graph
%
\Begin{
{\nl $\mathcal{B}$, $\mathcal{M},  \leftarrow \boldsymbol{\mathit{GC}}(\mathcal{B}, d)$}

{\nl $\delta_{\mathit{MAX}} \leftarrow \boldsymbol{\mathit{MOD}(\mathcal{B}, \mathcal{M})}$ }

{\nl $\mathcal{B}, \overline{\mathcal{M}}, \delta \leftarrow \boldsymbol{\mathit{EDD}(\mathcal{B}, \mathcal{M}, \delta_{\mathit{MAX}}, \ell_b)}$}

{\nl $\mathcal{B}_f, \overline{\mathcal{M}}, \delta \leftarrow \boldsymbol{\mathit{AFV}(\overline{\mathcal{M}}, \delta )}$}

{\nl $\mathcal{L} \leftarrow \boldsymbol{\mathit{LC}(\mathcal{B}_f, \overline{\mathcal{M}}, \delta )}$}

{\nl $\mathcal{B}_f \leftarrow \boldsymbol{\mathit{WFBC}(\mathcal{B}_f, \mathcal{L}  )}$}

%{\nl $\mathcal{EB}\leftarrow \mathit{Encrypt}(\mathcal{B})$}

%{\nl $\mathit{Outsource}(\mathcal{EB})$}

%{\nl $\mathbf{return}$ $\mathcal{B}_f, \overline{\mathcal{M}}, \delta, \mathcal{L}
%$}
{\nl $\mathbf{return}$ $\mathcal{B}_f
$}

}

\caption{\textsc{Veil-O} Padding Algorithm}
\label{alg:VEIL-OPadding}
\end{algorithm}
\setlength{\textfloatsep}{0pt}
\noindent\textsc{\textbf{Step 1. Graph Creation (GC).}} [Algorithm \ref{alg: Veil-O-GC}]
Input to this step are buckets in 
$\mathcal{B}$ and a value of $d$ and the output consists of 
a $d$ regular graph  $\mathcal{G}$ with nodes corresponding to buckets in $\mathcal{B}$ and undirected edges 
 between buckets that will be neighbors to each bucket\footnote{We will experimentally show the selection of $d$ in Experiment 5 in \S\ref{sec: experiment}.}. It is assumed that at least one of $d$ and $|\mathcal{B}|$ are even, as the number of edges in a $d$-regular graph is computed as $d|\mathcal{B}|/2$, and 
if both $d$ and $|\mathcal{B}|$ are odd then a $d$-regular graph cannot be constructed. The graph is represented with $\mathcal{B}$, the set of vertices, and an adjacency matrix $\mathcal{M}$ denoting the edges.
Let $\mathcal{B} = \{B_0, B_1, \ldots, B_{n-1}\}$. 
For a bucket $B_p \in \mathcal{B}$, 
\textsc{Veil-O} assigns a set of $d$ neighboring buckets $B_p^j$ ($j=1,\ldots, d$) using an \emph{ordered
set} of 
carefully chosen  $d$ functions $\mathcal{F} = \langle F_1, F_2, \ldots, F_d \rangle$, where for each function $F_j\in\mathcal{F}$,
$F_j: \mathcal{B} \rightarrow \mathcal{B}.$
We denote 
$F_j(B_p)$ as  $B_p^j$ 
and refer to it as the
$j^{th}$ neighbor of $B_p$. 
The functions in $\mathcal{F}$ are such that  if 
a bucket $B_q$ is a neighbor of $B_p$ (\textit{i}.\textit{e}., for some  function $F_j \in \mathcal{F}$,   $B_q = F_j(B_p)$),  then $B_p$ is also a neighbor of $B_q$, \textit{i}.\textit{e}., there exists a
 $1 \leq j^{\prime} \leq d$ such that $F_{j^{\prime}}(B_q) = B_p$. 

We define a set of such  functions in $\mathcal{F} = \{F_1, \ldots, F_d\}$ as follows.
\begin{itemize}[noitemsep,nolistsep]

         \item For $1 \leq i \leq \lfloor {d}/{2} \rfloor$ we have $F_i(B_p) = B_{q}$ where $q=(p + i)\mod{n}$
         % \item For $\lceil \frac{d+1}{2} \rceil \leq i < \lceil \frac{d+1}{2} \rceil  + \lceil \frac{d-1}{2} \rceil$ we have $F_i(B_p) = B_{(p - (i-\lfloor \frac{d}{2} \rfloor )) \mod{n}}$
         \item For $1 {\leq} i {\leq} \lfloor {d}/{2} \rfloor$, we have $F_{d-i+1}(B_p){=}B_q$ where $q{=}{(p-i+n){\bmod} n}$

         \item If $d$ is odd $F_{{(d+1)}/{2}}(B_p) = B_q$ where $q={(p + ({n}/{2}))\mod{n}}$.\footnote{Note that in this case $n$ must be even.}
\end{itemize} 
We illustrate the functions used for assigning neighbors using 
the example below.
\begin{example}
\label{ex: graphBucket}
% B0, B1, ... B5
% d = 2
% f1 = x +1 mod n
% f2 = x-1 mod n

% 0  -- {1, 5}
% 1 -- {2,0}
% 2 - {1,3}
Consider  six  buckets  $B_0, \ldots, B_5$ in Fig.~\ref{fig: single_cluster}
%with number of records $20, 12, 14, 12, 19, \text{ and } 12$, respectively. 
Thus, $n=6$. Let us consider $d$ to be $3$. %Since $3$ is odd, 
According to rule 1, we get $F_1(B_p) = B_{(p+1)\mod{6}}$,
which finds neighbor bucket $B_1$ for $B_0$, $\ldots$, %$B_2$ for $B_1$, $B_3$ for $B_2$, $B_4$ for $B_3$, 
$B_5$ for $B_4$, and $B_0$ for $B_5$. According to rule 2, we get $F_3(B_p)=B_{(p-1+6)\mod n}$, which finds neighbor bucket $B_5$ for $B_0$,  
$B_0$ for $B_1$, $\ldots$, and $B_4$ for $B_5$. According to rule 3, we get $F_2=B_{(p+\frac{n}{2})\mod n}$, which finds neighbor bucket $B_3$ for $B_0$,  
$B_4$ for $B_1$, $\ldots$, and $B_2$ for $B_5$.
Thus, 
neighbors for buckets will be:
$\mathcal{N}(B_0)=\{B_5, B_1, B_3\}$;  $\mathcal{N}(B_1) = \{B_2, B_0, B_4\}$;
$\mathcal{N}(B_2)= \{B_3, B_1, B_5\}$;
$\mathcal{N}(B_3)=\{B_4, B_2, B_0\}$;
$\mathcal{N}(B_4)=\{B_5, B_3, B_1\}$;
and $\mathcal{N}(B_5)=\{B_6, B_4, B_2\}$. $\Box$
\end{example}

The example above illustrates that  the functions  used to define the  neighbors ensure that  if $B_p$ is a neighbor of $B_q$, 
then $B_q$ is a neighbor of $B_p$. %We prove this in the theorem in Appendix~\ref{th: symFunc}. 
Since the neighbor relationship to buckets is symmetric, we can represent the set of buckets and their neighborhood in the form of a graph  $\mathcal{G}=(V, \mathcal{E})$  with the set $V = \{0, \ldots, n-1\}$ of vertices corresponding to the $n$ buckets. The set $\mathcal{E} = \{\langle B_p,B_q \rangle | B_p \in 
\mathcal{N}(B_q)\}$ are edges.
Note that the resulting graph $\mathcal{G}$ is $d$-regular graph.
Fig.~\ref{fig: single_cluster} 
illustrates the resulting $d$-regular graph for the buckets in Example \ref{ex: graphBucket}.

%==========ALGORITHM FOR data encryption
\LinesNotNumbered \begin{algorithm}[!t]

\small
% \footnotesize
%=========Inputs=========Inputs=========Inputs=========Inputs=========Inputs=========
\textbf{Inputs:} $\mathcal{B}$:Set of buckets $n$ generated using $\mathcal{BC}$. $d$: degree of the graph

\textbf{Outputs:} $\mathcal{B}$: the set of $n$ buckets; $\mathcal{M}$: A adjacency matrix of a $d$-regular (undirected) graph

\nl{\bf Function GC($\mathcal{B}, d$) }
\Begin{
\nl $\mathcal{M} \leftarrow 0$ \Comment{Initialise matrix to 0}

\nl {\For{$p \in [0, n-1]$}{
 \nl { \For{$j \in [1,d]$ \Comment{Assigning $d$ neighbors to $p^{th}$ bucket} }{ 
         \nl {$B_q \leftarrow F_j(B_p)$, $\mathcal{M}[p][q] \leftarrow 1 $} 
    }
    }
}
  \nl $\mathbf{return}\  \mathcal{B}, \mathcal{M}$
}
}
\caption{\textsc{Veil-O}:Graph Creation (GC).}
\label{alg: Veil-O-GC}
\end{algorithm}
\setlength{\textfloatsep}{0pt}
After constructing the $d$-regular graph for a given set of input buckets, \textsc{Veil-O} next
determines the maximum number of records that can be lent/borrowed, \textit{i}.\textit{e}, the weight of each edge denoted with $\delta$; the direction of the edges, \textit{i}.\textit{e}., which bucket lends/borrows; the labels, \textit{i}.\textit{e}., the records that will be shared by neighboring buckets.  We will illustrate this in the following steps. 

% illustrates a 3-regular graph constructed from 6 buckets, 
% listed along with their sizes: $B_1(20), B_2(12), B_3(14), B_4(16), B_5(18), B_6(12)$. 
% After creating the $d$-regular graph $\mathcal{G}$, for each edge in $\mathcal{G}$, $\textsc{Veil-O}$ determines 
%  its weight, label, and direction as discussed next. 
%2) \textbf{weights},
%or, the overlapping size; and 
%3) \textbf{directions} of edges.  %\textit{i}.\textit{e}., determining the ``lender (or ``provider'') and the ``borrower'' (or ``receiver) for each pair of neighboring buckets connected by an edge.
% We design \textsc{Veil-O} to build a $d$-regular graph.
%We will use the graph shown in  Fig.~\ref{fig: veil-OGraphMatrixExample} to help illustrate the steps. 
% Given a bucket $B_p$, and its neighborhood $\mathcal{N}(B_p)$, we will use $L_{max}^{\mathcal{N}(B_p)}$  to denote the number of records in the largest neighbor of $B_p.$
\noindent \textsc{\textbf{Step 2. Maximum Overlap Determination (MOD)}}. [Algorithm \ref{alg: Veil-O-MOD}]
This step finds an initial value (upper bound) of the overlapping size, \textit{i}.\textit{e}., the weight of the edges in the graph, denoted as $\delta$.
\textsc{Veil-O} determines the limiting value for $\delta$ as follows:

\begin{enumerate}
\item For each bucket $B_p \in \mathcal{B}$ such that $|B_p| = \ell_b$ (\textit{i}.\textit{e}., full buckets), 
it must be the case that $\delta \leq \ell_b-L_{\mathit{max}}^{\mathcal{N}(B_p)}$, where   
 $L_{max}^{\mathcal{N}(B_p)}$  denotes the number of records in the largest neighbor of $B_p.$ More specifically, $L_{max}^{\mathcal{N}(B_p)} = \texttt{MAX}_{B_q \in \mathcal{N}(B_p)} |B_q|$. 
To see this, suppose $B_q$ is the largest neighbor of $B_p$. Then it can only borrow no more than $\ell_b - |B_q|$ records and this limits $\delta$ from being larger.

\item It always holds that $\delta \leq \frac{\ell_b}{d}$, as in a $d$-regular graph, each node has $d$ edges, \textit{i}.\textit{e}., 
each bucket has $d$ neighbors, and each pair of neighboring buckets has to share an equal number of records.

\item If $L_{\mathit{min}}$ is the size of the smallest bucket in $\mathcal{B}$, then the overlapping size is at most $\delta\leq\lfloor\frac{\ell_b-L_{\mathit{min}}}{d}\rfloor$.
To construct a $d$-regular graph successfully, the smallest bucket must borrow records from each of its $d$ neighbors. Moreover, it must borrow an equal number of records from each of its neighbors.
% For the bucket $B_{smallest}$ with the smallest number of elements in $\mathcal{B},$
% \emph{The smallest bucket must receive values from all its neighbors}. To successfully build a $d$-regular graph, \textsc{Veil-O} ensures smaller buckets tend to receive more value, while larger buckets are more likely to provide values to their neighboring buckets. In this setting, the smallest bucket must receive values from all its neighbors, otherwise, it's highly likely that larger buckets may not find receivers amongst their neighbors, making the construction of a $d$-regular graph to fail. 
% If the smallest bucket size is $L_{\mathit{min}}$, the overlapping size is at most $\delta\leq\lfloor\frac{\ell_b-L_{\mathit{min}}}{d}\rfloor$.
\end{enumerate}

% The maximum number of values it may lend is 
% \noindent(i)\emph{ Buckets that are full (\textit{i}.\textit{e}., they already have $\ell_b$ KV pairs) must provide values to their neighbors}. 
% as they do not have space to receive values from their neighbors. To implement this, \textsc{Veil-O} examines all buckets which are full in $\mathcal{B}$. Suppose $B_p \in \mathcal{B}$ if a bucket which is full, $\textsc{Veil-O}$ determines their $d$ neighbors (\textit{i}.\textit{e}., the set $\mathcal{N}(B_p)$ of neighboring buckets). It then computes $L_{max}^{\mathcal{N}(B_i} = \texttt{MAX}_{B_i \in \mathcal{N}(B_p)} |B_i|$. The overlapping size $\delta$ must be less than or equal to $\ell_b-L_{\mathit{max}}^{\mathcal{N}(B_p)}$. 

}

The maximum possible overlapping size is then determined by taking the minimum of the values in (1)-(3).
%$\delta=\texttt{MIN}\{\ell_b-L_{\mathit{max}}^{\mathcal{N}}, \lfloor\frac{\ell_b}{d}\rfloor, \lfloor\frac{\ell_b-L_{\mathit{min}}}{d}\rfloor\}$.
Note that the overlapping size is not final and can be adjusted in Step 3.
%==========ALGORITHM FOR data encryption
\LinesNotNumbered \begin{algorithm}[!t]
% \footnotesize
%=========Inputs=========Inputs=========Inputs=========Inputs=========Inputs=========
\small
\textbf{Inputs:} $\mathcal{B}$: the set of $n$ buckets; $\mathcal{M}$:Undirected graph; $d$: degree of graph

\textbf{Outputs:} $\delta_{MAX}$: Maximum overlapping size

\nl{\bf Function MOD($\mathcal{B}, \mathcal{M}, d$) }
\Begin{

{\nl $B_{max}^{\mathcal{N}(B_p)}\leftarrow \mathit{find}\_\mathit{largest}\_\mathit{neighbor}\_\mathit{of}\_\mathit{full}\_\mathit{buckets}(\mathcal{B})$}

{\nl $L_{max}^{\mathcal{N}(B_p)}\leftarrow\mathit{compute}\_\mathit{bucket}\_\mathit{size}(B_{max}^{\mathcal{N}(B_p)})$}

{\nl $B_{\mathit{min}}\leftarrow \mathit{find}\_\mathit{smallest}\_\mathit{buckets}(\mathcal{B})$}

{\nl $L_{\mathit{min}} \leftarrow \mathit{compute}\_\mathit{bucket}\_\mathit{size}(B_{\mathit{min}})$}

\nl $\delta_1 \leftarrow \ell_b - L_{max}^{\mathcal{N}(B_p)}$

\nl $\delta_2 \leftarrow \lfloor \frac{\ell_b}{d}\rfloor$

\nl
$\delta_3\leftarrow \lfloor \frac{\ell_b-L_{\mathit{min}}}{d}\rfloor$

\nl  $\mathbf{return}$ $\textsf{MIN}\{\delta_1, \delta_2, \delta_3\}$

}

\caption{\textsc{Veil-O}:Maximum Overlap Size Determination (MOD).}
\label{alg: Veil-O-MOD}
\end{algorithm}
\setlength{\textfloatsep}{0pt}

\begin{example}
\label{ex: MOD}
    In Fig.~\ref{fig: single_cluster}, the overlapping sizes $\delta$ determined using the three rules are 8, 6, and 2, respectively, the initial value of $\delta =2$. 
\end{example}

%==========ALGORITHM FOR data encryption
\LinesNotNumbered \begin{algorithm}[!t]
\small
% \footnotesize
%=========Inputs=========Inputs=========Inputs=========Inputs=========Inputs=========
\textbf{Inputs:} $\mathcal{B}$: A set of $n$ buckets, $d$: degree of graph; $\mathcal{M}$:Adjacency matrix of the $d$-regular graph;  $\delta_{\mathit{MAX}}$: maximum overlapping size

\textbf{Outputs:} $\mathcal{B}$; Set of $n$ buckets; $\overline{\mathcal{M}}$: An adjacency matrix encoding the direction of each edge in $\mathcal{M}$

\nl{\bf Function EDD($\mathcal{B}, \mathcal{M}, \delta_{MAX}$) }
\Begin{

\nl $\mathit{sort}(\mathcal{B})$ \Comment{Sort buckets by their sizes in increasing order}

{\nl $\delta \leftarrow \delta_{MAX}$}

{\nl $\overline{\mathcal{M}} \leftarrow []$ }

\nl \For{$B_j \in \mathcal{B}$} {
    \nl $L_j \leftarrow \mathit{compute}\_\mathit{bucket}\_\mathit{size}(B_j)$ 

    \nl $\mathcal{N}(B_j)\leftarrow \mathit{find}\_\mathit{neighbors}(B_j, \mathcal{M})$

    \nl $\mathit{sort}(\mathcal{N}(B_j))$ \Comment{In decreasing order by sizes}

    \nl \For {$B_p \in \mathcal{N}(B_j)$}{

        \nl \If {$\overline{\mathcal{M}}[p][j] != 1$ and $\overline{\mathcal{M}}[j][p] != 1$} {

            \nl \If {$L_j + \delta < \ell_b $} {

            \nl $\overline{\mathcal{M}}[p][j] = 1$ \Comment{Assign an edge between $B_p$ and $B_j$}
            
            \nl $L_j \leftarrow L_j + \delta$
        } 

    \nl \Else{ 
    \nl $\overline{\mathcal{M}}[j][p] = 1$  \Comment{Assign an edge between $B_j$ and $B_p$}

    \nl $L_{p} \leftarrow \mathit{compute}\_\mathit{bucket}\_\mathit{size}(B_p)$

        \nl \If{$L_{p} + \delta > \ell_b$}{
                $\mathit{reduce}$ $\delta$
                }
            }
        }

        }
    }
    {\nl$\mathbf{return} \mathcal{B}, \overline{\mathcal{M}}, \delta$}
 }
\caption{$\textsc{Veil-O}$:Edge~Direction~ Determination (EDD)}
\label{alg: Veil-O-EDD}
\end{algorithm}
\setlength{\textfloatsep}{0pt}

%==========ALGORITHM FOR data encryption

% \sharad{
% Padding: input :  d, Buckets;  output is well formed bucket with overlap.

%   G(V,E) undirected (stored as M) <-  create_d-regular graph(B, d) ;
%    delta    <--        Compute_max_overlap(M);
%   G(V,E^(arrow)) directed <-- Determine_edge_dir(G, delta)
%     B0..B_{(n-1)} <-- Create_well_formed_bucket(G^arrow)

%     create_d-regular_graph(B,d)  
%         use F (f1..fd)functions to generate neighbors and use the neighbor to genersate G.
%           For each bucket B 
%              for i = 1.. d
%                   N(B) <- N(B) union F_i(B)
%             M <--- Create Graph  G(V,E) 
%             /*G(V,E) is represented as a adjacency matrix*/
%     Compute Max_Overlap( M ) {
%     what is the function.
%     }
% }

{\color{black}\noindent\textsc{\textbf{Step 3. Edge Direction Determination (EDD).}} [Algorithm \ref{alg: Veil-O-EDD}]
To determine the directions of edges, \textsc{Veil-O} starts from the bucket in $\mathcal{B}$ with the smallest size and proceeds in increasing order of their sizes. This is because buckets with fewer records are less likely to get full, and thus should borrow as many records as possible from their neighbors.
For each bucket that is not full, it must borrow records starting from its neighbor with the maximum size and continue doing so from other neighbors in decreasing order of bucket size. This is to maximize the overlapping size $\delta$, and buckets with more records are more likely to become full to restrict $\delta$, and they should lend values whenever possible. The algorithm creates an adjacency matrix $\overline{\mathcal{M}}$ containing the directions for each edge in $\mathcal{M}$.

\begin{example}
   Continuing Example \ref{ex: MOD}  ($\delta = 2$, $\ell_b = 20$, $d=3$), the Algorithm EDD 
   generates $\overline{\mathcal{M}}$ as shown in Fig.~\ref{fig: single_cluster}.
   For instance,  for bucket $B_1$ with neighbors $B_0$ and $B_2$, 
   $\overline{\mathcal{M}}$  shows edges from both $B_0$ and $B_2$ to  $B_1$. Thus, $B_1$ 
   borrows  $\delta$ (which is 2)  records from each of these buckets. 
\end{example}

\noindent \textbf{Step 4: Adding fake values.}
Once edge directions 
have been determined, for each bucket $B_p$, we add appropriate 
fake records to ensure that all buckets 
are equisized. 
Let $L_p^{home}$ refer to the number of records in $B_p$ when it was
created by $\mathcal{BC}$.
For each neighbor
$B_q \in \mathcal{N}(B_p)$ such that there
is an incoming edge from $B_q$ to $B_p$, 
the bucket $B_p$ borrows $\delta$
records. Let there be $L_p^{in}$ such neighbors. To ensure that bucket $B_p$ has $\ell_b$ records associated, \textsc{Veil-O}
adds $L_p^{\mathit{fake}} = \ell_b - (L_p^{\mathit{home}} + \delta L_p^{\mathit{in}} )$ fake records to $B_p$. The resulting set of buckets with fake values added to them is denoted with $\mathcal{B}^{\mathit{padded}}$.

\begin{example}
Based on $\overline{\mathcal{M}}$ as shown in Fig.~\ref{fig: single_cluster}, $B_1$ adds 20 - (12 + 2$\times$2) = 4 fake records. 
Likewise, $B_3$ adds 2,  and $B_5$ adds 4 fake records. Thus, 12 fake records are added in total. Contrast this with the disjoint strategy that would have required 32 fake records to be added in total.
\end{example}

\noindent\textbf{Step 5: Label Creation (LC).}
For each edge $(v_p, v_q)$ in the graph $\mathcal{G}$, we next determine exactly which records are borrowed/lent between neighbors.
Let $B_q = F_j(B_p)$ be  a neighbor of $B_p$
such that $B_p$ borrows from $B_q$.  Furthermore, let $F_k(F_i(B_p)) = B_p$ (\textit{i}.\textit{e}., 
$B_p$ is the $k^{th}$ neighbor of $B_q$).
Consider
$B_q$ that contains 
$L_q^{\mathit{home}} + L_q^{\mathit{fake}}$ records. Henceforth, 
we will consider $B_q$ to be an ordered set of 
$L_q^{\mathit{home}} + L_q^{\mathit{fake}}$ records
and denote it as $B_q^<$  \footnote{The exact ordering does not matter, it could be any random ordering}. When there is no ambiguity we will
still continue to refer to it as  $B_q$ for notational simplicity.
We will denote the $i^{th}$ record in a bucket $B_p^<$ as $B_p[i]$.
Thus, when the bucket 
$B_p$
  borrows $\delta$ records from a bucket $B_q$, it does so based
  on the ordering of the records in $B_q$. In particular, 
  it borrows records
$B_q[(k-1) \delta  + 1]$, $B_q[(k-1) \delta  + 2]$, $\ldots$, $B_q[(k-1) \delta  +  \delta]$ which are added to the $label[(v_q,v_p)]$.
Note that the above strategy ensures that
no records of bucket $B_q$ are in the label for   more than one neighboring buckets, thereby ensuring that intersection of any three neighboring buckets are always empty. 

\begin{example}
Consider bucket $B_3$ in Fig.~\ref{fig: d_regular_graph_overlapping_buckets}
where $B_3$ borrows from each of its
neighbors $B_0, B_2$ and $B_4$ with 
$\delta = 2$. Since $B_3= F_3(B_0)$,
$B_3 = F_1(B_2)$ and $B_3 = F_2(B_4)$,
it will borrow the  $B_0[5]$ and
$B_0[6]$ from $B_0$. Likewise it will borrow $B_2[1]$ and
$B_2[2]$ from $B_2$. Finally, it  will
borrow  $B_4[3]$ and
$B_4[4]$ from $B_4$. Thus, labels for
$\mathcal{L}[(3,0)]= $ 
$\mathcal{L}[(0,3)]= $  
$\{B_0[5], B_0[6]\}$; 
$\mathcal{L}[(3,2)]= $ 
$\mathcal{L}[(2,3)]= $  
$\{B_2[1], B_2[2]\}$; 
$\mathcal{L}[(3,4)]= $ 
$\mathcal{L}[(4,3)]= $  
$\{B_4[3], B_4[4]\}$; 
$\Box$ \end{example}

\LinesNotNumbered \begin{algorithm}[!t]
\small
\textbf{Input:} $\mathcal{B}_f$: A set of buckets representing the vertices of a $d$-regular graph; 

$\overline{\mathcal{M}}$: Adjacency matrix with $dir$ for each edge in the graph; 

$\delta:$ The weight of each edge in the graph

\textbf{Output:} $\mathcal{L}$: a map of labels for each edge in the graph.
{\nl $\mathcal{L} \leftarrow \varnothing$}
   
{\nl \For{$B_p \in \mathcal{B}_f$} { 

    {\nl \For{$k \in [1, \ldots, n]$}{

    {\nl \If{$\overline{\mathcal{M}}[p][k]$ is $1$}{
        
     %{\nl $q \leftarrow  F_i(p)$}
     
    {\nl $\mathit{labels} \leftarrow \varnothing $}
    
    {\nl \For{$\ell \in [1, \delta]$}{
        
         {\nl $\mathit{labels} \leftarrow \mathit{labels} \cup B_p^<[(k-1) \delta + \ell]$}
         }
    }

    {\nl $\mathcal{L}[(p,k)] \leftarrow \mathcal{L}[(p,k] \cup \mathit{labels}$}
    
    {\nl $\mathcal{L}[(k,p)] \leftarrow \mathcal{L}[k,p] \cup \mathit{labels}$}
    }
    }
 }
 }
 %$\mathcal{L} \leftarrow labels_p$
}
}
\caption{Label Creation (LC)}
\end{algorithm}

\noindent\textbf{Step 6: Well-Formed Bucket Creation (WFBC).}
Once labels have been generated for every two neighboring buckets, the 
well-formed overlapping buckets are
generated as follows:
for each bucket $B_p \in \mathcal{B}^f$,
for each of its neighbor $B_q \in 
    \mathcal{N}(B_p)$, if $dir(B_q,B_p)$   is 1 (\textit{i}.\textit{e}., $B_p$ borrows from $B_q$), we add  to $B_p$ all records in 
    $\mathcal{L}[(q,p)]$ which correspond to the the
    set of records $B_p$ borrows from $B_q$.
    Note that for buckets $B_p$ and
    $B_q$, if they are neighbors, 
    they will contain $\delta$ common records.

\subsection{\textsc{Veil-O}: Outsourcing and Querying}

Outsourcing and querying in \textsc{Veil-O} is identical to
that in \textsc{Veil}. First, for
all real and fake tuples in $\mathcal{B}^f$  are shuffled,
encrypted and outsourced as set of encrypted records along with the RID.
Then to
create a multimap index, 
we construct for each bucket $B_p$ 
a map $Mmap[B_p]$ consisting of RID
of the records in $B_p$. As before, 
the encrypted record set, along with the $Mmap[B_p]$  for each 
 bucket $B_p$ is outsourced to the server.

To execute a query for a keyword $k$, first, the keyword
is mapped to the  $f$ buckets as in \textsc{Veil}. 
For each such bucket $B_p$, the server retrieves the
RID of every record in the multimap index $Mmap[B_p]$ and uses these RIDs to
retrieve data from the encrypted record store which are then returned
to the client.

%\textcolor{red}{
%Note that in the modified \textsc{Veil-O multimap} strategy, the client does not need to 
%store any additional metadata as needed in the
%\textsc{Veil-O} strategy and neither does
%the query-length change -- it simply consists of 
%$f$ tokens corresponding to the buckets to be retrieved. 
%Nonetheless, 
Note that unlike \textsc{Veil}, in \textsc{Veil-o} 
 the same RID may appear in more than
 one $MMap[B_p]$ since \textsc{Veil-o} 
 allows for overlap between buckets. 
%while the strategy above based on using multimaps reduces the number of fake tuples by allowing buckets to overlap, nonetheless, RIDs of the same record are replicated in multiple maps at the server
%side.
Since we restrict a RID to be replicated in atmost two buckets, the additional number of RIDs replicated equals $\frac{d \delta  |\mathcal{B}|}{2}$. Given a pointer is four bytes, the number of bytes of overhead can be computed as  $ 2d \delta |\mathcal{B}|$ bytes for a savings of 
$\frac{d \delta  |\mathcal{B}|}{2}$ fake records. 
Such an overhead remains significantly small even when
record sizes are relatively small, but are much more pronounced
when individual records can be large in which case, the overhead
remains a small fraction of the savings.
%which in comparison to the 
%overhead of storing fake tuples is very small. 

%as compared to the client side additional storage of $d |\mathcal{B}| $ of \textsc{Veil-O} (an overhead of $16 \Delta$). The primary benefit
%of \textsc{Veil-O multimap} is the reduced
%size of querry label which is the same as that
%in the original \textsc{Veil}. 

\subsection{Discussion}\label{sec:overlapping_discussion}

%%% veil-o is secure under VSR. It reduces SA overhead. However, Veil-o while it is data independnet in
% choosing neighbors, chooses the maximum things to overlap between  buckets in the data dependent manner
%to optimize the overlap.  Unlike Veil-o, This can result in two databases with the saame nunmber size + max size to be differentiated....   As a result the Veil-o, while it ensures VSR security, 
%the security it offers is weaker than \textsc{Veil} (or xor or moti). We can, howerver, make veil-o follow a similar model
%by instead making overlap a input variable (and thus independent of the data sst),. Of course, we will 
%say user specifies M1, max is M2... then lay out the logic of the approach from the mail. it will increase
% stash

%We refer to the above stragegy as Veil-O' and in the experiments we add a new experiment to show  how it affects
%stash.. Note that since in veil-o' more things in stash, they removed from server, so the SA slightly reduces
%(exactly by the same number to stash). 

\textsc{Veil-O}, like \textsc{Veil},  remains secure under VSR, as shown in Appendix \S\ref{th: veilo}.% Due to the page limitations, we remove the proof of correctness of \textsc{Veil-O}.  %as shown in \S Appendix \ref{th: veilo}.
Thus, adversary, cannot determine which key the user is accessing based on how queries are processed.
However, in \textsc{Veil-O}, since the maximum overlap between buckets is data dependent, adversary, could
distinguish between databases outsourced based on observed overlap (and, thus, the effective \textit{SA} achieved) by \textsc{Veil-O}. A slight modification can,  however make 
\textsc{Veil-O} achieve 
indistinguishability. In the modified
  version, the desired overlap between buckets is specified by the user (independent of the database) - let us refer to it 
  as $O_{\mathit{desired}}$. \textsc{Veil-O} learns the maximum overlap that can be supported (given random neighbor assignment) as in the 
  original protocol - let us denote it by $O_{\mathit{max}}$. If $O_{\mathit{desired}} \leq O_{\mathit{max}}$ we simply revert back to $O_{\mathit{desired}}$ and continue
  with the rest of the \textsc{Veil-O} to generate buckets. Alternatively, if $O_{\mathit{desired}} > O_{\mathit{max}}$, then  for each pair of 
  neighbor buckets $B_i$ and $B_j$ such that $B_j$ is a receiver and $B_i$ a lender of records  we perform the following check.
  Let $f_{B_j}$ be the number of RIDs pointing to fake  records in $B_j$. We first replace those fake records by 
  RIDs to additional records from $B_i$, thereby increasing the amount of overlap between $B_i$ and $B_j$.
  If $O_{\mathit{desired}} - O_{\mathit{max}}$ is  greater than then number of fake records $f_{B_j}$, to ensure  overlap of $O_{\mathit{desired}}$, we shift
  $ (O_{\mathit{desired}} - O_{\mathit{max}}) - f_{B_j}$ real records from $B_j$ to the stash, thereby creating space for equivalent number of RIDs to be 
  borrowed from $B_i$ to $B_j$. Note that in the modified \textsc{Veil-O} strategy, the number of overlapping records between any two buckets
  are equal to $O_{\mathit{desired}}$ and independent of the database being indexed. As a result, 
 the adversary cannot gain information about the database being indexed from the data representation, query access patterns, and
 volume in addition to being unable to learn query keywords. We refer to the modified version of \textsc{Veil-O} as 
 \textsc{Veil-O}$^\prime$. Note that \textsc{Veil-O}$^\prime$ may have a higher stash compared to \textsc{Veil-O} but the effective \textit{SA} it achieves
 could be even better compared to \textsc{Veil-O} since it could reduce number of fake records stored on the server side.
 In the experiment section, we will study impact of the above modified strategy (to make \textsc{Veil-O} secure based on the security model
 used in~\cite{XorMM, Moti_dprfMM}) on increase in stash.
 
%In particular, given two database $\mathcal{D}_1$, $\mathcal{D}_2$ with same number of records and $L_{max}$,  it is possible that
% the maximum overlap determined by \textsc{Veil-O} cfor $D_1$ annot be achieved irrespective of any neighbor assignment %for
% database $\mathcal{D}_2$, enabling the adversary to distinguish between $\mathcal{D}_1$ and $\mathcal{D}_2$. \textsc{Veil-O} can, however, be 
 % modified to offer indistinguishability (and not just VSR) by adding a simple post-processing step.

% \subsection{An Advanced Overlapping Strategy}
% {\color{red}
% We developed \textsc{Veil-O}$^{\prime}$ that allows users to define the desired overlapping size. This user-defined overlapping size is now a mandatory parameter when constructing the $d$-regular graph, providing users with more flexibility and control, regardless of the distribution of the dataset.
    % \textsc{Veil-O}$^{\prime}$ utilizes the same $d$-regular graph creation method as \textsc{Veil-O}, with the primary difference being the approach to overlapping size determination and the borrowing / lending process. Specifically, each smaller bucket is required to borrow records from neighboring buckets. However, in case where there is insufficient space to accommodate the incoming records, the existing records within the bucket must be relocated to the local stash to create the necessary room to allow the bucket to accept new records to maintain the user-specified overlapping size. }

\section{Supporting Dynamic Changes}\label{sec:dynamic}
So far, similar to \cite{Moti_dprfMM, XorMM}, we have   discussed \textsc{Veil} and \textsc{Veil-O} under static setting where database to be outsources is pre-known. While a full development and evaluation of dynamic operations in \textsc{Veil} is outside the scope of this paper, 
we briefly discuss how \textsc{Veil} can be extended to support dynamic operations. Indeed, the flexibility and
ease of \textsc{Veil} in supporting dynamic operations compared to other prior work is one of its advantages
as will be discussed in \S\ref{sec: related_work}.
 We focus on the case of insertion though discussion below which can be extended to updates and deletions as well.
In  \textsc{Veil}, adding new data does not require re-execution of the bucktization strategy as long as the value of $L_{max}$ after insertions 
does not exceed $L_{max} \times \mathit{QA}$. Insertion can be supported by 
retrieving the set of records in the $f$ buckets corresponding to the
key for the  record being inserted, replacing one of the fake tuples (if present) in the buckets 
by the newly inserted tuples, 
re-encrypting the records in the buckets, and re-outsourcing that modified buckets. 
Of course, if the $f$ buckets  do not contain  any fake tuples (and hence
have no residual capacity to store more data), the newly inserted tuple is stored in a stash, as would have
been the case had the buckets become full in the static situation when constructing the original buckets.
The above insertion strategy 
would continue to work, as long as, the new value of $L_{max}$ after insertions, say $L_{max}'$, 
remains below $QA \times L_{max}$ \footnote{ In fact, it can continue to work beyond that except that all new records 
after the point of increase in $L_{max}'$ beyond $QA \times L_{max}$ will need to be stored locally in the stash.}.
As $L_{max}$ increases, since we are not changing the fanout $f$ or the 
bucket size, effectively the $QA$ value reduces which, in turn, increases the probability of a record to have to be stored
in the stash. When  $L_{max}'$ goes above  $QA \times L_{max}$, we can either
continue to map new records to stash (increasing client side overhead) and/or reorganize the data. 
While we do not conduct a formal analysis or experimental validation, we believe that the strategy 
above will allow a large number of insertions between reorganizations since  $L_{max}$ can be expected to grow slowly. Furthermore, an user can begin with a larger value of $QA$ to control how often reorganization is required.
}

\section{Related Work}
\label{sec: related_work}

{\color{black}
Volume hiding (VH) as a security goal,  has  gained attention starting with  the seminal
work by Kamara \emph{et al} in ~\cite{kamara2018suppression,kamaraVLH} that 
utilizes a multimap data structure for encrypted keyword search.  VH is easier
to achieve when access pattern (AP) hiding technque, such as ORAM~\cite{goldreich1987ORAM,goldreich1996software,ostrovsky1990efficient,stefanov2018path}, is already being 
used. Given 
prohibitive nature of AP hiding, recent work ~\cite{kamara2018suppression,kamara_dynamic,kamaraVLH,Moti_dprfMM,Moti21_dynamic,feifei_hybridx}
has explored 
VH
without requiring ORAM. \textsc{Veil} falls into this  category of work (discussed in 
details in \S Appendix~\ref{sec:full_related_work}).
Below,  we focus  on two techniques most relevant to \textsc{Veil}:
\textsf{dprfMM}~\cite{Moti_dprfMM} that uses
a cuckoo hash based strategy, and \textsf{XorMM}~\cite{XorMM} that uses the xor filter to support
volume hiding in key-value stores. We compare \textsc{Veil} and its variants to these approaches
experimentally in the following section. In the remainder
of this section we focus on  a qualitative comparison of the schemes in terms of several criteria 
including support for dynamic updates, security offered, applicability as an indexing technology, and 
expected performance.

\noindent\textbf{Supporting  Dynamic Changes.}
Different approaches,  viz., \textsf{XorMM}, \textsf{dprfMM} and \textsc{Veil} 
offer different level of ease in extending them to support dynamic operations (insertions, updates, deletion).
In particular, \textsc{Veil} and \textsf{dprfMM} offer significant flexibility and ease to support dynamic operations compared to \textsf{XorMM} as shown in \S\ref{sec:dynamic}.
\textsf{DprfMM}, like \textsc{Veil}, also  offers flexibility and can allow 
periodic reorganization by storing new data in the cuckoo hash and/or stash.
In contrast, 
%While the study
%of insertion and dynamic maintenance of \textsc{Veil} is outside the scope of this work, we note that our approach is
%significantly more flexible and amenable for dynamic updates compared to prior work such as  %\cite{Moti_dprfMM} or
%\cite{XorMM}. For instance,  
\textsf{XorMM} cannot support insertions, as it requires 
the xor filter to be recontructed every time there is an insertion. This because the xor filters are constructed through a sequential process that creates dependencies between 
cell representations,  \textit{i}.\textit{e}., a record inserted later into the filter may depend upon the cipher representation of the record inserted earlier. As a result, 
when new data has to be inserted 
(whether or not it results in change to $L_{max}$), the filter has to be recomputed from scratch.
%we requireto recompute the filter from scratch.
%an insertion of a  record (whether or not it causes the 
%$L_{max}$ value to change) would require recomptuation of the entire  filter. 
The challenge of supporting dynamic operations, coupled with the 
flexibility to choose
$QA$ and $SA$ parameters to control the overheads, are the primary advantages of \textsc{Veil}.

%\noindent \textbf{Index} -- secondary versus primary index
\noindent\textbf{Security Offered. }
We developed \textsc{Veil} and its variants under the security goal VSR (\S\ref{sec:adv}) 
that roughly corresponds to ensuring that adversary cannot differentiate between query keywords.
One could also consider stronger security model that in addition also ensures indistinguishability 
amongst databases (\S\ref{sec:overlapping_discussion}) used in \cite{Moti_dprfMM,XorMM}.   We note that 
\textsc{Veil}  is already secure against the enhanced security model  used in \cite{Moti_dprfMM,XorMM},
that we refer to as 
{\em indistinguishability}, - we show this formally in \S Appendix \ref{sec:security_analysis_veil}. 
\textsc{Veil-O}, as discussed in \S\ref{sec:overlapping_discussion}, however, while it
ensures VSR, does not ensure indistinguishability directly. It can, however, be modified
with a simple post-processing step to offer the same level of security as  \cite{Moti_dprfMM,XorMM}, 
with the post processing step causing a slight increase in the stash storage. 

\noindent\textbf{Probabilistic versus Guaranteed Success.}
Given a database $\mathcal{D}$,  \textsf{dprfMM}  and \textsc{Veil} are always successful in forming the multimap representation on the server, XOR filters are probabilistic, and there is  a chance (however, small) 
that the technique may 
fail to create an appropriate xor filter~\cite{XorMM}.
Hence xor filters are not guaranteed to create a multimap to support volume-hiding in all
situations.

\noindent\textbf{Primary versus Secondary Index.} 
The basic \textsc{Veil} approach can outsource records in the form of 
buckets containing encrypted data records, or in the form of buckets with RIDs that point to 
encrypted record store. Likewise, \textsf{dprfMM} can also be extended to store either RIDs or
the records. In contrast,  \textsf{xorMM} can only be used to store records and cannot easily be
extended to store RIDs, since
the security requirement of xor filter requires 
data in the filter to be encrypted. Thus, storing RIDs in the filter will introduce
 additional round of communication between client and
server for the client to decrypt the RID and subsequently ask the server to retrieve the
records corresponding to those RIDs. This will significantly increase communication costs.
Also, the RIDs returned by the server may point to some record that does not exist in the data storage, and such information is revealed to the server, which directly leaks the volume.
\textsc{Veil-O}, in contrast,  provides benefit  (reduced \textit{SA}) when used as a secondary index
with RID that are referenced to retrieve data from the encrypted store. 

\noindent\textbf{Expected Performance Comparison.} 
While we conduct thorough experimental comparison between strategies, we make a few observations
about expected performance of different strategies.\\
\noindent
$\bullet$  In terms of QA and SA,   \textsf{xorMM} (if it can be successfully created) overshadows 
\textsf{dprfMM} since it offers fixed QA and SA or 1 and 1.23 as compared to 2 and 2.6. \textsc{Veil} 
supports a tunable QA and SA values - by appropriately choosing 
QA and SA, it can outperform  \textsf{xorMM}.

\noindent$\bullet$ In terms of client storage, all schemes store very little meta data (\textit{e}.\textit{g}., \textsc{Veil} needs
to store  just two numbers -  fanout and the number of buckets ) in addition to the stash containing
some overflow data records.  Stash is also stored in \textsf{dprfMM} which, as experiments will show, is higher than what is stored in \textsf{Veil}. \textsf{xorMM} does not maintain a client side
stash, but as a result, has a non-zero probability of failing to form. 
%A more complete performance  comparison of the approaches follow in \S\ref{sec: experiment}.

% \begin{table*}[h!]
% \small
% \begin{center}
% \begin{tabular}{|c|c|c|c|c|c|c|c|c|c|c|}
% \hline
%  \multirow{2}{4em}{} & \multicolumn{2}{|c|}{0.4} & \multicolumn{2}{|c|}{0.6} & \multicolumn{2}{|c|}{0.8} & \multicolumn{2}{|c|}{1.0} & \multicolumn{2}{|c|}{1.2} \\ \cline{2-11}
%  &  key \# & max key size  & key \# & max key size& key \# & max key size  & key \# & max key size & key \# & max key size\\\hline
% 6M & 200,000 &357&  200,000 & 5,234 & 81,153 & 131,749 & 199,919 & 370,760  & 2,154 & 1,329,059 \\\hline
% 12M & 400,000 & 456 & 399,999 & 10,430 & 89,905 & 263,499 & 399,806 & 711,060 & 2,148 & 2,658,118 \\\hline
% 18M & 600,000  & 527  & 599,996 & 15,619 & 93,022 & 395,182 &599,612  & 1,041,544 & 2,149 &3,987,177 \\\hline
% \end{tabular}
% \caption{Skewed TPC-H datasets}
% \label{tb: skewed_datasets}
% \end{center}
% \end{table*}

% \newpage
% \begin{table}[h!]
% \begin{center}
% \begin{tabular}{|c|c|c|c|c|c|}
% \hline
%   & 6M & 12M & 18M & 24M & 30M \\ \hline % & 36M\\ \hline
% key \#    & 200,000 & 400,000 & 600,000 & 800,000 & 1,000,000 \\ \hline %&0.058  \\ \hline
% max key size & 57 & 61 & 58 & 60 & 61\\ \hline %&0.058  \\ \hline
% \end{tabular}
% \caption{Non-skewed TPC-H datasets}
% \label{tb: non-skewed_datasets}
% \end{center}
% \end{table}
}

\section{Experimental Evaluation}\label{sec: experiment}

In this section, we evaluate the performance of \textsc{Veil}. We study:
\begin{itemize}[noitemsep,nolistsep,leftmargin=0.15in]
%$\bullet$
\item  Impact of  user-specified parameters (Storage Amplification $\mathit{SA}$, Query Amplification $\mathit{QA}$, and  fanout $f$) on    Stash Ratio $\mathit{SR}$. %Note that  we  expect $\mathit{SR}$ to decrease with the  increase in the bucket size and, furthermore, in the number of buckets. 
Based on Equation~\ref{eq: bucket_size_num},   we make the following observations:
\begin{enumerate}[noitemsep,nolistsep,leftmargin=0.2in]

\item For fixed $\mathit{QA}$ and $f$, 
as $\mathit{SA}$ increases, while the bucket size remains constant, the number of buckets increases linearly, 
and hence lower the expected SR (\textbf{Exp 2}).
% \noindent
\item For fixed $\mathit{SA}$ and $f$, as $\mathit{QA}$ increases,   
bucket size increases linearly, but the number of buckets decreases proportional to  $\frac{1}{QA}$. For fixed $\mathit{QA}$ and SA, increasing  $f$ causes 
 bucket size to decrease, but increases the number of buckets.  
 We  explore how $\mathit{SR}$ changes as a function of $\mathit{QA}$ and $f$ experimentally in \textbf{Exp 1} and \textbf{Exp 3}.
 
  \item Since~\cite{Moti_dprfMM} already achieves a $\mathit{QA}$ and $\mathit{SA}$ of 2 with a small stash  (viz., $\mathit{SR}$), we focus on $\mathit{QA}$ and $\mathit{SA}$ in the range 1 to 2.

   \item Since $\mathit{QA}$ influences the number of records retrieved from the server, 
 it is desirable to keep it as close to 1 as possible while still ensuring
a small $\mathit{SR}$ and an $\mathit{SA}$ below 2. This would be strictly better than the state-of-the-art.
\end{enumerate}
%\noindent$\bullet$ 
\item Effect of  overlapping strategy \textsc{Veil-O} on reducing $\mathit{SA}$ 
(\textbf{Exp 5}).

%\noindent$\bullet$ 
\item Comparison of \textsc{Veil} with existing approaches, including \textsf{dprfMM}~\cite{Moti_dprfMM} and \textsf{XorMM}~\cite{XorMM}, in terms of $\mathit{SA}$, $\mathit{QA}$, and 
$\mathit{SR}$ (\textbf{Exp 7}). % when data has different levels of skewness.(Experiment 8) .

%\noindent$\bullet$ 
\item 
 \textsc{Veil} and \textsc{Veil-O} performance on larger data sets. (\textbf{Exp 8} and \textbf{9}).

%\noindent$\bullet$ 
\item  Setup time and query time for \textsc{Veil} and \textsc{Veil-O} and compare them with those for state-of-art (Exp \textbf{10-12}).

\end{itemize}

\subsection{Setup}
We evaluate \textsc{Veil} using the LineItem table from the variant of TPC-H dataset entitled TPC-H-SKEW~\cite{crolotte2012introducing}. TPC-H-SKEW generates the same tables, except that it allows us to control the skewness of the data using
a  ``skew factor'', denoted by $z$. We show the impact of changing $z$ values in Table~\ref{tb: datasets} by listing the $L_{\mathit{max}}$ and the number of keys generated 
in the LineItem table for a scale factor 1 (which corresponds to 6M records in the LineItem  table). As shown in Table~\ref{tb: datasets}, 
as $z$ increases, the data gets more skewed with larger $L_{\mathit{max}}$. When 
$z$
is zero, the data generated is non-skewed  as in the original TPC-H dataset.
We vary $z$ from 0 to 1 and use a default of 0.4, {\blue{which is consistent with real-world datasets that are most likely to be skewed}}.
To generate key-value datasets, we use the Partkey (PK) column as the keys and use values of all other columns as their %and Suppkey (SK), using PK as the keys and SK as their
associated values. 
For most of our experiments, we use the TPC-H with scale factor 1 (\textit{i}.\textit{e}., the datasize is 6M). We also include the results with the scale factor of 6 to see how \textsc{Veil} scales to a larger
dataset. 
Our experiments were conducted on a MacBook Pro equipped with an M1 Pro processor and 32 GB of RAM, running the macOS Monterey operating system. We utilized SHA-256~\cite{sha-256} as the hash function and employed AES encryption (CBC mode)~\cite{daemen1999aes} for our symmetric encryption scheme.

%the same machine for both the client and the server, which is  

\begin{table}[h!]
\begin{center}
\small
\caption{TPC-H datasets with different skew factor $z$}
\label{tb: datasets}
\begin{tabular}{|c|c|c|c|c|c|c|c|}
\hline
 % & Fanout & Bucket size & $\mathit{QA}$ & $\mathit{SA}$ & Stash Ratio\\ \hline 
Skew Factor $z$ & 0  & 0.2 & 0.4 & 0.6 & 0.8 & 1 \\ \hline
$L_{\mathit{max}}$ & 57  & 63 & 357 & 5,234 & 131,749 & 370,760 \\ \hline
\# of keys &  200,000  &  200,000 & 200,000  & 200,000 & 81,153  & 199,
919 \\ \hline
\end{tabular}

% \BBB\BBB\BBB\BB
\end{center}
\end{table}
% select l_partkey, count(*) as count from lineitem group by l_partkey order by count desc;

\noindent \textbf{Comparison Metrics: }
We evaluated our approaches and the state-of-art approaches using the metrics defined in \S\ref{sec:intro}, including \textit{query amplification} (\textit{QA}), \textit{storage amplification} (\textit{SA}), and \textit{stash ratio} (\textit{SR}). To measure the physical storage overhead at the local side and at the server, we introduced the following two additional metrics:

\noindent \textit{Client Storage Amplification}, denoted by \textit{CSA}: that is the ratio of the size of
total physical storage required at the local side to implement a given volume hiding scheme over 
the size of encrypted representation  of the dataset. 
For instance, in \textsc{Veil}, clients need to store values of fanout and  number of buckets  (2 integers) in addition
to the records in the stash. Thus, the \textit{CSA} corresponds to the storage requirement of these divided by the 
encrypted representation of the data.

\noindent \textit{Server Storage Amplification}, denoted by \textit{SSA}: the ratio of the  size of physical storage required at the server side to implement a given volume hiding scheme to the the size of encrypted representation  of the dataset. 
For instance, in \textsc{Veil}, the server needs to store the encrypted record set as well as the multimap index \S\ref{sec:components of veil}). Thus, \textit{SSA} corresponds to the ratio of the total storage needed at the server by the size of the encrypted record set.

\subsection{Evaluation of Parameters on Stash Ratio}
We conducted an ablation study of three factors, --- the fanout $f$, the storage amplification $\mathit{SA}$,  the query amplification $\mathit{QA}$, and the skewness factor $z$ by fixing three others in each experiment to evaluate the impact of each factor on the overall performance of \textsc{Veil}.
%For each experiment, we generated 9 databases at a given skew factor $z$ using the TPC-H skew benchmark. %As the TPC-H skew generates deterministic databases, we varied $z$ slightly to create the 9 distinct databases, \textit{e}.\textit{g}., setting $z$ to 0.392, 0.394, 0.396, 0.398, 0.4, 0.403, 0.404, 0.406, and 0.408, respectively, where the average skew factor $z$ is exactly 0.4. 
For each generated database, we ran \textsc{Veil} 5 times to map keys to buckets randomly and computed an average stash ratio ($\mathit{SR}$). The outcomes of our experiments are illustrated in Fig.~\ref{fig: exp_stash_vs_fanout_fixed_QA_SA_9_datasets}, Fig.~\ref{fig: exp_stash_vs_sa_fixed_QA_F_9_datasets}, Fig.~\ref{fig: exp_stash_vs_qa_fixed_SA_F_9_datasets}, and Fig.~\ref{fig: exp_stash_vs_z}. We also recorded the average stash size, shown as the labels in Fig.~\ref{fig: exp_stash_vs_fanout_fixed_QA_SA_9_datasets}, Fig.~\ref{fig: exp_stash_vs_sa_fixed_QA_F_9_datasets}, Fig.~\ref{fig: exp_stash_vs_qa_fixed_SA_F_9_datasets}, and Fig.~\ref{fig: exp_stash_vs_z}.\footnote{** Numbers on top of the lines in the figures represent the average number of records in the stash. We will have similar numbers in all plots with Stash Ratio (Fig.~\ref{fig: exp_stash_vs_fanout_fixed_QA_SA_9_datasets}, Fig.~\ref{fig: exp_stash_vs_sa_fixed_QA_F_9_datasets}, Fig.~\ref{fig: exp_stash_vs_qa_fixed_SA_F_9_datasets}, Fig.~\ref{fig: exp_stash_vs_z}, and Fig.~\ref{fig:exp_stash_vs_desired_overlapping}). }

\begin{figure*}[htbp]
% \BBB\BB
    \centering
    \begin{minipage}{0.33\textwidth}
        \centering
        \includegraphics[width=\linewidth]{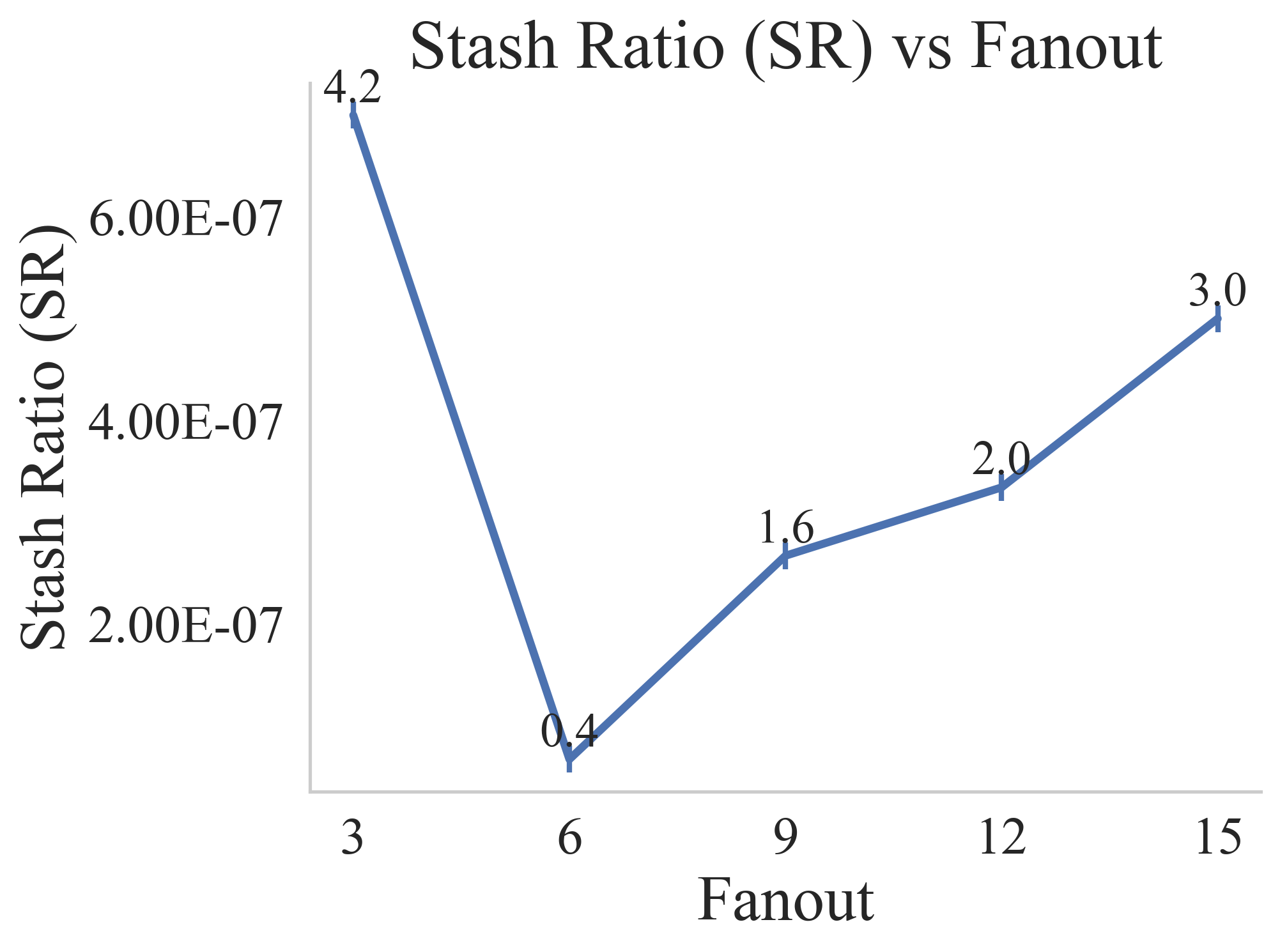}
        % \BBB\BBB\BBB\BB
        \caption{Impact of $f$ ($\mathit{SA}$=1.2, $\mathit{QA}$=1)**}\label{fig: exp_stash_vs_fanout_fixed_QA_SA_9_datasets}
    \end{minipage}%
    \hfill
    \begin{minipage}{0.33\textwidth}
        \centering
        \includegraphics[width=\linewidth]{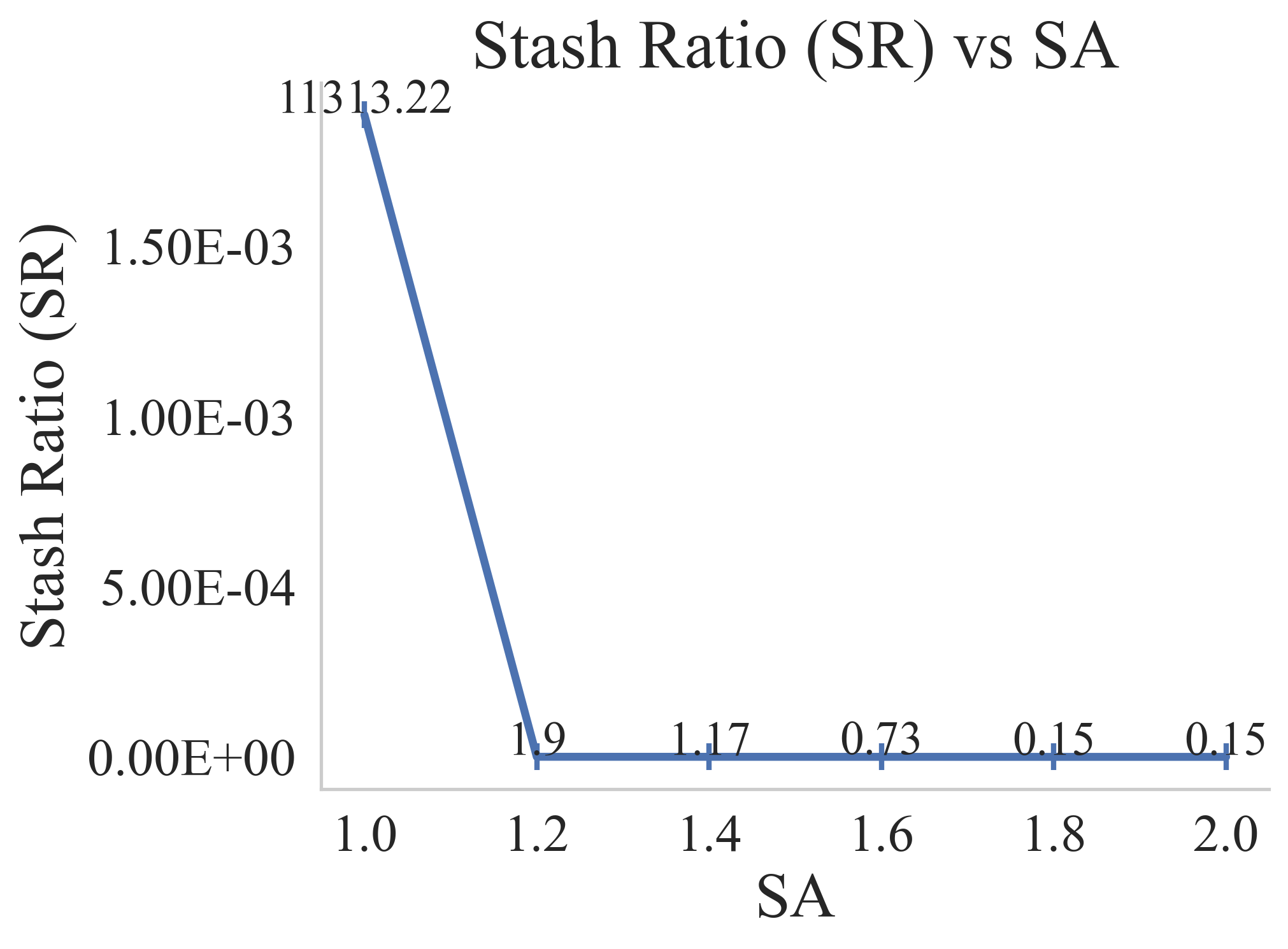}
        % \BBB\BBB\BBB\BB
        \caption{Impact of SA ($\mathit{QA}$=1, $f$=6)**}\label{fig: exp_stash_vs_sa_fixed_QA_F_9_datasets}
    \end{minipage}%
    \hfill
    \begin{minipage}{0.33\textwidth}
        \centering
        \includegraphics[width=\linewidth]{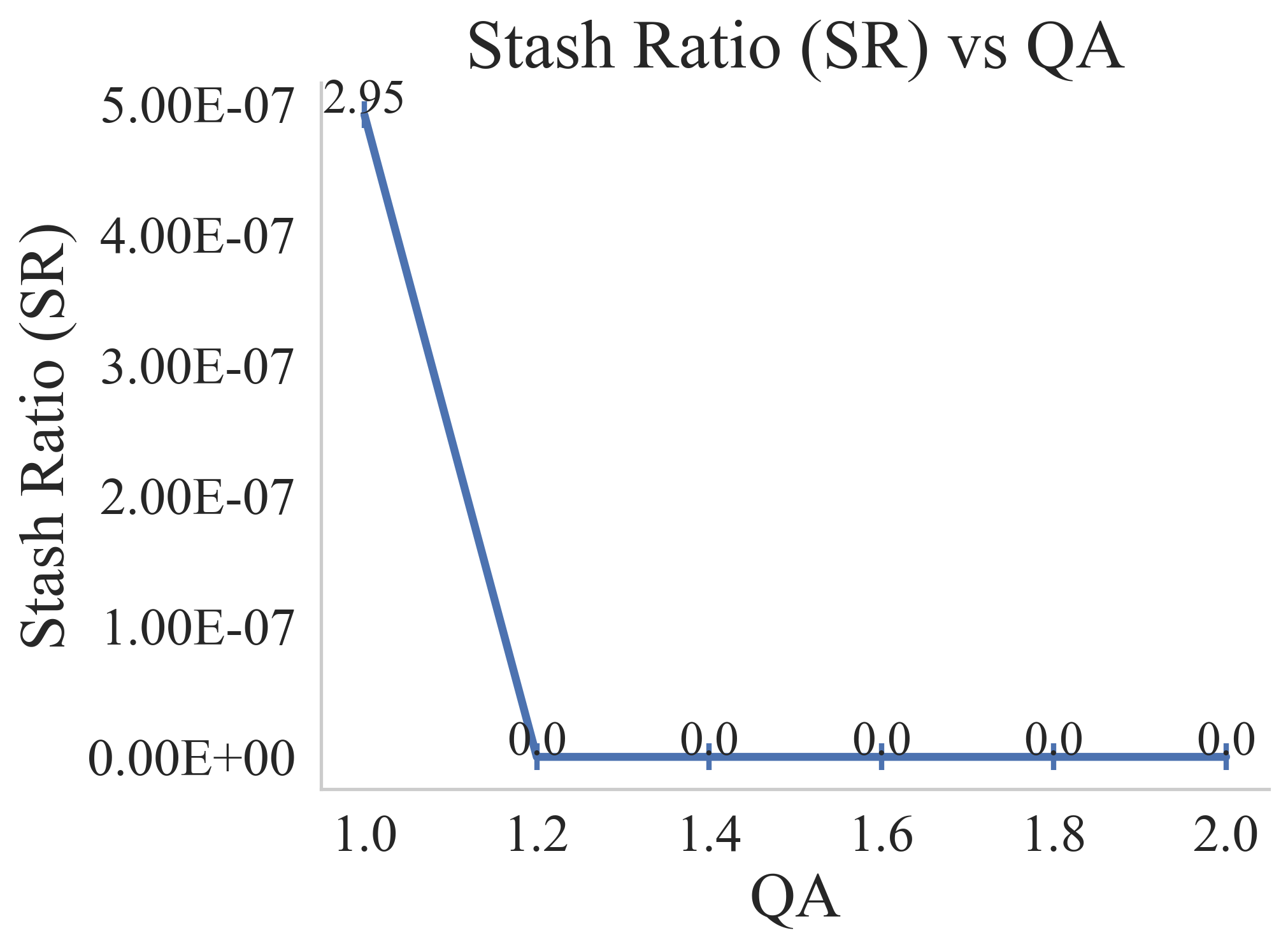}
        % \BBB\BBB\BBB\BB
        \caption{Impact of QA ($\mathit{SA}$=1.2, $f$=6)**}\label{fig: exp_stash_vs_qa_fixed_SA_F_9_datasets}
    \end{minipage}%
    % \BBB\BB
\end{figure*}

\begin{figure}[htbp]
    \centering
    % \begin{minipage}{0.33\textwidth}
    %     \centering
    %     \includegraphics[width=\linewidth]{exp_images/exp_stash_vs_qa_fixed_SA_F_9_datasets.png}
    %     \BB\BBB
    %     \caption{Impact of QA ($\mathit{SA}$=1.2, $f$=12)}\label{fig: exp_stash_vs_qa_fixed_SA_F_9_datasets}
    % \end{minipage}%
    \begin{minipage}{0.33\textwidth}
        \centering
        \includegraphics[width=\linewidth]{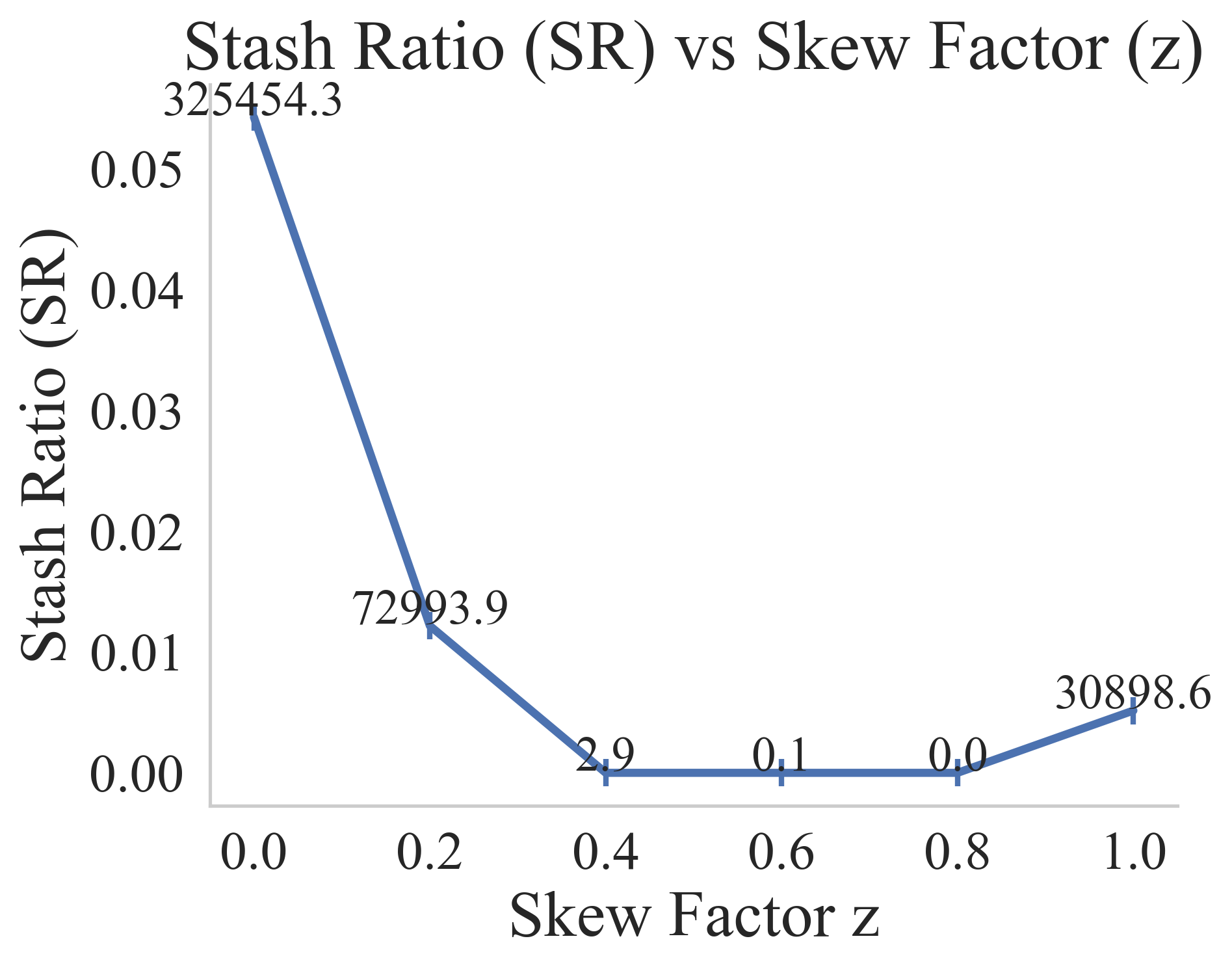}
        \caption{Impact of $z$ ($f$ = 6)**}\label{fig: exp_stash_vs_z}
        % \caption{**Impact of $z$ ($\mathit{SA}$=1.2, $\mathit{QA}$=1, $f$ = 6)}\label{fig: exp_stash_vs_z}
    \end{minipage}%
    % \hfill
    \begin{minipage}{0.33\textwidth}
        \centering
        \includegraphics[width=\linewidth]{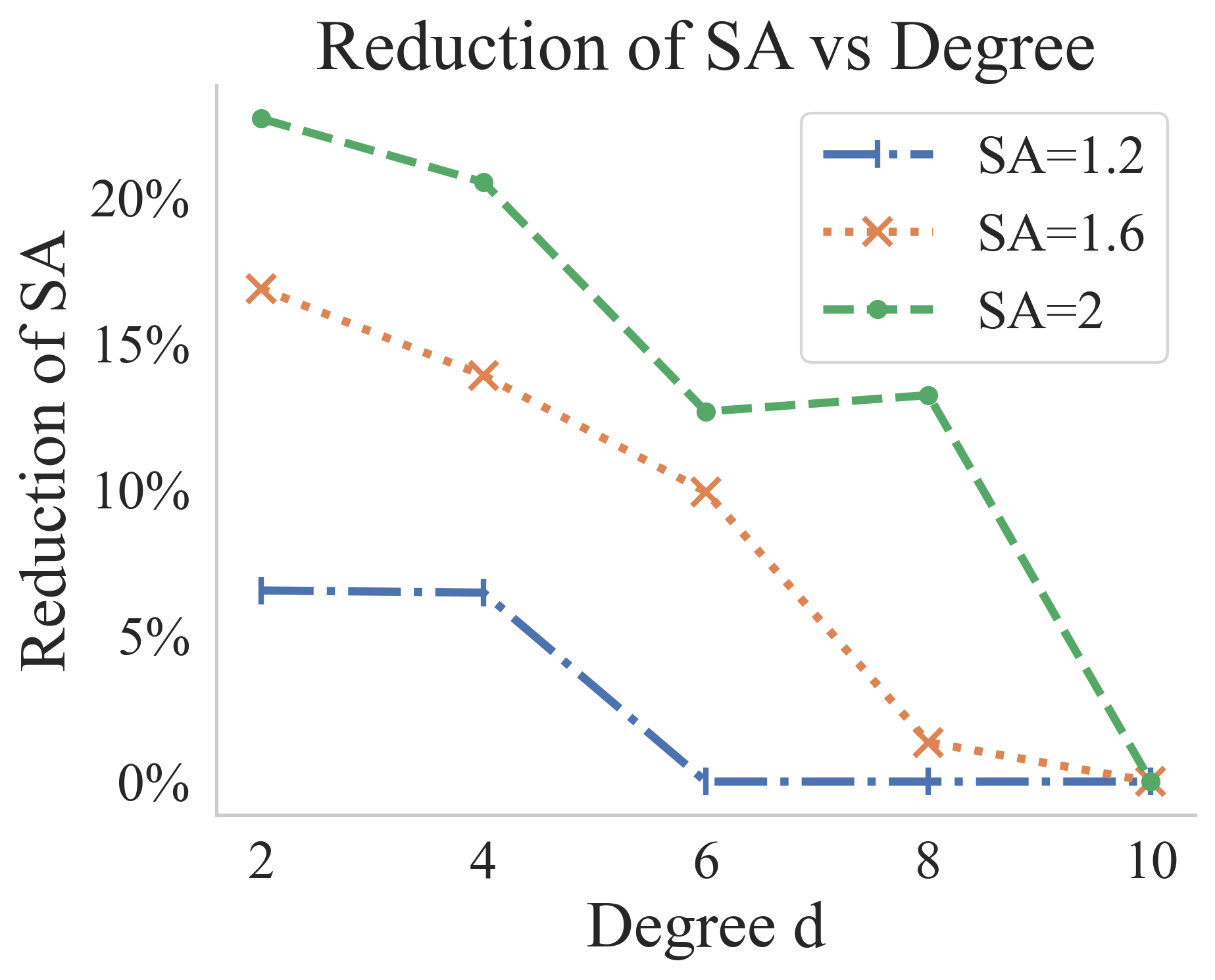}
        \caption{Impact of $d$ ($\mathit{QA}$=1, $f$ = 6)}\label{fig: exp_real_sa_vs_d}
    \end{minipage}%
    \begin{minipage}{0.33\textwidth}
        \centering
        \includegraphics[width=\linewidth]{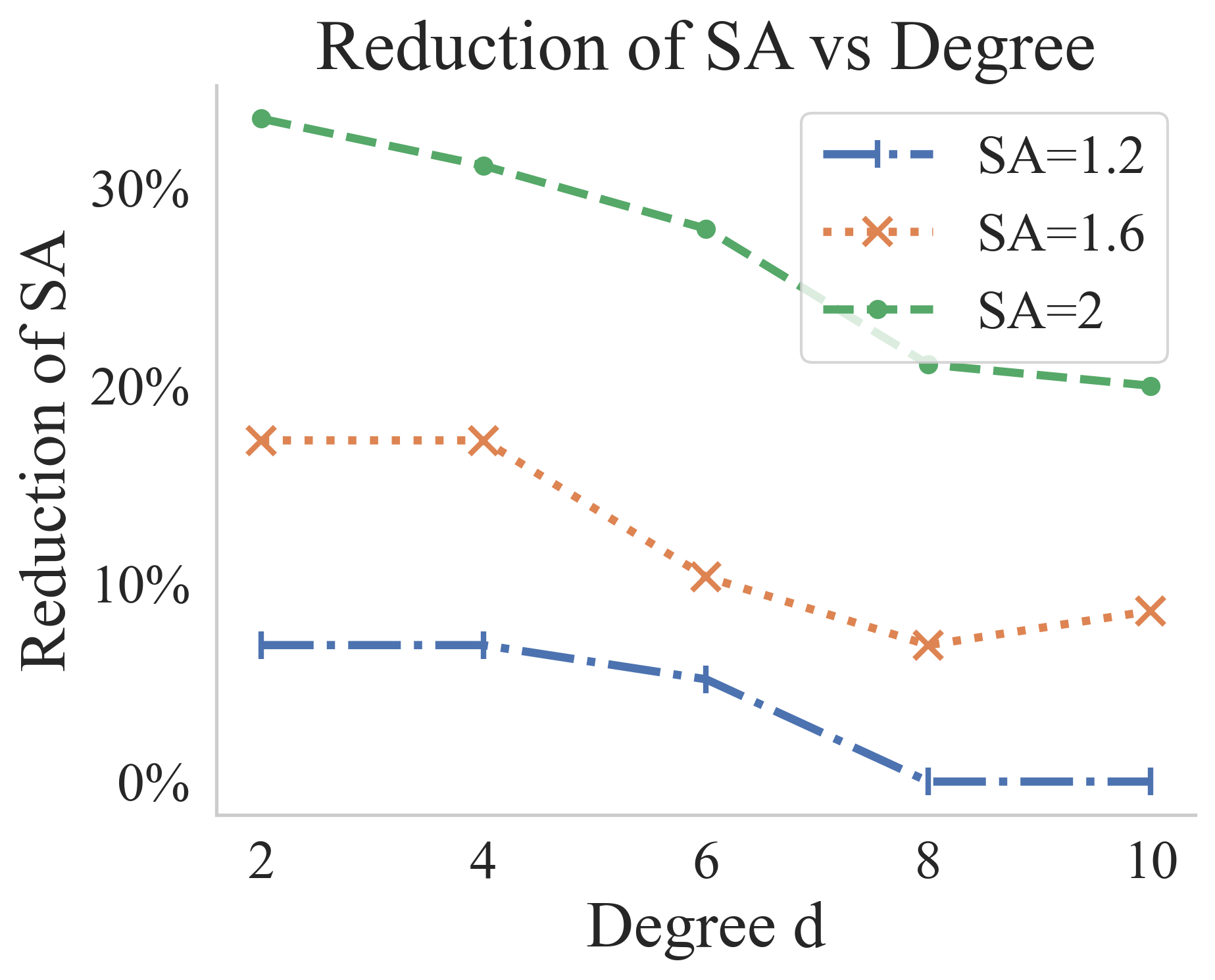}
        \caption{Impact of $d$ on 36M Data}\label{fig: exp_real_sa_vs_d_36M}
    \end{minipage}%
    % \vspace{1em}
\end{figure}

\begin{figure*}[htbp]
    \centering
    \begin{minipage}{0.33\textwidth}
        \centering
\includegraphics[width=0.98\linewidth]{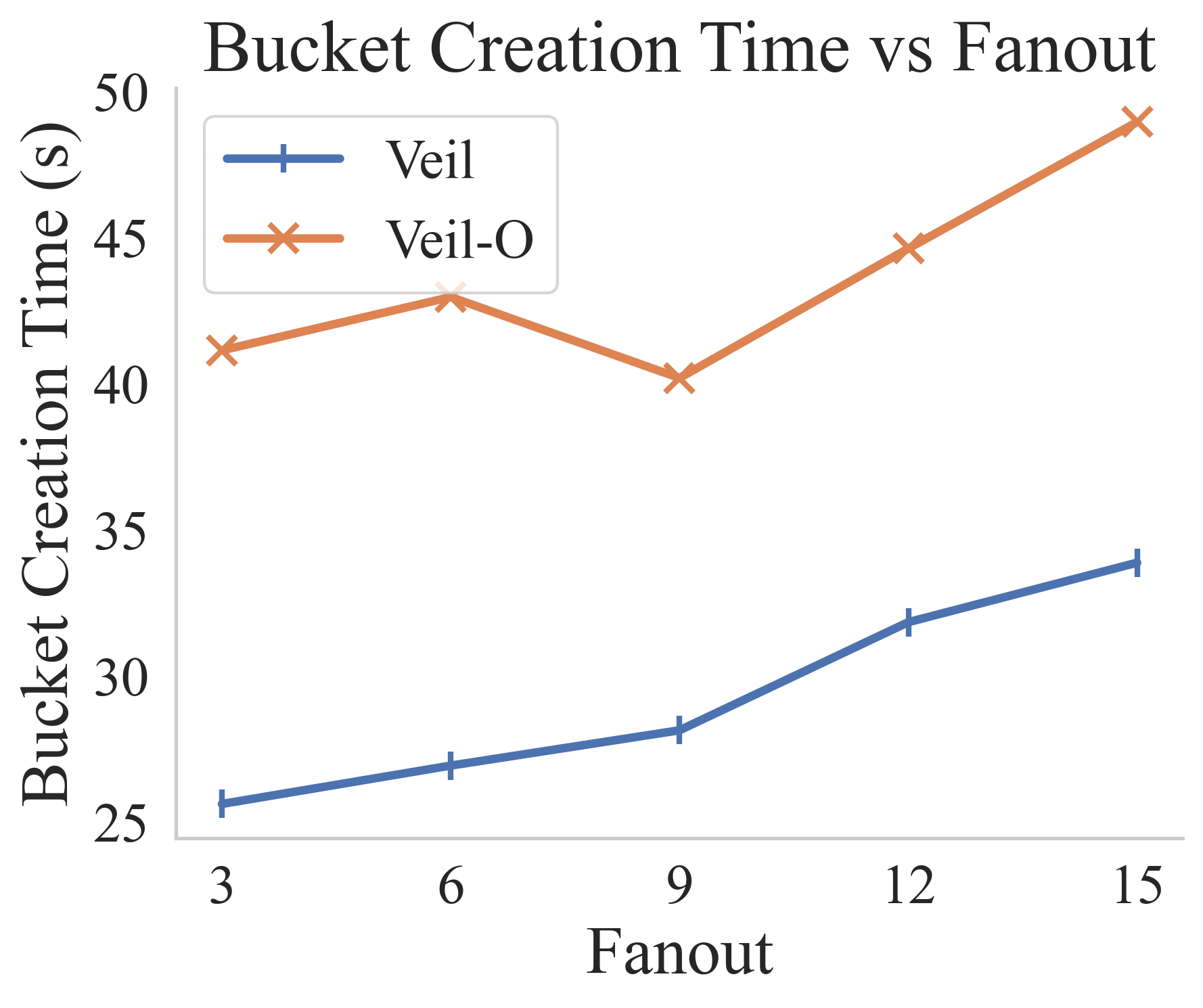}
        \caption{Bucket Creation Time vs $f$}\label{fig: exp_bucket_creation_time_vs_fanout}
    \end{minipage}%
    \begin{minipage}{0.33\textwidth}
        \centering
\includegraphics[width=\linewidth]{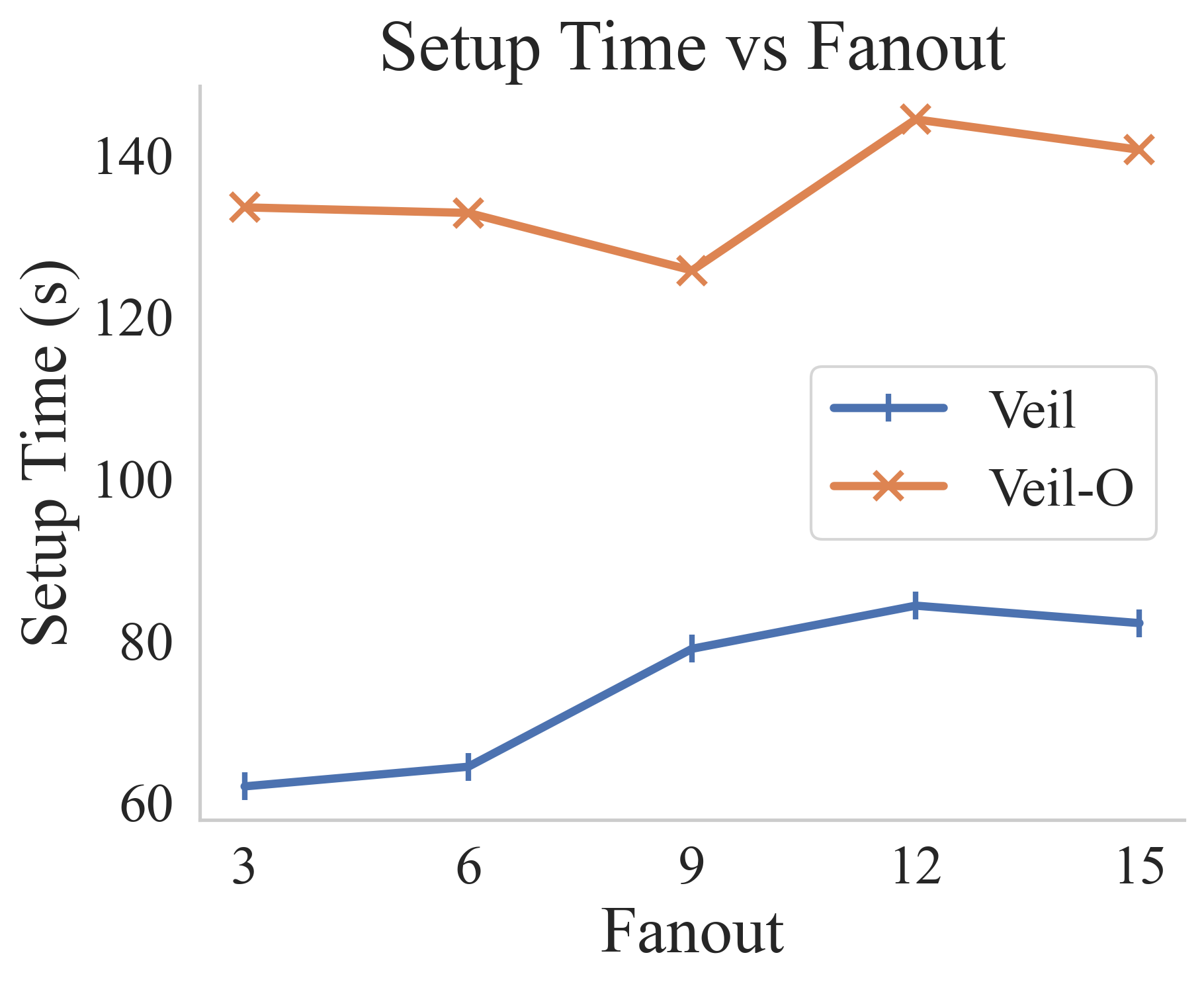}
        \caption{Setup Time vs Fanout}\label{fig: exp_set_up_time_vs_fanout}
    \end{minipage}%
    \begin{minipage}{0.33\textwidth}
        \centering
        \includegraphics[width=0.98\linewidth]{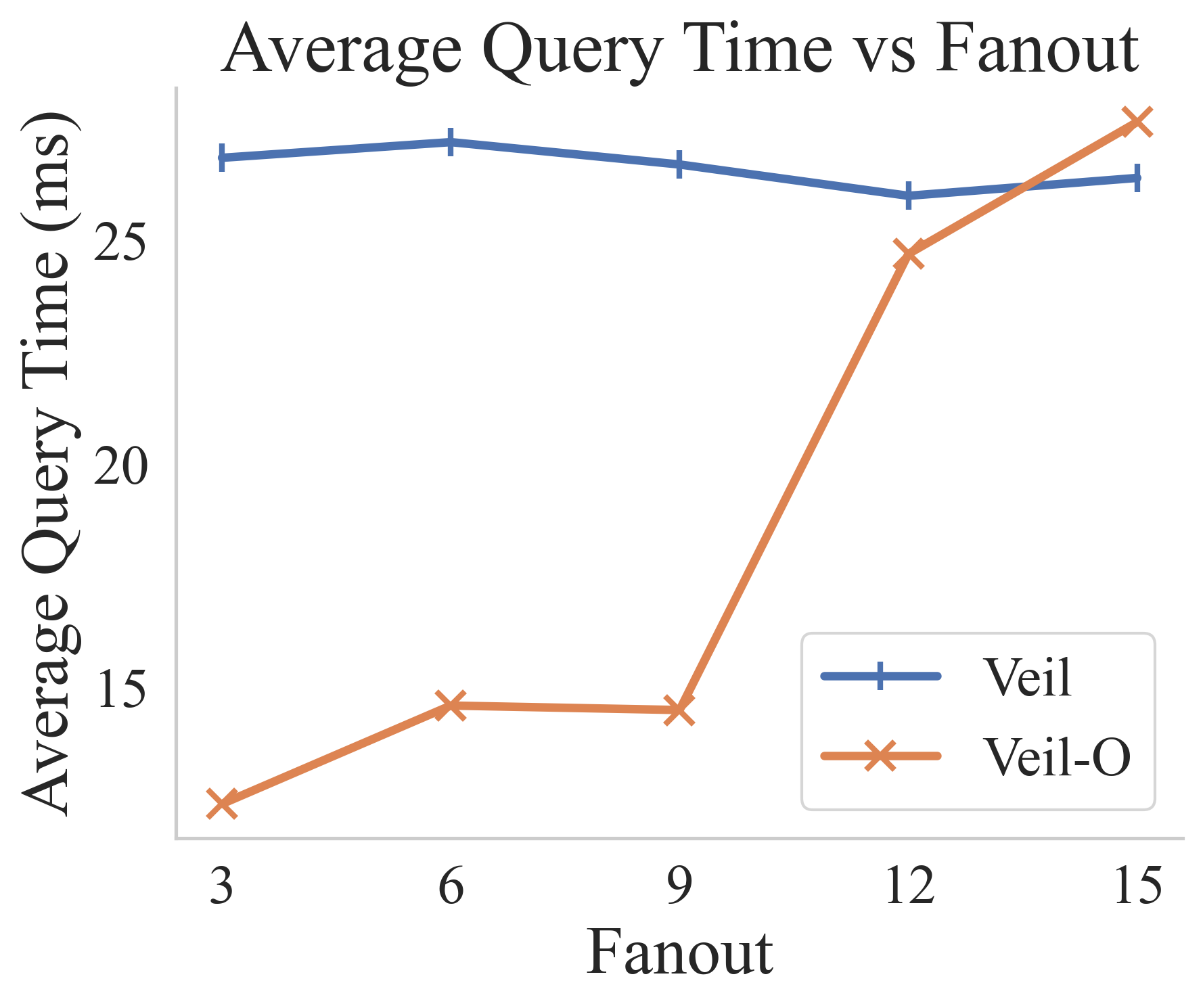}
        \caption{Average Query Time vs $f$}\label{fig: exp_query_time_vs_fanout}
    \end{minipage}
    % \vspace{1em} .
\end{figure*}

\noindent\textbf{Exp 1: Impact of fanout $f$.}
Fig.~\ref{fig: exp_stash_vs_fanout_fixed_QA_SA_9_datasets} presents $\mathit{SR}$ as a function of $f$ with fixed values of $\mathit{QA}$ and $\mathit{SA}$, where $\mathit{QA}$ is 1 and $\mathit{SA}$ is 1.2.\footnote{$\mathit{QA}$ is more important in a secure outsourcing system compared with the low price for storage in a public cloud.} The results show that $\mathit{SR}$ decreases as the fanout $f$ increases. This is because a larger $f$ provides more ``choices''  when selecting buckets for records, and a record is placed into the local stash only if all the chosen $f$ buckets are full. However, when the fanout reaches a certain threshold (\textit{e}.\textit{g}., 6 in Fig.~\ref{fig: exp_stash_vs_fanout_fixed_QA_SA_9_datasets}), this benefit is not observed. $\mathit{SR}$ may increase if the fanout $f$ is further increased. This is because, at this stage, the bucket size, according to Equation~\ref{eq: bucket_size_num}, becomes small, making the buckets more likely to become full during the random bucketing process. \emph{In the rest of the experiments, we will use 6 as an optimal value of $f$.}

\noindent\textbf{Exp 2: Impact of storage amplification $\mathit{SA}$.}
Fig.~\ref{fig: exp_stash_vs_sa_fixed_QA_F_9_datasets} illustrates the relationship between the stash ratio $\mathit{SR}$ and the storage amplification $\mathit{SA}$, with fixed values of $\mathit{QA} = 1$ and fanout $f = 6$ (the optimal value according to \textbf{Exp} 1).
The results show that when $\mathit{SA}$ is 1, \textsc{Veil} uses a small stash, with $\mathit{SR}$ to be approximately 0.002. When $\mathit{SA}$ gets larger (\textit{e}.\textit{g}., $\mathit{SA}\geq 1.2$), the $\mathit{SR}$ reduces to nearly zero, with only a small number of values (less than 5) stored in the local stash. This is because when $\mathit{QA}$ and fanout $f$ are fixed, the bucket size $\ell_b$ remains constant according to Equation~\ref{eq: bucket_size_num}. Consequently, an increase in $\mathit{SA}$ leads to the creation of more buckets, which in turn results in a reduced $\mathit{SR}$.
%Additionally, the $\mathit{SR}$ for the skewed dataset is smaller than that for the non-skewed dataset. In a skewed dataset, fewer keys are considered ``big'', \textit{i}.\textit{e}., keys with a large number of associated values, while the majority of keys have a small number of values. The maximum key size $L_{\textit{max}}$ is large compared to most other keys; given a desired fanout $f$, the bucket size $\ell_b$ is substantial (according to Equation~\ref{eq: bucket_size_num}), making it less likely for the buckets to become full when allocating records to them.

% Additionally, we set the Query Amplification (QA) to 1.2, as according to Fig.~\ref{fig: exp_stash_vs_fanout_fixed_SA}, when $\mathit{QA}$ is greater than 1.2, the stash is nearly zero, making it difficult to observe the effect of varying $\mathit{SA}$ on the stash; conversely, when $\mathit{QA}$ is small, for instance, 1.0, the stash is considerably larger.

% For comparison, we also used \textsf{dprfMM} as a baseline with $\mathit{SA}$ to be 2 and 2.6, 
% For non-skewed TPC-H dataset, the Stash Ratio is 0.13 (when $\mathit{SA}$ is 2) and 0.09 (when $\mathit{SA}$ is 2.6), while for skewed TPC-H dataset, the Stash Ratio is 0.08 (when $\mathit{SA}$ is 2) and 0.05 (when $\mathit{SA}$ is 2.6).

\noindent\textbf{Exp 3: Impact of $\mathit{QA}$.}
Fig.~\ref{fig: exp_stash_vs_qa_fixed_SA_F_9_datasets} illustrates the relationship between the stash ratio $\mathit{SR}$ and the query amplification $\mathit{QA}$, with fixed values of $\mathit{SA} = 1$ and fanout $f = 12.$
%, where $\mathit{SA}$ is set to 1.2 and $f$ is set to 12. 
The results show that $\mathit{SR}$ is considerably small and is even reduces to 0 when $\mathit{QA}$ is greater than 1.2.

\noindent\textbf{Exp 4: Impact of skew factor $z$.}
Fig.~\ref{fig: exp_stash_vs_z} illustrates the relationship between the stash ratio $\mathit{SR}$ and the skew factor $z$, while fixing $\mathit{SA}$ of 1.2, $\mathit{QA}$ of 1, and $f$ of 12. With the skew factor $z$ increasing, the maximum key size $L_{\mathit{max}}$ and bucket size increases. $\mathit{SR}$ also decreases, since the bucket size is larger in comparison to the majority of other keys in the dataset, and there is a high likelihood of available space for a record within the buckets, resulting in a reduced probability of a record being placed in the stash.
However, as $z$ continues to increase, reaching a certain threshold, such as 0.8 in Fig.~\ref{fig: exp_stash_vs_z}, the stash expands due to the significant disparity between $L_{\mathit{max}}$ and the sizes of most keys in the dataset. With a fixed $\mathit{SA}$, the number of buckets decreases in accordance with Equation~\ref{eq: bucket_size_num}, causing the buckets corresponding to the keys with a large number of records to be more likely to reaching full capacity. Consequently, records mapped to these buckets are more likely to be placed into the stash.

\subsection{Padding Strategy Evaluations}
\noindent\textbf{Exp 5: Impact of the degree $d$ when creating a $d$-regular graph.}
We fixed $\mathit{SA}$ to 1.2, $\mathit{QA}$ to 1, and varied the degree $d$ to 2, 4, 6, 8, and 10 to evaluate the impact of $d$ on the reduction of $\mathit{SA}$ in \textsc{Veil-O}. 
The results in Fig.~\ref{fig: exp_real_sa_vs_d} indicate that when $d$ is 2, the $\mathit{SA}$ is the smallest, which is 1.08. As $d$ increases, the SA increases, since each bucket has more ``neighbors'', and the probability of two large buckets being neighbors gets higher, making the number of common records for each two neighbored buckets (\textit{i}.\textit{e}., the weights over edges in the $d$-regular graph) smaller. When the degree $d$ is increased to 6, the overlapping size becomes zero, and the $\mathit{SA}$ is equal to the desired value.
% \noindent\textbf{Exp 6: Effectiveness of the Padding Strategies.}
% We fixed $\mathit{QA}$ to 1 and varied the desired $\mathit{SA}$ within the range of 1.2 to 2.0 to evaluate the effectiveness of the padding strategies  in minimizing the number of fake values introduced. 
% Specifically, we evaluate the reduction of $\mathit{SA}$ using ($\mathit{SA}$ - Real $\mathit{SA}$)/$\mathit{SA}$, where Real $\mathit{SA}$ indicates the storage amplification after applying the overlapping strategies. 
% The results in Fig.~\ref{fig: exp_padding_strategies_skew} show that the \textsc{Veil-O} outperforms the basic approach that equalizes bucket sizes through padding. Specifically, \textsc{Veil-O} can reduce the number of fake values by allowing buckets to share common values. Furthermore, the \textsc{Veil}-C Optimized exhibits promising results, necessitating only a minimal amount of fake values compared with the other approaches.

{\color{black}\noindent\textbf{Exp 6: Impact of desired overlapping size in $\textsc{Veil-O}^{\prime}$. }We fixed \textit{SA} to 1.2, \textit{QA} to 1, degree  $d$ to 2, and varied the desired overlapping size to 2, 4, 6, 8, and 10 to evaluate the impact of the desired overlapping size on both \textit{SR} and \textit{SA}. Results are presented in Fig.~\ref{fig:impact_of_overlapping}. 
The results show that the \textit{SA} decreases with an increasing desired overlapping size, as more common records are shared between neighboring buckets. Meanwhile, more records are placed into the stash due to the mandatory sharing of common records.}

\begin{figure}[htbp]
    \begin{subfigure}[b]{0.38\textwidth}
        \includegraphics[width=\textwidth]{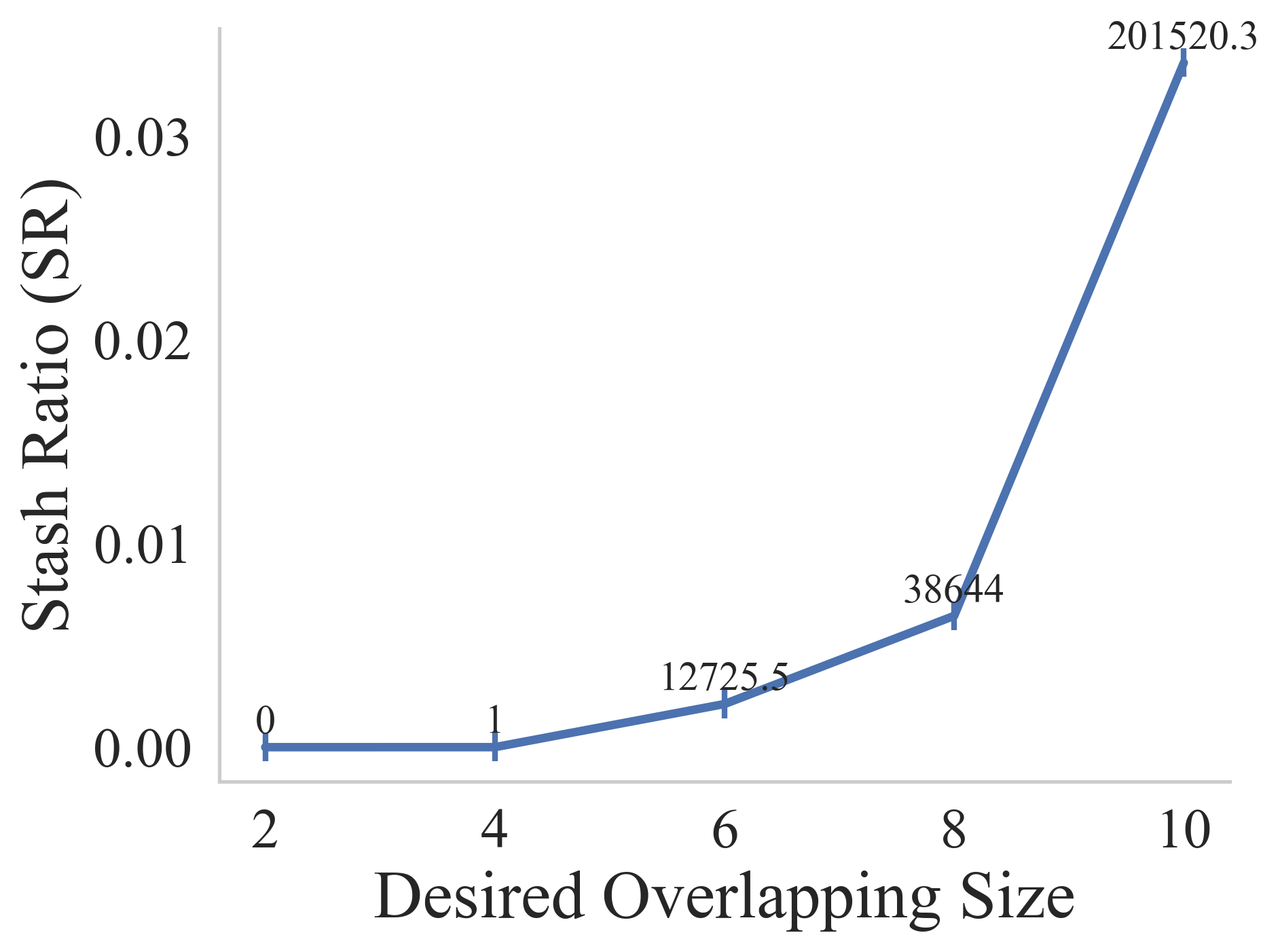}  
        \caption{** \textit{SR} vs desired overlapping}
        \label{fig:exp_stash_vs_desired_overlapping}
    \end{subfigure}
    \hfill
    \begin{subfigure}[b]{0.38\textwidth}
        \includegraphics[width=\textwidth]{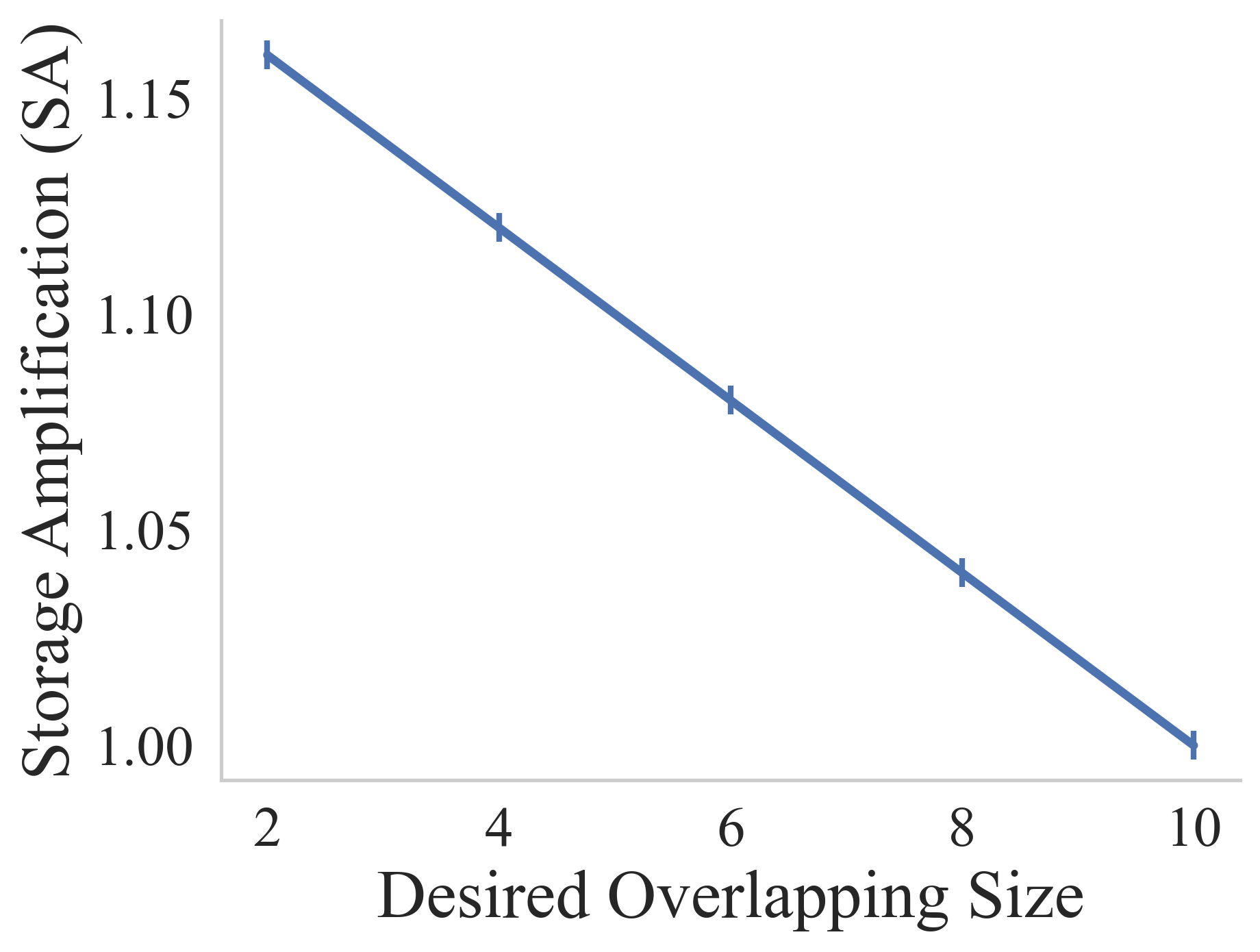} 
        \caption{\textit{SA} vs desired overlapping}
        \label{fig:exp_SA_vs_desired_overlapping}
    \end{subfigure}
    \caption{Impact of the desired overlapping size in $\textsc{Veil-O}^{\prime}$. }\label{fig:impact_of_overlapping}
\end{figure}

\subsection{Comparison with the State-of Arts}
\noindent\textbf{Exp 7: Comparisons with \textsf{dprfMM} and \textsf{XorMM}.}
To compare \textsc{Veil} with \textsf{dprfMM}~\cite{Moti_dprfMM},
{\color{black} we fixed the fanout $f$, the desired $\mathit{SA}$ and $\mathit{QA}$ of \textsc{Veil} to 12, 1.2 and 1,} respectively, and  
set the $\mathit{SA}$ and $\mathit{QA}$  of \textsf{dprfMM}~\cite{Moti_dprfMM} to 2 and 2.6, in accordance with their experimental configurations. {\color{black}We also included \textsf{XorMM}~\cite{XorMM}, which has a fixed \textit{SA} of 1.23 and \textit{QA} of 1, for comparison.}
% For the random bucketing strategy, we set $\mathit{SA}$ to 1.2, $\mathit{QA}$ to 1.2, and set the fanout to 12. 
We used two different TPC-H datasets with 6M records, with varying skew factor $z$ to be 0 (\textit{i}.\textit{e}., the original TPC-H dataset) and 0.4, respectively. 
The results are presented in Table~\ref{tb: compare_with_existing_approaches_z0}, and~\ref{tb: compare_with_existing_approaches_z0_4}, respectively.
The results demonstrate that \textsc{Veil} outperforms \textsf{dprfMM} when processing skewed datasets, utilizing a significantly smaller stash, almost zero when $z$ is 0.4. Additionally, the overlapping strategy \textsc{Veil-O} reduces $\mathit{SA}$ from 1.2 to 1.12 when $z$ is 0.4, as the bucket sizes for skewed datasets are larger in comparison to the number of records for most keys. 
%\vishal{This sentence sounds incorrect.} 
%This increases the likelihood of records being allocated to buckets. 
However, this is not observed for non-skewed data which results in a higher stash ratio ($\mathit{SR}$).

% SELECT
% pg_size_pretty(pg_total_relation_size('dprfmmtpch_6m0_3_1')) AS total_size

\begin{table}[h!]
\centering
\small
\caption{Comparisons on 6M non-skewed TPC-H dataset ($z=0$). }
\begin{threeparttable}
\begin{tabular}{|P{2.5cm}|P{1.2cm}|P{1.2cm}|P{1.2cm}|P{1.5cm}|P{1.8cm}|P{1.5cm}|}
% \begin{tabular}{|c|c|C{3cm}|C{3cm} |c|c|c|}
% {lcccC{2.0cm}C{2.0cm}C{2.0cm}C{2.0cm}}
\hline
 & $\mathit{QA}$ & $\mathit{SA}$ & $\mathit{SR}$ & $\mathit{CSA}$ & $\mathit{SSA}$\\ \hline 
\textsf{dprfMM}~\cite{Moti_dprfMM} & 2 & 2 & {5.949E-5} & {1.467E-5}& 1.000 \\ \hline
%39013 bytes; 1319 +1218 MB
% 6001215, 357 in stash
\textsf{dprfMM}~\cite{Moti_dprfMM} & 2 & 2.6 & {0} & {1.287E-9}& 1.000 \\ \hline
% 4 bytes; 1625+1339MB

\textsf{XorMM}~\cite{XorMM} & 1 & 1.23 & - &2.18E-8& 1.000 \\%5.81E-8\\ 
% (12 bytes)&(32bytes); 526MB
\hline
\textsc{Veil} & 1 & 1.2 & 0.056 &0.0229&1.020\\ \hline
%local stash: 36536290.33bytes, server idx: 31686424 bytes, total storage 1518 MB

\textsc{Veil-O} ($d$= 2 )  & 1 & 1.2 & 0.056 &0.0033&1.041\\ \hline

{\textsc{Veil} }& 1 & 2 & 4.566E-5 &1.552E-5&1.027\\ \hline
%local stash: 36536290.33bytes, server idx: 50.364200592041016/1834 MB

{\textsc{Veil-O} ($d$= 2)}  & 1 & 2 & 4.649E-5 &2.066E-6&1.050 \\\hline
% 1014mb; index: 50.365264892578125 MB; stash 0.0020952224731445312 MB

\end{tabular}
\begin{tablenotes}[leftmargin=5in]
% \setlength\labelsep{0pt}
% \setlength\labelwidth{0pt}
% \setlength\itemindent{0pt}
        % \footnotesize
        \small
        % \centering
      \item Note: The $\mathit{CSA}$ and $\mathit{SSA}$ of \textsf{dprfMM} are low because \textsf{dprfMM} stores two tables that contain $2|\mathcal{D}|$ or $2.6|\mathcal{D}|$ records in total. The SR of \textsc{Veil} and \textsc{Veil-O} can be high when QA is optimal and SA is low (\textit{e}.\textit{g}., when QA is 1 and SA is 1.2), as more collisions happen when allocating records into buckets.
\end{tablenotes}
\end{threeparttable}

\label{tb: compare_with_existing_approaches_z0}
\end{table}

\begin{table}[h!]
\begin{center}
\small
\caption{Comparisons on 6M skewed TPC-H dataset ($z=0.4$)}
\begin{tabular}{|P{2.5cm}|P{1.2cm}|P{1.2cm}|P{1.2cm}|P{1.5cm}|P{1.8cm}|P{1.5cm}|}
\hline
 & $\mathit{QA}$ & $\mathit{SA}$ & $\mathit{SR}$ & $\mathit{CSA}$ & $\mathit{SSA}$\\ \hline 
\textsf{dprfMM}~\cite{Moti_dprfMM}& 2 & 2 & {5.000E-7} & {1.411E-5} & 1.000\\ \hline
% 3 in stash; 36231 bytes stash; 1320+1128MB 

\textsf{dprfMM}~\cite{Moti_dprfMM}  & 2 & 2.6 & {0} & {1.287E-9}&1.000\\ \hline
% 1625+1340MB: data storage; 4 byte stash

\textsf{XorMM}~\cite{XorMM} & 1 & 1.23 & - & 2.18E-8& 1.000 \\\hline %5.81E-8\\ \hline
% (12 bytes)&(32bytes), dbsize 525MB

% \textsf{XorMM}~\cite{XorMM} & -  & - & 1 & 1.23 & 0.53 \\ \hline
 % & Fanout& Bucket size & Bucket \# & $\mathit{QA}$ & $\mathit{SA}$ & stash \\ \hline
 % NF~\cite{next_fit_bin_packing} & 1 & 357 & 1 & 1.05 & N.A. \\ \hline
% FFD~\cite{ffd}& 1 & 357 & 1 & 1.008 & N.A.\\ \hline
\textsc{Veil} & 1 & 1.2 & 4.26E-07 & 2.71E-7
%(441.4 bytes/1552MB)
& 1.018
%(27.92MB) 
\\ \hline
\textsc{Veil-O} ($d=2$) & 1 & 1.2 & 3.67E-07& 2.68E-8   & 1.0314\\ \hline
%25 bytes; 27.92381MB; DB storage = 889MB; 2.2 records in stash avg

{\textsc{Veil}} & 1 & 2 & 0 & 3.477E-9
%(8 bytes/2194MB)
& 1.027
%(50.364200592041016MB / 1834mb) 
\\ \hline

{\textsc{Veil-O} ($d=2$)} & 1 & 2 & 0& 7.524E-9  & 1.056 \\ \hline
% stash 7.62939453125E-6 MB index: 57.22074508666992 MB
% total 1014mb
\end{tabular}

\label{tb: compare_with_existing_approaches_z0_4}
\end{center}
\end{table}

% \vspace{-20pt}
\subsection{Performace on Large Datasets}
% To evaluate how \textsc{Veil} handles large datasets, we generated a TPC-H dataset with 36 million records (corresponding to the scale factor of 6 in TPC-H), incorporating a skew factor ($z$) of 0.4, and evaluated \textsc{Veil} using the stash ratio ($\mathit{SR}$) and the actual $\mathit{SA}$.
To evaluate how our algorithms (\textsc{Veil} and \textsc{Veil-O}) handle large datasets, we generated a TPC-H dataset with 36 million records (corresponding to the scale factor of 6 in TPC-H), incorporating a skew factor ($z$) of 0.4 and ran experiments to measure stash ratio $\mathit{SR}$ with fixed $\mathit{QA}$ and $\mathit{SA}.$ We also evaluate how changing the degree $d$ in \textsc{Veil-O} impacts the reduction of $\mathit{SA}$. 
%We measured evaluated them using the stash ratio ($\mathit{SR}$) and the actual $\mathit{SA}$.

\noindent\textbf{Exp 8: Stash Ratio on Large Datasets.}
We executed \textsc{Veil} on the 36M dataset 10 times to compute an average stash ratio ($\mathit{SR}$), which was found to be nearly zero, with 1.3 records placed in the stash on average. Specifically, in 7 out of the 10 experiments, the stash size is zero.

%As a comparison, we ran $\textsf{dprfMM}$ on the 36M dataset 10 times, which produced an $\mathit{SR}$ of 0.075 and a stash size of 451,595 for each of the experiments.

\noindent\textbf{Exp 9: Reduction $\mathit{SA}$ in \textsc{Veil-O} on Large Datasets.}
We ran \textsc{Veil-O} 10 times on the 36M dataset to compute an average practical $\mathit{SA}$. % achieved after employing the overlapping strategy \textsc{Veil-O}. 
We varied the degree of the $d$-regular graph in \textsc{Veil-O} to 2, 4, 6, 8, and 10, and varied the desired $\mathit{SA}$ to 1.2, 1.6, and 2.0. 
The results are in Fig.~\ref{fig: exp_real_sa_vs_d_36M}. 
The results show that 
\textsc{Veil-O} effectively reduces the number of fake values, especially when the desired $\mathit{SA}$ is higher.
%, \textit{e}.\textit{g}., 2. 

%In all 10 experiments, \textsc{Veil-O} creates 2-regular graph with the overlapping size on each edge to be 4, and $\mathit{SA}$ was reduced to 1.12 as a result of record-sharing across buckets.

\subsection{Running Time}
\noindent\textbf{Exp 10: Bucket creation time.} 
We varied the fanout from 3 to 15 and evaluated the time required for bucket creation, including allocating each record to a bucket and padding. The results of this experiment are in Fig.~\ref{fig: exp_bucket_creation_time_vs_fanout}. As the fanout increases, the time required for bucket creation increases for both the basic approach \textsc{Veil} and the overlapping approach  $\textsc{Veil-O}$. 
This is because more buckets need to be created given user-desired $\mathit{SA}$ and $\mathit{QA}$ thereby increasing the total processing time.

\noindent\textbf{Exp 11: Setup time.} 
We varied the fanout from 3 to 15 at increments of three and evaluated the setup time, \textit{i}.\textit{e}., the time taken to create, encrypt, and store the buckets. The results are in Fig.~\ref{fig: exp_set_up_time_vs_fanout}. 
We also take \textsf{dprfMM}~\cite{Moti_dprfMM} and \textsf{XorMM}~\cite{XorMM} as a comparison. When $\mathit{SA}$ was set to 2.6, the setup time for \textsf{dprfMM} was 210.09s, which dropped to 155.349s when $\mathit{SA}$ was 2. {\textcolor{black}{The setup time for \textsf{XorMM} is 124.11s. }}
The results show that $\textsc{Veil}$ outperforms $\textsc{Veil-O}$, which can be attributed to the additional time $\textsc{Veil-O}$ takes to generate a $d$-regular graph. Also, $\textsc{Veil}$ required less setup time than \textsf{XorMM}, while $\textsc{Veil-O}$ exhibited a setup duration comparable to that of \textsf{XorMM}. Furthermore, both $\textsc{Veil}$ and $\textsc{Veil-O}$ demonstrated shorter setup times than \textsf{dprfMM}.

% while $\textsc{Veil}$ take less time to set up than \textsf{XorMM}, and $\textsc{Veil}-O$ takes similar time to setup.
% --------

% The results show that for fixed $\mathit{SA} = 1.2$, and $\mathit{QA} = 2$,  $\textsc{Veil}$ outperforms \textsc{Veil-O} and \textsf{dprfMM}, as \textsc{Veil} does not encrypt the bucket-ids when outsourcing the data, which reduces the setup time. 
% This shows that the total encryption time in \textsc{Veil-O} is larger than the time taken to create and insert fake values in \textsc{Veil}. 
%includes less number of fake records due to the sharing of common values between neighboring buckets, and this reduces the total number of records inserted into the database. 

\noindent\textbf{Exp 12: Query time.} 
To evaluate the time taken for query execution, we varied the fanout $f$ from 3 to 15, with increments of three, and ran 20 queries to compute an average query time for each query. We also included \textsf{dprfMM} and \textsf{XorMM} as a comparison. The average query time for \textsf{dprfMM} is 201ms when $\mathit{SA}$ is 2.6 and 63ms when $\mathit{SA}$ is 2. {\color{black}The average query time for \textsf{XorMM} is 33.1ms.
The results are shown in Fig.~\ref{fig: exp_query_time_vs_fanout}. 
Results show that regarding query time, \textsc{Veil} and \textsc{Veil-O} outperform both \textsf{dprfMM} and \textsf{XorMM}. Also, note that the query time for \textsc{Veil} is similar when the fanout varies. That is because fanout only determines the number of buckets to retrieve for each query and  doesn’t have any impact on the number of records retrieved. When processing a query, both \textsc{Veil} and \textsc{Veil-O} get the desired bucket-ids using \textsf{Map} and retrieve all buckets together. }

\section{Conclusion}\label{sec:conclusion}
This paper proposed \textsc{Veil} to prevent volume leakage in encrypted search. 
\textsc{Veil} selects multiple buckets for each key and employs a greedy strategy to distribute key-value pairs into different buckets. 
With the created buckets, \textsc{Veil} applies various strategies to add fake values to the created buckets to pad them to an equal size. This increases $\mathit{SA}$. 
To reduce the number of fake values that need to be added to buckets, we proposed \textsc{Veil-O} which allows sharing values among 
buckets by creating regular graphs with vertices corresponding to the buckets and edges connecting two vertices if the corresponding buckets share values. This significantly improves $\mathit{SA}.$ We verified this using evaluations and showed that both of our approaches are strictly better than the state of art.

\section*{Acknowledgements}
This work was supported by NSF Grants No. 2032525, 1545071, 1527536, 1952247, 2008993, 2133391, and 2245372. Vishal Chakraborty was also supported by the HPI Research Center in Machine Learning and Data Science at UC Irvine. 
The authors thank the reviewers for their feedback.

% \begin{acks}
% To Robert, for the bagels and explaining CMYK and color spaces.
% \end{acks}

%%
%% The next two lines define the bibliography style to be used, and
%% the bibliography file.
\bibliographystyle{ACM-Reference-Format}
\bibliography{references}

\appendix

\section{Leakage Functions}\label{sec:leakage_functions}
We follow the typical leakage functions introduced by Kamara \textit{et al.}~\cite{kamara2018suppression}, as follows.

\noindent \textbf{Query Equality}: the equality pattern leaking whether two queries are to the same key or not. Given queried keys $\mathcal{K}_Q=\{k_1, \ldots, k_t\}$, $\textsf{qeq}(k_1, \ldots, k_t) = \mathcal{M}$ consists of a $t \times t$ matrix such
that $\mathcal{M}[q][q^\prime] = 1$ if and only if $k_q = k_{q^\prime}$.

\noindent \textbf{Response Length}: the number of
records associated with queried keys, \textit{i}.\textit{e}., the real volume, or the real output size. Formally, Given queried keys $\mathcal{K}_Q=\{k_1, \ldots, k_t\}$ and a KV dataset $\mathcal{D}$, $\textsf{rlen}(\mathcal{D}, k_1, \ldots, k_t) =
(|\mathcal{D}[k_q]|)_{k_q\in\mathcal{D}}$ where $|\mathcal{D}[k_q]|$ refers to the tuple of values associated with $k_q$.

\noindent \textbf{Maximum Response Length}: The maximum number of values associated with any key in $\mathcal{D}$. For a KV dataset $\mathcal{D}$, $\textsf{mrlen}(\mathcal{D}) = \texttt{MAX}\{|k_i|\}_{k_i\in\mathcal{K}} = L_{\mathit{max}}$.

\noindent \textbf{Domain Size}: The total number of records in the KV dataset $\mathcal{D}$. Formally,  $\textsf{dsize}(\mathcal{D}) = |\mathcal{D}| = \sum |k_i|_{k_i\in\mathcal{K}}$.

\section{Analysis of \textsc{Veil} Based on $\mathcal{BC}$}\label{sec:security_analysis_veil}

We can analyze \textsc{Veil} based on both performance and
security. From the performance perspective, \textsc{Veil} provides guaranteed 
storage and query amplification as specified
by the user. 
Furthermore, we further show that if 
the value of $\mathit{SA} \geq 1$, $\mathit{QA} \geq 1$, and $f \geq 0$, 
the expected size of stash  is zero.
%execution  properties based on the stash ratio $\mathit{SR}$ as well as the user-defined $\mathit{SA}$ and $\mathit{QA}$. 
Our experimental results in \S\ref{sec: experiment} further illustrate that the stash size remains very small - much smaller
than the stash created by \cite{Moti_dprfMM} at significantly lower 
$\mathit{SA} $ and $\mathit{QA}$.
Below, we focus on security analysis of \textsc{Veil}.
In particular, we show that \textsc{Veil} guarantees secure execution 
as discussed in \S\ref{sec:setting}. We state Theorem~\ref{th: veil_basic} as follows.

\begin{theorem}
\label{th: stash size}
Given a dataset $\mathcal{D}$, user-defiend $\mathit{SA}$ and $\mathit{QA}$, and fanout $f$, suppose $\mathcal{BC}$ creates a set of buckets $\mathcal{B}$ over $\mathcal{D}$ using $\textsf{MAP(*)}$. If $\mathit{SA} \geq 1$, $\mathit{QA} \geq 1$, and $f \geq 0$, 
\textcolor{black}{then the average size of the stash, represented as $s$, is expected to be zero, denoted as $\mathbb{E}[s]=0$. This implies that, although $s$ of individual cases might vary from zero, under the defined conditions, the overall average (or expected value) of the stash size will be zero.}

\end{theorem}
% \begin{proofhint}
% The probability of a candidate bucket being selected (\textit{i}.\textit{e}., the smallest one) is $\frac{1}{f}$. Therefore, the probability of each bucket $B_j$ being selected for a record is $\frac{f}{n}\cdot \frac{1}{f}=\frac{1}{n}$. 
% After allocating all records in $\mathcal{D}$ to the buckets in $\mathcal{B}$, the expectation of the number of records in each bucket $B_j$ is $|\mathcal{D}|\cdot \frac{1}{n}$, which is equal to the minimum possible bucket size (\textit{i}.\textit{e}., corresponding to the case when $\mathit{SA}$ is 1). This observation implies that even if the desired $\mathit{SA}$ is the optimal value (which is 1), we do not need to place records into the local stash, resulting in an expected stash size of 0. 
% $\Box$
% \end{proofhint}

\noindent\textbf{Proof of Theorem~\ref{th: stash size}.}
% \label{proof: stash size}
Consider a key-value dataset $\mathcal{D}$ with a key set $\mathcal{K}$. Let $\mathcal{B}$ be the bucket set created over $\mathcal{D}$, with the number of buckets being $n$.
The bucket size is calculated as 
$\frac{QA\cdot L_{\mathit{max}}}{f}$, which is also equal to $\frac{\mathit{SA}\cdot|\mathcal{D}|}{n}.$  Since $\mathit{SA} \geq 1$, the bucket size $\ell_b\geq\frac{|\mathcal{D}|}{n}$.

For each key $k_i\in \mathcal{K}$, let its size (or volume) be $|k_i|$. After shuffling $\mathcal{D}$, we allocate the records into buckets. The probability of a record being associated with a key $k_i$ is $\frac{|k_i|}{|\mathcal{D}|}$. For a bucket $B_j\in\mathcal{B}$, the probability of $B_j$ being a candidate bucket for key $k_i$ (\textit{i}.\textit{e}., one of the $f$ buckets obtained using \textnormal{\textsf{MAP}}) is $\frac{f}{n}$.
The probability of a candidate bucket being selected (\textit{i}.\textit{e}., the smallest one) is $\frac{1}{f}$. Therefore, the probability of each bucket $B_j$ being selected for a record is $\frac{f}{n}\cdot \frac{1}{f}=\frac{1}{n}$. 
After allocating all records in $\mathcal{D}$ to the buckets in $\mathcal{B}$, the expectation of the number of records in each bucket $B_j$ is $|\mathcal{D}|\cdot \frac{1}{n}$, which is equal to the minimum possible bucket size (\textit{i}.\textit{e}., corresponding to the case when $\mathit{SA}$ is 1). This observation implies that even if the desired $\mathit{SA}$ is the optimal value (which is 1), we do not need to place records into the local stash, resulting in an expected stash size of 0.

We next focus on security analysis of VEIL. In particular, we
show that \textsc{Veil} guarantees secure execution as discussed in XX. We
state the theorem  and give a proof sketch followed by a formal proof based on real-ideal game.

\begin{theorem}
\label{th: veil_basic}
Let $\mathcal{D}$ be a dataset with the key set $\mathcal{K}$. Given user-defined $\mathit{SA}$ and $\mathit{QA}$, and fanout $f$, suppose $\mathcal{BC}$ creates a set of buckets $\mathcal{B}$ over $\mathcal{D}$ using $\textsf{MAP(*)}.$ The \textsc{Veil} strategy with $\mathcal{BC}$ is secure with respect to VSR.
\end{theorem}

% Intuitively, the adversary cannot learn association between buckets and keys, as the $f$ buckets for each key are randomly selected, and the dataset $\mathcal{D}$ is shuffled to mix records of different keys to achieve a random order when inserting records to buckets.
% Also, for each query, the retrieved number of buckets ($f$) are same, and  buckets are of same size, making the total number of values to retrieve is the same.
% Thus by simply observing buckets, adversary cannot guess which key is
% being queried.
% Even in the worst case, \textit{i}.\textit{e}., the adversary knows queries to $|\mathcal{K}|-2$ keys, \textsc{Veil} is still secure when the adversary observes a new query to one of the two unknown keys, and the adversary can not determine which key the query is to. This is because by observing the $f$ buckets retrieved for the new query, although the adversary may know all the other keys in the $f$ buckets, the adversary can not know the number of value of each key in each bucket, as they may be  distributed in different buckets or the local stash. 

\noindent\textbf{Proof of Theorem~\ref{th: veil_basic}.}
Before we develop a more formal proof  based on Real-Ideal game, 
we first provide an intuitive sketch of why \textsc{Veil} 
is secure based on our security requirement.
Let $\mathcal{K}_q = \{k_1, \ldots, k_m \}$ be the set of all prior queried keys.
Suppose for each queried key $k_q\in\mathcal{K}_q$, $\textsf{Cipher}(k_q)$ is the set of buckets $B \in \textnormal{\textsf{MAP}}(k_q)$. Since all buckets are
equi-sized and $\textnormal{\textsf{MAP}}(k_q)$  returns a set of $f$ buckets,
the size of ciphertext returned, $|\textsf{Cipher}(k_q)|$, is the same
irrespective of they query keyword $k_q$. Thus, the size of ciphertext
returned does not reveal anything to the adversary about $k_q$. 
To show \textsc{Veil} is secure with respect to VSR, we need to further show
that if
 $k_\alpha \notin \mathcal{K}_q$ , for all distinct keys $k_1$, $k_2$, such
 that $k_1,k_2 \notin \mathcal{K}_q$, the adversary cannot determine 
if $k_\alpha$
corresponds to $k_1$  or $k_2$. That is, 
$\mathit{Prob}[k_\alpha = k_1|A] - \mathit{Prob}[k_\alpha = k_2|A]  <= \epsilon $ where $\epsilon$ is an extremely small positive number and $A$ is the knowledge of the adversary.

To see this, consider the worst-case situation when the adversary has the maximal knowledge, -- \textit{i}.\textit{e}., the adversary knows the keywords for all prior queries executed so far, and there are only two keys $k_1$ and $k_2$ that have never beend queried.  
Thus, the adversary knows all $k_q \in \mathcal{K}_q$ and the corresponding $\textnormal{\textsf{MAP}}(k_q)$. Let, for the query to key $k_\alpha$, the set of buckets retrieved be $\textnormal{\textsf{MAP}}(k_\alpha).$
Since \textsc{Veil} uses a random assignment of keys
to buckets, for buckets $B_p, B_q \in \mathcal{B},$ we have,
$Pr[B_p \in \textnormal{\textsf{MAP}}(k1)\mid \textnormal{\textsf{MAP}}(k_1), \ldots, \textnormal{\textsf{MAP}}(k_m))]  = 
Prob[B_q \in \textnormal{\textsf{MAP}}(k2) \mid \textnormal{\textsf{MAP}}(k_1), \ldots, \textnormal{\textsf{MAP}}(k_m)]$.
Thus, an observation of the access pattern of
the query to key $k_\alpha$ does not enable the adversary to distinguish between $k_1$ and $k_2$. 
%\end{proofhint}

%\section{Real-Ideal Game Proofs}
We now turn our attention to proving the above theorem more formally using a
use a traditional real-ideal game based approach to proving securty.
In the proof, the game is  played between an adversary $\mathcal{A}$ and a challenger $\mathit{C}$.

\begin{definition}\label{def:game}
The schema is secure if there exists a simulator $\mathcal{S}$ such that for all adversaries $\mathcal{A}$: 
$\mathit{|Pr[Real_{\mathcal{A},\mathcal{C}} = 1]- Pr[Ideal_{\mathcal{A}, \mathcal{S}}=1]|\leq negl(\lambda)}$
\end{definition}

%\subsection{Proof of Theorem~\ref{th: veil_basic}}\label{sec: proof_of_veil_basic}
% \%label{proof: proof_of_veil_basic}

% \subsection{Proof of Theorem~\ref{th: veil_basic}}\label{sec: proof_of_veil_basic}
\begin{proof}
We consider buckets $B_1, B_2, \ldots, B_n$ that are created over a KV dataset $\mathcal{D}$ with a set $\mathcal{K}$ of keys. Suppose the fanout is $f$. To prove the security, we consider the extreme case, where the adversary knows the queries corresponding to $|\mathcal{K}|-2$ keys in $\mathcal{K}$ -- only two keys are unknown, denoted as $k_0$ and $k_1$.
Formally, for each query $Q$ to the unknown keys, the challenger chooses a bit $\rho\in \{0, 1\}$ uniformly at random. If $\rho$ is 0, the challenger returns the buckets to the query key, while if $\rho$ is 1, the challenger will return the buckets to the other key. Recall that $\mathcal{A}$ is the adversary and $\mathcal{C}$ is the challenger.

\noindent{\textbf{$\mathit{Real_{\mathcal{A}, \mathcal{C}}}$}}: 
\begin{enumerate}
    \item $\mathcal{A}$ randomly selects and sends a key $k_i$ from $\{k_0, k_1\}$, to the challenger $\mathit{C}$.
    \item $\mathcal{C}$, on receiving $k_i$, finds the buckets $B_j \in \mathcal{B}$ that contains values corresponding to $k_i$.
    \item $\mathcal{C}$ returns $E(B_j)$ to $\mathcal{A}$.
    \item $\mathcal{A}$ outputs a bit $\mathit{b\in \{0, 1\}}$.
\end{enumerate}

\noindent{\textbf{$\mathit{Ideal}_{\mathcal{A}, \mathcal{S}}$}}: 
\begin{enumerate}
    \item $\mathcal{A}$ randomly selects a key $k_i$ from $\{k_0, k_1 \}$, and sends $k_i$ to simulator $\mathit{S}$.
    \item $\mathit{S}$ finds $\rho$ for the mapping --- if $\rho$ is 1, the simulator $\mathit{S}$ queries $k_i$; otherwise, $\mathit{S}$ queries $k_{1-i}$. Suppose the returned buckets are $E(B_j)$, $\mathcal{S}$ then returns the buckets $E(B_j)$ to $\mathcal{A}$.
    \item The adversary $\mathcal{A}$ outputs a bit $\mathit{b\in \{0, 1\}}$; 1 for the simulator returning the correct buckets, while 0 for the simulator returning the wrong bucket. 
\end{enumerate}\end{proof}

\section{Correctness of VEIL-O}
\label{th: veilo}
\begin{theorem}
\label{th: veil_o}
Let $\mathcal{D}$ be a dataset with the key set $\mathcal{K}$. Given user-defiend $\mathit{SA}$ and $\mathit{QA}$, and fanout $f$, suppose $\mathcal{BC}$ creates a set of buckets $\mathcal{B}$ over $\mathcal{D}$ using $\textsf{MAP(*)}.$ The \textsc{Veil} strategy with $\mathcal{BC}$ is secure with respect to VSR.
\end{theorem}

%\subsection
%\textbf {Proof of Theorem~\ref{th: veil_o}}\label{sec: proof_of_veil_o}:
\begin{proof}
We use a game between the adversary $\mathcal{A}$ and the challenger $\mathcal{C}$ and Definition~\ref{def:game} to formally prove that \textsc{Veil-O} is secure. Consider a set of buckets $\mathcal{B}=B_1, B_2, \ldots, B_n$ that are created over a KV dataset $\mathcal{D}$, using which \textsc{Veil-O} creates a $d$-regular graph. 
Suppose the simulator $\mathit{S}$ creates a mapping between buckets in the $d$-regular graph using a self-reverse function $F_s$, where for each input bucket $B_j$, $F(F(B_j))=B_j$. In other words, for each two buckets $B_j$ and $B_{j^{\prime}}\in\mathcal{D}$, if $F(B_j)=B_{j^{\prime}}$, then $B_{j^{\prime}} = B_j$.
% For each mapping, the simulator flips a fair coin. If the coin is head, the simulator switches values between the two clusters; otherwise, the simulator does not do anything. 
When observing a new query that fetches one of $B_j$ and $B_{j^{\prime}}$, the simulator can choose to return the real bucket or the mapped bucket.
For each two mapped buckets, the simulator chooses a bit $\rho\in \{0, 1\}$ uniformly at random. If $\rho$ is 1, the simulator will return the desired bucket, while  if $\rho$ is 0, the simulator returns the mapped bucket. %For those pairs of clusters with $\rho$ of 1, the simulator randomly generates a number $\sigma$ ($\sigma\in[1,\beta]$).

\noindent{\textbf{$\mathit{Real_{\mathcal{A}, \mathcal{C}}}$}}: 
\begin{enumerate}
    \item The adversary $\mathcal{A}$ randomly selects a key $k$, and sends $k$ to challenger $\mathit{C}$.
    \item After the challenger $\mathcal{C}$ receives $k$, the challenger finds the $f$ buckets using $\textnormal{\textsf{MAP}}(k_i)$, denoted as $\mathcal{B}$[$k_i$].
    \item The challenger retrieves buckets corresponding to $\mathcal{B}$[$k_i$] and sends it to $\mathcal{A}$.
    \item The adversary $\mathcal{A}$ outputs a bit $\mathit{b\in \{0, 1\}}$.
\end{enumerate}

\noindent{\textbf{$\mathit{Ideal}_{\mathcal{A}, \mathcal{S}}$}}:
\begin{enumerate}
    \item The adversary $\mathcal{A}$ randomly selects a key $k$, and sends $k$ to challenger $\mathit{C}$.
    \item After the challenger $\mathcal{C}$ receives $k$, the challenger finds the $f$ buckets using $\textnormal{\textsf{MAP}}(k_i)$, denoted as $\mathcal{B}$[$k_i$].
    \item The challenger retrieves buckets corresponding to $\mathcal{B}$[$k_i$] and sends it to $\mathcal{A}$.
    \item The adversary $\mathcal{A}$ randomly outputs a bit $\mathit{b\in \{0, 1\}}$.
\end{enumerate}

\begin{enumerate}
    \item The adversary $\mathcal{A}$ randomly selects a key $k$, and sends $k$ to simulator $\mathit{S}$.
    \item After the simulator $\mathit{S}$ receives $k$, the challenger finds the $f$ buckets using $\textnormal{\textsf{MAP}}(k_i)$.
    \item The simulator $\mathit{S}$ finds $\rho$ for the mapping; if $\rho$ is 0, $\mathit{S}$ returns buckets corresponding to $\mathcal{B}$[$k_i$] to $\mathcal{A}$; otherwise, the simulator $\mathit{S}$ finds mapped buckets $\mathcal{B^{\prime}[k_i]}$, where for each $B_j\in\mathcal{B[k_i]}$, $F(B_j)\in \mathcal{B^{\prime}[k_i]} $, and returns buckets corresponding to $\mathcal{B^{\prime}[k_i]}$.
    \item The adversary $\mathcal{A}$ outputs a bit $\mathit{b\in \{0, 1\}}$; 1 for the simulator returning the correct buckets, while 0 for the simulator returning the mapping buckets. 
\end{enumerate}\end{proof}

% \shanshan{In our submission, we have two related work sections: one in section 7 (discribed above), and the other is full related work section, -- we put it in appendix. As the camera-ready version does not allow appendix, how do we handle the two related works?}

\section{Full Related Works}\label{sec:full_related_work}
% \section{Related Works}
% \label{sec: related_work}

% \sharad{Volume hiding was only recently gotten attention.
% Orifginal work ignored it and allowed it as acceptable leakage.
% However, as menrtion in intro, this is serious leakage. 
% The first set of papers PRT. Prominent approaches
% are ..TDST. Storage optimized
% DST. Then Moti.

% People also considered approaches that may sometimes not give
% full answers. they may return less than lmax. Example is XX.
% New CCS paper.
% clear out and do fina version of related work.

% Above approaches only key searches. Hyberiindex consider range
% searches... how is it done. But it is leaky.

% Bucket based approch in our work can be generalize to support
% range searches... give key idea. But dertailed exploerartion is
% outside the scope of this paper. We are acrtively looking at it.

% }

% \section{Expanded Related Work}
Volume hiding has recently gained attention in the research community. In earlier studies, it was often overlooked and allowed as acceptable leakage. However, as mentioned in \S\ref{sec:intro}, volume hiding is a critical aspect of preserving privacy and preventing information leakage in secure data outsourcing.

Kamara \emph{et al}. introduced the first line of volume-hiding techniques~\cite{kamara2018suppression,kamaraVLH}, utilizing a multi-map data structure for encrypted keyword search. \textsf{DST}~\cite{kamaraVLH} creates batches over the dataset $\mathcal{D}$, where values of a key are allocated to at most $L_{\mathit{max}}$ batches, and pads the batches to equal-size when outsourcing. It retrieves $L_{\mathit{max}}$ (the maximum key size) batches each time when querying for a key. %\textsf{DST} reduces storage overhead compared to padding each key to the maximum key size at the cost of higher communication overhead.
Storage Optimized \textsf{DST}~\cite{kamaraVLH} reduces storage overhead of \textsf{DST} for datasets with a large number of duplicated values by storing those duplicated values only once. Although this approach seems reasonable, identifying duplicated values associated with multiple keys requires significant computation and has not been extensively studied. Moreover, even though duplicated values are stored only once, additional storage is required for each key to record the locations of all its values, resulting in a storage overhead of $|\mathcal{D}|$.

{\color{black}
Patel \emph{et al.} proposed \textsf{dprfMM}~\cite{Moti_dprfMM} that utilizes cuckoo hashing and two hash functions ($h_1$ and $h_2$) to construct two tables ($T_1$ and $T_2$). Each key-value pair $\langle k_i, v_j^{i}\rangle$ in the dataset is allocated a position in $T_1$ using $h_1$. If the position is taken, the algorithm uses $h_2$ to compute a position in $T_2$. When both positions are occupied, the key-value pair is stored in a local stash. To query for records of a specific key, \textsf{dprfMM} uses $h_1$ and $h_2$ to  fetch $L_{\mathit{max}}$ records from each of $T_1$ and $T_2$ then checks the stash for potential records that are stored locally. This approach offers robustness against potential adversaries, however, it suffers from high $\mathit{SA}$ and $\mathit{QA}$, both of 2, while our approach supports user-defined $\mathit{SA}$ and $\mathit{QA}$ and can achieve much better $\mathit{SA}$ and $\mathit{QA}$ with a comparable stash size.

\textsf{XorMM}~\cite{XorMM} leverages an XOR Filter~\cite{graf2020xor} to store the key-value dataset. The $\mathit{SA}$ is approximately 1.23 and the $\mathit{QA}$ is 1. It utilizes a hash function and a PRF to assign 3 slots for each record within the XOR filter. Each record is stored in its encrypted form in one slot, after undergoing an XOR operation with values in the other two slots. This technique allows for the retrieval of each encrypted records through an XOR operation on the values from the 3 corresponding slots for the record.
Upon processing a query for a specific key, XorMM employs the three hash functions to compute slots corresponding to multiple records. It then retrieves these slot values and conducts an XOR operation across the three chosen slots for each record to obtain encrypted records.

While XorMM achieves a much better $\mathit{SA}$ and $\mathit{QA}$ compared with \textsf{dprfMM}, \textsc{Veil} still outperforms it, for the following reasons: 1) 
\textsc{Veil} enable users to define $\mathit{SA}$ and $\mathit{QA}$, leading to the potential for various encrypted database sizes, even from the same or different datasets. For example, suppose two datasets $\mathcal{D}_1$ and $\mathcal{D}_2$ where $|\mathcal{D}_1| \neq |\mathcal{D}_2|$, and suppose the user desired $\mathit{SA}$ for the two datasets are $\mathit{SA}_1$ and $\mathit{SA}_2$, respectively. It may happen that $\mathit{SA}_1|\mathcal{D}_1| = \mathit{SA}_2*|\mathcal{D}_2|$ even if the datasizes of the two datasets are different. 
Such variability makes it challenging for adversaries to distinguish between encrypted databases.
2) XorMM falls short as it does not support long value records, even when we store the record-ids (RIDs) in the XOR filter rather than the records themselves. This limitation arises because each RID is computed using values from three locations, which can occasionally generate meaningless RIDs that do not map to any records in the data storage. 3) \textsc{Veil} has greater flexibility for future extension to dynamic database updates, as it provides a ``buffer'' for potential insertions into the database. XorMM, on the other hand, lacks this flexibility, as even minor changes in the database size would necessitate a complete reorganization of the XOR filter.}

Researchers have also explored approaches that may sometimes return fewer results than the maximum key size $L_{\mathit{max}}$, reducing query overhead but inevitably being lossy or leaking approximate volumes of queries. Examples include \textsf{PRT}~\cite{kamaraVLH} and \textsf{dprfMM}~\cite{Moti_dprfMM}.
\textsf{PRT} generates fake volumes for keys using a pseudorandom function, adding fake values to the key if the generated volume for the key is larger than the key size, or truncating some values of the key to the fake volume if the key size is larger. 
\textsf{PRT} pads each key to a generated volume instead of retrieving $L_{\mathit{max}}$ values for each query, thus reducing the storage overhead. However, such construction is lossy with negligible probability. Although the authors indicated that truncation is minimal when the input data is Zipf-distributed, they did not provide any guarantee for data loss.
The other approach, \textsf{dpMM}~\cite{Moti_dprfMM},  computes a fake volume for each query key using the key size and noise drawn from the Laplace distribution, instead of retrieving a fixed number of values for each query. This approach reduces the query overhead at the cost of lossy query answers with negligible probability.
% % {\color{red}Compared with these approaches, our approaches can generate non-lossy answers}

HybrIDX~\cite{feifei_hybridx} handles volume hiding for range queries  of
the kind $ k_i > k$ or $k_i < k$ for some key $k$. It partitions each key
$k_i$ into small, fixed-size batches, called ``blocks'' - the
last block of each key is padded to ensure that all batch sizes are equal. 
HybrIDX partitions each range query into ``batches'', where each batch contains the same number of keys, denoted as $q$, and each time it retrieves blocks for $q$ keys from either the cloud or a secure hardware enclave. 
For keyword queries, a special case of range queries, when $q>1$, the approach retrieves blocks for more than one key, potentially resulting in a high query overhead; while when $q$ is 1, HybrIDX retrieves all blocks corresponding to the queried key, which is efficient at a cost of leaking approximate key sizes through the volume of each query. Assuming the adversary is aware of the block size $\ell_{\mathit{block}}$, upon observing the volume $\ell$ of a query,  they can deduce that the actual volume lies between $\ell-\ell_{\mathit{block}}+1$ and $\ell$.
Although HybrIDX employs a secure enclave to cache some records, the issue remains unresolved, particularly when the query key differs from previous queries. Moreover, even without knowledge of the block size, an adversary can observe multiple queries, calculate the differences between the volumes of the queries, and compute common divisors of those ``differences'' to obtain the batch size.

{\color{black}
Similar to
previous work on volume hiding \cite{Moti_dprfMM, Moti21_dynamic,XorMM}, 
we studied \textsc{Veil} 
for the static case when insertions, deletions and updates to data do not occur.
However, \textsc{Veil} can support dynamic updates without requiring a re-execution of the bucketization strategy every time there is a update  as long as the value of $L_{max}$ after insertions 
does not exceed $L_{max} \times \mathit{QA}$. For instance, insertion can be supported by retrieving the set of records 
in the $f$ buckets corresponding to the
key for the record being inserted, replacing one of the fake tuples (if present) in the buckets 
by the newly inserted tuples (or converting 
a deleted tuple into a fake tuple in case of deletion), 
re-encrypting the records in the buckets to prevent adversary from learning
about the change, and re-outsourcing that modified buckets. 
Of course, if the $f$ buckets corresponding to the key  do not contain  any fake tuples (and hence
have no residual capacity to store more data), the newly inserted tuple will be stored in a stash, as would have
been the case had the buckets become full even in the static case. 
Such a strategy will continue to work, as long as, the new value of $L_{max}$ after insertions, say $L_{max}'$ 
remains below $QA \times L_{max}$ \footnote{ In fact, it will continue to work beyond that except that all new records 
after the point of increase in $L_{max}'$ beyond $QA \times L_{max}$ will need to be stored locally in the stash.}.
As $L_{max}$ increases, since we are not changing the fanout $f$ or the 
bin size, effectively the $QA$ value reduces which, in turn, increases the probability of a record to have to be stored
in a stash. When  $L_{max}'$ goes above  $QA \times L_{max}$, we can either
continue to map new records to stash (increasing client side overhead) and/or reorganize the data. While the study
of insertion and dynamic maintenance of \textsc{Veil} is outside the scope of this work, we note that our approach is
significantly more flexible and amenable for dynamic updates compared to prior work such as  \cite{Moti_dprfMM} or
\cite{XorMM}. For instance,  \cite{XorMM} constructs a xor filter through a sequential process such that 
an insertion of a  record (whether or not it causes the 
$L_{max}$ value to change) would require recomptuation of the entire  filter. As a result, there is no easy way to 
support dynamic updates on approaches such as \cite{XorMM} and this, coupled with flexibility by the user to choose
$QA$ and $SA$ parameters to control the overheads, is one of the primary advantages of \textsc{Veil}.

}

% \input{old_related_work}

%% SHARAD here you need to prove that neighbor functions are correctr.T

% \section{Neighbor Functions}
%  \begin{theorem}
%  \label{th: symFunc}
%     Let $n \neq 3$ and $2 \leq d \leq n-1.$
%       Consider any set of $n$ buckets, 
%       $\{B_0,\ldots, B_(n-1)\}$. 
%       For $0 \leq p \leq n-1$, associate neighbors of bucket $B_p$ (viz., $\mathcal{N}(B_p)$) using the set of $d$ functions $\mathcal{F}$ defined as follows:
%       \begin{itemize}
%          \item For $1 \leq i \leq \lfloor \frac{d}{2} \rfloor$ we have $F_i(B_x) = B_{(x + i) \mod{n}}$
%          \item For $\lceil \frac{d+1}{2} \rceil \leq i < \lceil \frac{d+1}{2} \rceil  + \lceil \frac{d-1}{2} \rceil$ we have $F_i(B_x) = B_{(x - (i-\lfloor \frac{d}{2} \rfloor )) \mod{n}}$
%          \item If $d$ is odd $F_d(B_x) = B_{(x + \lfloor \frac{n}{2} \rfloor ) \mod{n}}$
% \end{itemize} 
% With $0 \leq p \leq n-1$ and with $1 \leq q \leq d$, for each bucket $B_p$ and $B_q$, it holds that $B_q \in \mathcal{N}(B_p)$ if and only if $B_p  \in \mathcal{N}(B_q)$.
% \end{theorem}

\end{document}